\documentclass{imsart}

\RequirePackage[OT1]{fontenc}
\RequirePackage{amsthm,amsmath,amssymb,mathtools}
\RequirePackage{graphicx}
\RequirePackage[numbers]{natbib}
\RequirePackage[colorlinks,citecolor=blue,urlcolor=blue]{hyperref}
\usepackage[margin=1in]{geometry}
\usepackage{subfig}
\arxiv{math.PR/00000000}

\startlocaldefs
\numberwithin{equation}{section}
\theoremstyle{plain}
\newtheorem{thm}{Theorem}[section]

\newtheorem{cor}[thm]{Corollary}
\newtheorem{lem}[thm]{Lemma}
\newtheorem{prop}[thm]{Proposition}

\newtheorem{hyp}[thm]{Hypothesis}
\newtheorem{rmk}[thm]{Remark}

\newcommand{\be}{\begin{equation}}
\newcommand{\ee}{\end{equation}}
\newcommand{\benn}{\begin{equation*}}
\newcommand{\eenn}{\end{equation*}}

\newcommand{\bq}{\begin{eqnarray}}
\newcommand{\eq}{\end{eqnarray}}
\newcommand{\ex}{\mathbb{E}}
\newcommand{\half}{\frac{1}{2}}

\newcommand{\ind}{\mathbf{1}} 
\newcommand{\nn}{\nonumber}

\newcommand{\ul}{\underline}
\newcommand{\ol}{\overline}
\newcommand{\mc}{\mathcal}
\newcommand{\p}{\partial}
\newcommand{\e}{\varepsilon}
\newcommand{\et}{\mathrm{e}}

\newcommand{\pr}{\mathbb{P}}

\newcommand{\diff}{\textup{d}}
\newcommand{\prq}{{\Phi(r+q)}}

\newcommand{\R}{\mathbb{R}}
\newcommand{\timset}{\mathcal{T}}

\usepackage{color}

\hypersetup{pdfpagemode=UseNone}
\hypersetup{pdfstartview=FitH}

\renewcommand{\bar}{\overline}
\renewcommand{\tilde}{\widetilde}

\usepackage{verbatim} 

\begin{document}

\begin{frontmatter}
\title{\large BEATING THE OMEGA CLOCK: AN OPTIMAL STOPPING PROBLEM WITH RANDOM TIME-HORIZON UNDER SPECTRALLY NEGATIVE L\'EVY MODELS (EXTENDED VERSION)}
\runtitle{Beating the Omega clock}

\begin{aug}
\author{\fnms{Neofytos} \snm{Rodosthenous}\thanksref{t2} 
\ead[label=e1]{n.rodosthenous@qmul.ac.uk}}

\address{School of Mathematical Sciences, Queen Mary University of London\\ Mile End Road, London E1 4NS, UK\\
\printead{e1}}

\author{\fnms{Hongzhong} \snm{Zhang}\thanksref{t3}
\ead[label=e3]{hz2244@columbia.edu}
\ead[label=u1,url]{www.foo.com}}

\address{Department of IEOR, Columbia University\\
 New York, NY 10027, USA\\
\printead{e3}
}

\thankstext{t2}{The author gratefully acknowledges support from EPSRC Grant Number EP/P017193/1.}
\thankstext{t3}{Corresponding author.}
\runauthor{Rodosthenous and Zhang}

\affiliation{Queen Mary University of London and Columbia University}

\end{aug}

\begin{abstract}
We study the optimal  stopping of an American call option in a random time-horizon under exponential spectrally negative L\'evy models. The random time-horizon is modeled as the so-called Omega default clock in insurance, which is 
the first time when the occupation time of the underlying L\'evy process  
below a level $y$, exceeds an independent exponential random variable with mean $1/q>0$.  
We show that the shape of the value function varies qualitatively with different values of $q$ and 
$y$. In particular, we show that for certain values of $q$ and $y$, 
some quantitatively different but traditional up-crossing strategies are still optimal, 
while for other values we may have two disconnected 
continuation regions, resulting in the optimality of two-sided exit strategies.
By deriving the joint distribution of the discounting factor and the underlying process under a random discount rate, we give a complete characterization of all optimal exercising thresholds. 
Finally, we present an example with a compound Poisson process plus a drifted Brownian motion. 
\end{abstract}


\end{frontmatter}

\section{Introduction}

We consider a market with a risky asset whose price is modeled by  $\et^{X}$, where $X=(X_t)_{t\ge0}$ is a spectrally negative L\'evy process on a filtered probability space $(\Omega, \mathbb{F}=(\mathcal{F}_t)_{t\ge0},\pr)$. Here $\mathbb{F}$ is the augmented natural filtration of $X$. 
Fix a positive constant $q>0$, an Omega clock with rate $q$,  which measures the amount of time when $X$ is below a pre-specified level $y\in\mathbb{R}$, is defined as
\be\omega_t^y:=q\int_{(0,t]}\ind_{\{X_s<y\}}\diff s.\label{eq:omega}\ee
Let $\mathbf{e}_1$ be a unit mean exponential random variable which is independent of $X$, and denote by $\timset_0$ the set of all $\mathbb{F}$-stopping times.
We are interested in the following optimal stopping problem:
\[v(x;y):=\sup_{\tau\in \timset_0}\ex_x[\et^{-r\tau}(\et^{X_{\tau}}-K)^+\ind_{\{\omega_\tau^y<\mathbf{e}_1,\tau<\infty\}}],\]
where $r>0$ is a discount rate\footnote{The case $r\le0$ can be handled using the measure change technique as in \cite{TimKazuHZ14}.}, $K>0$, and $\ex_x$ is the expectation under $\pr_x$, which is the law of $X$ given that $X_0=x\in\R$. In other words, we look for the optimal exercising strategy for an American call option with strike $K$ and a random maturity given by the alarm of an Omega clock.

The study of random maturity American options was commenced by Carr \cite{Carr98}, where a Laplace transform method was introduced to finance to ``randomize'' the maturity, a technique known as Canadization. Different from us, a Canadized American option has a random maturity that is completely independent of the underlying asset. On the other extreme, an American barrier option is a random maturity American option with a maturity completely determined by the underlying asset,
see Trabelsi \cite{Trabelsi11} for an example. 
One of our motivations in this study is thus to build a general framework that concatenates the aforementioned special cases, through an Omega clock that only accumulates time for maturity when the underlying asset price is below the threshold $\et^y$.

The first use of occupation times in finance dates back to the celebrated work of Chesney, Jeanblanc and Yor \cite{CJYParis97}, who introduced and studied the so-called Parisian barrier option.  
Since then,
there has been a considerable amount of work on Parisian ruins for both L\'evy processes and reflected L\'evy processes at the maximum, see e.g. \cite{LLZ15, Loeffen2013}, among others.
Other related path-dependent options, whose payoffs reflect the time spent by the underlying asset price in certain ranges, were studied under Black-Scholes models by \cite{HR94, Miura92}.
The idea of using the cumulative occupation times  
originates from Carr and Linetsky's intensity based evaluation of executive stock options, or ESOs \cite{CarrLin00}.  
 Interestingly, the same idea has been applied later in insurance literature for studying the optimal dividend payment threshold in the presence of a so-called Gamma-Omega model \cite{AGS11}. 
 Subsequent research using this concept in insurance and applied probability literatures includes 
\cite{occupationLevy, OccupationInterval, ZhangAAP15}.

The problem addressed in this paper can be interpreted as the evaluation of an American variant of ESOs, where the risk is not on the resignation or early liquidation of the executives (as in the typical ESO setting), but is on the ``impatience'' Omega clock. 
This is an American option, which takes into account the cumulative amount of time that the underlying asset price is in a certain ``bad zone'' that reduces the holder's confidence on the underlying and hence shortens the statistical time horizon of the problem. 
The aim of this novel formulation is to capture quantitatively the accumulated impatience of decision makers in financial markets, when the latter do not move in their favor. Our mathematical analysis and derivation of optimal strategies can serve as an informative analytical tool for the commonly observed financial transactions that are affected by impatience 
(we refer to \cite{HaldaneImpatience} for an extensive report on the role of impatience in finance). 

Our problem can also be equivalently recast as an optimal stopping of a perpetual American call option with a random discount rate. Indeed,
by the independence between $X$ and $\mathbf{e}_1$, we have $\pr_x(\omega_\tau^y<\mathbf{e}_1, \tau<\infty|\mathcal{F}_\infty)=\exp(-\omega_\tau^y)\ind_{\{\tau<\infty\}}$, which implies that 
\be 
\label{eq:problem0}
v(x;y)=\sup_{\tau\in \timset_0}\ex_x[\et^{-A_\tau^y}(\et^{X_{\tau}}-K)^+\ind_{\{\tau<\infty\}}],
\ee
where $A_t^y$ is the occupation time
\be \label{A}
A_t^y:=rt+\omega_t^y,\quad\forall t\ge0.
\ee
A European-type equivalent to the option \eqref{eq:problem0} was considered by Linetsky \cite{Linetsky99} under the Black-Scholes model, that he named a Step call option.

A study of general perpetual optimal stopping problems with random discount rate 
is done in \cite{Dayanik2008b} (though the problem in \eqref{eq:problem0} was not considered there), by exploiting Dynkin's concave characterizations of the excessive functions for general linear diffusion processes. However, this approach has limitations when dealing with processes with jumps, due to the extra complication resulting from overshoots.

Part of our solution approach for the optimal stopping problem is inspired by \cite{Alili2005,mordecki2002} and the recent work \cite{TimKazuHZ14}. As noted in \cite{KazuEgamiAAP14,TimKazuHZ14}, while it may be standard to find necessary conditions for candidate optimal thresholds, it is far from being trivial to verify whether the associated value function satisfies the super-martingale condition, a key step in the verification. 
On the other hand, \cite{Alili2005,mordecki2002} solved the optimal stopping problems for the pricing of  perpetual American call and put options by directly constructing a candidate value function and
verifying a set of sufficient conditions for optimality,
using the Wiener-Hopf factorization of L\'evy processes and relying on neither continuous nor smooth fit conditions, reflecting the power of this approach. 
 Building on these ideas, \cite{Surya2007} reduced an optimal stopping problem to an averaging problem, which was later solved in \cite{TimKazuHZ14} by the equation that characterizes the candidate optimal thresholds.  
In this paper, we primarily focus on the case when the discounted asset price is a super-martingale and   
show that such a connection, as described above, still holds, thus generalizing the approach developed in \cite{Alili2005,TimKazuHZ14,mordecki2002,Surya2007}. 
In particular, 
using this result we prove 
the optimality of a certain up-crossing strategy for certain values of the parameters $q$ and $y$ of the Omega clock, by also checking its dominance over the intrinsic value function.  

However, under a random discount rate, up-crossing strategies may not be optimal for some set of model parameters. 
Intuitively, when the Omega clock has a large rate $q>0$, which results in a statistically shorter time-horizon if $X$ spends too much time below $y$,  then it might be optimal to stop if $X-y$ is too small.  
This leads to the consideration of two-sided exit strategies (at which point the above idea of averaging problem ceases to work). 
To be more precise, we prove that 
the optimal stopping region and the continuation region can have two disjoint components, 
which provides an interesting insight for the optimal exercising strategy, which can be either a profit-taking-type exit or a stop-loss-type exit (relative to the starting point), and  varies qualitatively with the starting price. 
Moreover, the continuation and stopping regions appear alternately.\footnote{\label{fn2}Therefore, in the case when the underlying process $X$ jumps from the upper continuation region over the lower stopping region, it will then be optimal to wait until $X$ increases to the upper boundary of that lower continuation region. So the optimal exercising strategy cannot be expressed as a one-shot scheme like a first passage or first exit time.} In order to establish the result, we use a strong approximation technique as in \cite{OccupationInterval} to explicitly derive the value function of a general two-sided exit strategy (see Proposition \ref{prop:V2value} for this new result), and then 
prove that, in cases of unbounded variation,  the smooth fit condition holds (regardless of $\sigma>0$ or not) for the value function at every boundary point of the optimal continuation region, echoing assertions in \cite{Alili2005, KazuEgamiAAP14}; in cases of bounded variation, the smooth fit condition still holds at all boundary points but one, where only continuous fit holds (see Proposition \ref{prop:sfit}).
As a consequence, in Corollary \ref{cor:fit} and Proposition \ref{thm:pair}, we obtain a novel qualitative characterization of the upper continuation region $(a^\star(y), b^\star(y))$.

Building upon the above results, we also study the case when the discounted asset price is a martingale and  
prove that for ``small'' $y$, 
the solution is trivial and identical to the perpetual American call option case, i.e. for $y=-\infty$, studied by Mordecki \cite{mordecki2002}. 
Surprisingly, for ``large'' $y$ 
not only the solution is not trivial, but takes significantly different forms, depending on the value of $X$. To be more precise, the stopping region is a closed finite interval below $y$, which results in a profit-taking exit at a lower value than the alternative stop-loss exit's value, while it is never optimal to stop if the asset price is above $\et^y$. Similar to the super-martingale case, the stop-loss exit strategy may even not be a one-shot scheme if an overshoot occurs$^{ \ref{fn2}}$. 
We also prove that when the discounted asset price is a sub-martingale, the solution is the same trivial one as in Mordecki's problem  for $y=-\infty$ \cite{mordecki2002}. 

The remaining paper is structured as follows. 
We begin with stating the main results of the paper, given by  Theorems \ref{thm:ult1}, \ref{thm:ult2} and \ref{thm:ult3}, in Section \ref{sec:res}, and devote the remaining sections to the development of their proofs. 
In particular, 
in Section \ref{sec:comp} we give some useful comparative statics for the value function. 
Then, Section \ref{sec:opt} is devoted to the study of a super-martingale discounted asset price and the proofs of Theorems \ref{thm:ult1} and \ref{thm:ult2}. More specifically, 
in Subsection \ref{sec:up} we investigate candidate up-crossing exercising thresholds and the equation satisfied by them. 
The optimality of these up-crossing and alternative two-sided exit strategies is established in Subsections \ref{sec:H>0}--\ref{sec:BV}. 
On the other hand, in Section \ref{sec:mart} we study the cases of a sub-martingale and a martingale discounted asset price and prove Theorem \ref{thm:ult3}.
Next, we consider a compound Poisson process plus a drifted Brownian motion in Section \ref{sec:gbm} and present numerical examples of a single and two disconnected components of optimal stopping region, illustrated in Figure \ref{fig1}. 
Finally, Section \ref{sec:con} provides the conclusion of the paper. 
The proofs of lemmas, omitted technical proofs 
and a collection of useful results on spectrally negative L\'evy processes and their scale functions, 
can be found in Appendices \ref{sec:prof} and \ref{sec:pre}.

\section{Main results}\label{sec:res}

In this section we begin by setting the scene and stating the main results, 
which we prove and present in greater detail in the subsequent Sections 
\ref{sec:comp}--\ref{sec:mart} and Appendix \ref{sec:prof}. 
We denote by $(\mu,\sigma^2,\Pi)$ the L\'evy triplet of the 
spectrally negative L\'{e}vy process $X$, and by $\psi$ its Laplace exponent. That is,
\begin{align}
\psi(\beta):=&\frac{1}{t}\log\mathbb{E}_0[\et^{\beta X_{t}}]=\mu \beta+\frac{1}{2}\sigma
^{2}\beta^{2}+\int_{(-\infty,0)}\left(  \et^{\beta x}-1-\beta x\ind_{\{x>-1\}}\right)
\Pi(\mathrm{d}x), \label{decomp}%
\end{align}
for every $\beta\in\mathbb{H}^{+}\equiv\{z\in\mathbb{C}: \Re z\ge0\}$. 
Here, $\sigma\geq0$, and the L\'{e}vy measure $\Pi(\mathrm{d}x)$ is supported 
on $(-\infty,0)$ with
$
\int_{(-\infty,0)}(1\wedge x^{2})\Pi(\mathrm{d}x)<\infty.
$
In case $X$ has paths of bounded variation, which happens  if and only if
$\int_{(-1,0)}|x|\Pi(\mathrm{d}x)<\infty$ and $\sigma=0$, we can
rewrite (\ref{decomp}) as
\begin{equation}
\psi(\beta)=\gamma\beta+\int_{(-\infty,0)}(\et^{\beta x}-1)\Pi(\mathrm{d}x),\quad \text{for } \; \beta\geq0 
\quad \text{and } \; \gamma:=\mu+\int_{(-1,0)}|x|\Pi(\mathrm{d}x)\,.\,
\label{psi bv}%
\end{equation}
where $\gamma>0$ holds. 
For any $r\geq0$, the equation $\psi(\beta)=r$ has at least one 
positive solution, and we denote the largest one by $\Phi(r)$. 
Then, the $r$-scale function $W^{(r)}:\mathbb{R}\mapsto\lbrack0,\infty)$ is 
the unique function supported on $[0,\infty)$ with Laplace transform
\be
\int_{[0,\infty)}\et^{-\beta x}W^{(r)}(x)\mathrm{d}x=\frac{1}{\psi(\beta)-r},\quad \beta>\Phi(r). \label{laplace}
\ee
We extend $W^{(r)}$ to the whole real line by setting $W^{(r)}(x)=0$ for $x<0$.
Henceforth we assume that 
the jump measure $\Pi(\mathrm{d}x)$ has no atom,
thus $W^{(r)}(\cdot)\in C^{1}(0,\infty)$ (see e.g.  
\cite[Lemma 2.4]{Kuznetsov_2011}); lastly, we also assume that $W^{(r)}(\cdot)\in C^{2}(0,\infty)$ for all $r\geq0$, which is guaranteed if $\sigma>0$ (see e.g.  
\cite[Theorem 3.10]{Kuznetsov_2011}).\footnote{However, $\sigma>0$ is not a necessary condition for $W^{(r)}(\cdot)\in C^2(0,\infty)$. For instance, a spectrally negative $\alpha$-stable process with $\alpha\in(1,2)$  satisfies this condition without a Gaussian component.} The $r$-scale function is closely related to the first passage times of $X$, which are defined as
\be \label{firstpass}
T_{x}^{\pm}:=\inf\{  t\geq0:X_{t} \gtrless x\},\quad  x\in%
\mathbb{R}.
\ee
Please refer to Appendix \ref{sec:pre} or
\cite[Chapter 8]{Kyprianou2006} for some useful results on this matter.

The infinitesimal generator $\mc{L}$ of $X$, 
which is well-defined at least for all functions ${\tilde F}(\cdot)\in C^2(\R)$, is given by
\[
\mc{L}{\tilde F}(x)=\frac{1}{2}\sigma^2{\tilde F}''(x)+\mu {\tilde F}'(x)+\int_{(-\infty,0)}({\tilde F}(x+z)-{\tilde F}(x)-\ind_{\{z>-1\}}z{\tilde F}'(x))\Pi(\diff z).
\]
In particular, $\mc{L}\et^{\beta x}=\psi(\beta)\,\et^{\beta x}$ for any $\beta\ge0$. 

Fix a $y\in[-\infty,\infty]$, let us introduce the ``stopping region'' $\mc{S}^y$, the set where the so-called time value vanishes:
\be
\mc{S}^y:=\{x\in\R: v(x;y)-(\et^x-K)^+=0\}. \label{S}
\ee 
In other words, the optimal exercise strategy for the problem \eqref{eq:problem0} is given by the stopping time 
\[
\tau^y_\star:=\inf\{t\geq 0\;:\; X_t\in\mc{S}^y\} .
\]
Note that, the special cases when $y=\infty$ or $y=-\infty$, namely the perpetual American call options with discount rates $r+q$ or $r$, respectively, have been studied in \cite{mordecki2002}. 
In what follows, we split our analysis in three parts, in which we study the cases of the discounted asset price being a super-martingale, a martingale and a sub-martingale. 

We begin by assuming that the discounted asset price process $(\et^{-rt+X_t})_{t\ge0}$ 
is a (strict) $\pr_x$-super-martingale, which is equivalent to the model parameters satisfying $r>\psi(1)$. 
In this case, it is well-known (see, for example, \cite{mordecki2002}) that 
\begin{align}
\underline{v}(x)
:=&\sup_{\tau\in\timset_0}\ex_x[\et^{-(r+q)\tau}(\et^{X_{\tau}}-K)^+\ind_{\{\tau<\infty\}}]
=\ind_{\{x\le\underline{k}\}}\et^{\Phi(r+q)(x-\underline{k})}(\underline{K}-K)+\ind_{\{x>\underline{k}\}}(\et^x-K),
\label{eq:underlinev}\\
\overline{v}(x)
:=&\sup_{\tau\in\timset_0}\ex_x[\et^{-r\tau}(\et^{X_{\tau}}-K)^+\ind_{\{\tau<\infty\}}]
=\ind_{\{x\le\overline{k}\}}\et^{\Phi(r)(x-\overline{k})}(\overline{K}-K)+\ind_{\{x>\overline{k}\}}(\et^x-K),
 \label{eq:overlinev}
\end{align}
and the optimal stopping regions take the form 
\be
\mc{S}^{\infty} = [\underline{k},\infty) \quad \text{and} \quad 
\mc{S}^{-\infty} = [\ol{k},\infty) , \label{stopinf}
\ee
where the exercising thresholds are given by
\be
\underline{k}=\log\underline{K},\quad \text{for} \;\;\underline{K}=\frac{\Phi(r+q)}{\Phi(r+q)-1}K, 
\quad \text{and} \quad
\overline{k}=\log\overline{K},\quad \text{for} \;\; \overline{K}=\frac{\Phi(r)}{\Phi(r)-1}K.  
\label{k}
\ee 
Notice that the fractions in (\ref{k}) are well-defined, because $\Phi(r+q)>\Phi(r)>1$, where the latter inequality follows from the standing assumption that $r>\psi(1)$.

When $y\in(-\infty,\infty)$, which is the subject of this work, the discount rate changes between $r$ and $r+q$. Hence, the value function  $v(x;y)$ should be bounded from below by $\underline{v}(x)$ and from above by $\ol{v}(x)$ (see Proposition \ref{prop:compare1}). As a consequence, 
 we know that the optimal stopping region $\mc{S}^y\subset\mc{S}^{\infty}=[\underline{k},\infty)$, and we can equivalently express the problem in \eqref{eq:problem0} as 
\be 
v(x;y)=\sup_{\tau\in \timset_0}\ex_x[\et^{-A_\tau^y}\underline{v}(X_{\tau})\ind_{\{\tau<\infty\}}]. \label{eq:problem00}
\ee
Notice that $\underline{v}(\cdot)\in C^1(\R)\cap C^2(\R\backslash\{\underline{k}\})$.

To state our result in the most general situation, it will be convenient to talk about the following hypothesis as a key insight/conjecture for our problem. 
\begin{hyp} \label{1stop}
Let $O$ be any connected component of the continuation region $(\mc{S}^y)^c$ of problem \eqref{eq:problem0},\footnote{Independent of this hypothesis, in Proposition \ref{prop:compare1}$(ii)$ we prove that $v(\cdot;y)$ is continuous, so the optimal stopping set $\mc{S}^y$ is a closed set and the continuation set $(\mc{S}^y)^c$ is an open set, i.e. a union of disjoint open intervals.} then
$O\cap\{x\in\R\backslash\{\underline{k}\}: (\mc{L}-r-q\ind_{\{x<y\}})\underline{v}(x)\ge 0\}\neq\emptyset$.
\end{hyp}

The following equivalence relation is crucial for deriving the solution to  problem \eqref{eq:problem0}. It is proved in the Appendix \ref{sec:prof} using results from Section \ref{sec:comp}.
\begin{lem}\label{lem:equivalent}
Hypothesis \ref{1stop} is equivalent to the assertion that there is at most one component of the stopping region that lies on the right hand side of $y$.
\end{lem}
 
\begin{rmk} \label{rmk:1stop}\normalfont 
Hypothesis \ref{1stop} is a natural conjecture for the obvious fact within a diffusion framework, in which (by Dynkin's formula) the (positive) time value at any point in the continuation region is the expectation of a time integral of $\et^{-A_t^y} (\mc{L}-r-q\ind_{\{x<y\}})\underline{v}(X_t)$ until entering the stopping region. 
For a general L\'evy process, even though jumps will nullify this argument, 
it is still peculiar if there is
a part of continuation region $O$ that falls completely inside the super-harmonic set of $\underline{v}(\cdot)$: $\{x\in\R\backslash\{\underline{k}\}: (\mc{L}-r-q\ind_{\{x<y\}})\underline{v}(x)< 0\}$, because it would imply (by It\^o-L\'evy's formula) that, it is beneficial to continue when the underlying process is in $O$, despite the fact that the expected gain in time value is negative at any instance before entering into the stopping region. 
A complete characterization  of the applicability of Hypothesis \ref{1stop} and alike for general L\'evy models is beyond the scope of this work\footnote{It is however possible to give a conclusive answer if we replace the indicator $\ind_{\{x<y\}}$ with $\ind_{\{x\ge y\}}$. However, by doing so we lose the interesting trade-off and consequently the rich solution structures in our problem.}. 
To give a concrete idea on the applicability of this conjecture, we show in Proposition \ref{prop:1stop} that a monotone density of the L\'evy measure $\Pi(\diff x)$ implies Hypothesis \ref{1stop}.  
\end{rmk}

We now present the main results of this work.
We first treat the case where Hypothesis \ref{1stop} is not needed as a sufficient condition, but on the contrary follows as a conclusion from the results.

Let us define a positive function that will be useful afterwards: 
\be
\mc{I}^{(r,q)}(x):=\int_{(0,\infty)}\et^{-\prq u}W^{(r)}(u+x)\diff u,\quad\forall x\in\R,
\label{eq:Zdef}
\ee
which is easily seen to be continuously twice differentiable over $(0,\infty)$, and satisfies $\mc{I}^{(r,q)}(x)=\et^{\prq x}/q$ for all $x\le 0$. 
\begin{thm} \label{thm:ult1}
Suppose that the model parameters satisfy $r>\psi(1)$, $X$ has paths of unbounded variation and 
\be \label{cond1}
\mc{I}^{(r,q),\prime\prime}(x)\geq \mc{I}^{(r,q),\prime}(x) ,
\quad\;\forall x>0.
\ee
Then the optimal stopping region of the problem \eqref{eq:problem0} is given by $\mc{S}^y=[z^\star(y),\infty)$, at whose boundary the smooth fit condition holds. In particular, 
\begin{enumerate}
\item[(i)] if $y\in(-\infty,\underline{k})$, then the optimal threshold $z^*(y)\in(\underline{k},\ol{k})$ is defined in \eqref{z*}, and the value function is given by 
\begin{align} \label{v1}
v(x;y)= \ind_{\{x<z^\star(y)\}}(\et^{z^\star(y)}-K)^+\frac{\mc{I}^{(r,q)}(x-y)}{\mc{I}^{(r,q)}(z^\star(y)-y)}+\ind_{\{x\ge z^\star(y)\}}(\et^x-K)\,;
\end{align}
\item[(ii)] if $y\in[\underline{k},\infty)$, then 
$z^*(y)=\underline{k}$ and 
$v(x;y)=\underline{v}(x)$.
\end{enumerate}
\end{thm}

Thus, when $X$ has paths of unbounded variation and \eqref{cond1} holds, the  traditional up-crossing threshold-type exercising strategy is still optimal and the optimal stopping region is a connected set $\mc{S}^{y} = [z^\star(y),\infty)\subset [y,\infty)$, for some exercise threshold $z^\star(y)>\underline{k}$, when $y<\underline{k}$. 
On the other hand, the optimal stopping region is the connected set $\mc{S}^{y} = [\underline{k},\infty) \equiv [\underline{k},y)\cup[y,\infty)$, when $y\geq \underline{k}$, which is identical to the standard problem for $y=\infty$ presented above. 
In view of Lemma \ref{lem:equivalent}, Hypothesis \ref{1stop} holds true in both cases $(i)$--(ii).

In the following theorem, we consider all remaining cases not covered by Theorem \ref{thm:ult1}.

\begin{thm} \label{thm:ult2}
Suppose that the model parameters satisfy $r>\psi(1)$ and 
either $X$ has paths of unbounded variation and \eqref{cond1} fails, 
or $X$ has paths of bounded variation.
Then there exist constants $\tilde{y}, y_m$ defined in \eqref{ytilde} and \eqref{ym} satisfying $\tilde{y}<y_m<\ol{k}$, such that, 
\begin{enumerate}
\item[(i)] if $y\in(-\infty,\tilde{y})$, then $\mc{S}^y=[z^\star(y),\infty)$, where the optimal threshold $z^*(y)\in(\underline{k},\ol{k})$ is defined in \eqref{z*}, and the value function is given by \eqref{v1};
\item[(ii)] if $y\in [y_m,\infty)$, then $\mc{S}^y=[\underline{k},\infty)$ and $v(x;y)=\underline{v}(x)$ is given by \eqref{eq:underlinev};
\item[(iii)] if $y=\tilde{y}$ (and Hypothesis \ref{1stop} holds, when \eqref{cond1} fails), then  $\tilde{y}>\underline{k}$ and $\mc{S}^{\tilde{y}}=\{\underline{k}\}\cup[z^\star(\tilde{y}),\infty)$, where 
$z^*(\tilde{y})\in(\underline{k},\ol{k})$ is defined in \eqref{z*}, and the value function is still given by \eqref{v1};
\item[(iv)] if $y\in(\tilde{y},y_m)$ and Hypothesis  \ref{1stop} holds, then $\tilde{y}>\underline{k}$ and $\mc{S}^y=[\underline{k},a^\star(y)]\cup[b^\star(y),\infty)$, where the optimal thresholds satisfy $\underline{k} < a^\star(y) < y < y_m < b^\star(y) < z^*(\tilde{y})$ and are given by 
\eqref{a*b*}, while the value function takes the form
\begin{align} \label{vflast}
v(x;y) = \underline{v}(x) + \ind_{\{x\in(a^\star(y),b^\star(y))\}} 
\Delta(x,a^\star(y);y)\,,
\end{align}
with $\underline{v}(x)$ given by \eqref{eq:underlinev} and $\Delta(x,a;y)$ being the positive function defined in \eqref{D}.
\end{enumerate}
Moreover, the smooth fit condition always holds at all boundaries of $\mc{S}^y$ in parts $(i)$, $(ii)$, $(iii)$ and at $\underline{k}$ and $b^\star(y)$ of part $(iv)$. Furthermore, in part $(iv)$, the smooth (continuous, resp.) fit condition holds at the boundary $a^\star(y)$  when $X$ has paths of unbounded  (bounded, resp.) variation.
\end{thm}

In the cases studied in Theorem \ref{thm:ult2}, the level of $y$ plays an important role in the structure of the optimal stopping region. Specifically,
for ``small" values of $y$, the traditional up-crossing threshold-type exercising strategy is still optimal and the optimal stopping region is one connected component $\mc{S}^{y} = [z^\star(y),\infty)\subset (y,\infty)$, for some exercise threshold $z^\star(y)>y$. 
Also for ``large" values of $y$, the optimal strategy is identical to the case $y=\infty$ and the optimal stopping region is the traditional connected set $\mc{S}^{y} = [\underline{k},\infty)\equiv [\underline{k},y)\cup[y,\infty)$, 
as in Theorem \ref{thm:ult1}. 
However, for some ``intermediate" values of $y$, the traditional threshold-type strategy is no longer optimal. Instead, there are exactly two components of stopping region, one inside $[\underline{k},y)$ and another in $(y,\infty)$. 
Therefore, in all cases, except 
for the ones it is used as a condition, we can apply Lemma \ref{lem:equivalent} to conclude that Hypothesis \ref{1stop} holds true.

It is also observed from part $(iv)$ that, when $x\in(a^\star(y),b^\star(y))$, 
it is optimal to follow a non-traditional two-sided exit strategy. 
Moreover, in the event of an overshoot of $X$ from the set $(a^\star(y),b^\star(y))$ to $(-\infty,\underline{k})$, it is not optimal to stop immediately but wait until $X$ increases to $\underline{k}$.
This means that, in contrast to most optimal stopping problems in L\'evy models literature with two-sided exit strategies (see e.g. \cite{ChristensenIrle13}), the optimal stopping time for our problem may not be a one-shot scheme like a first passage or first exit time. 
The target exercising threshold has to be re-adjusted if an overshoot occurs. 

It is seen that condition \eqref{cond1} plays a pivotal role in distinguishing Theorem \ref{thm:ult1} from Theorem \ref{thm:ult2} when $X$ is of unbounded variation, and deciding whether Hypothesis \ref{1stop} is needed as a sufficient condition in Theorem \ref{thm:ult2}$(iii)$ when $X$ is of bounded variation. 
In order to facilitate the verification of whether \eqref{cond1} holds or not, we will later provide in Remark \ref{rmk:u} convenient equivalences to  condition \eqref{cond1}, based on the sign of the quantity $\bar{u}$ defined by \eqref{u-} (see also Lemma \ref{lem:Hsign2q} for the relation of $\bar{u}$ with the problem's parameters). 
If we limit ourselves only to special classes of L\'evy jump measures, we have the following criterion.

\begin{lem}\label{lem:explicit}
Suppose that the scale function $W^{(r)}(\cdot)\in C^2((0,\infty))$ and  the tail jump measure of $X$, denoted by $\overline{\Pi}(x):=\Pi(-\infty,-x)$ for $x>0$, either has a completely monotone density or is  log-convex, then \eqref{cond1} holds if and only if 
\be
\label{eq:explicit}
(\prq-1)(\prq-qW^{(r)}(0))-qW^{(r)\prime}(0+)\ge0.
\ee
\begin{proof}
See Appendix \ref{sec:prof}.
\end{proof}
\end{lem}

Finally, we close this section by considering the cases of discounted asset price process $(\et^{-rt+X_t})_{t\ge0}$ being either a (strict) $\pr_x$-sub-martingale or a $\pr_x$-martingale, which are equivalent to the model parameters satisfying $r<\psi(1)$ or $r=\psi(1)$, respectively.

\begin{thm} \label{thm:ult3}
If the model parameters satisfy:
\begin{enumerate}
\item[(a)] $r<\psi(1)$, then it is never optimal to stop the process $X$, namely $\mc{S}^y=\emptyset$, and the value function of the problem \eqref{eq:problem0} is $v(x;y)=\infty$ for all $x,y\in\R$. 
\item[(b)] $r=\psi(1)$, then there exists a unique value $y_\infty\in(0,\infty)$ given by \eqref{y_inf}, such that:
\begin{enumerate}
\item[(i)] if $y\in(-\infty, y_\infty)$, then it is never optimal to stop the process $X$, namely $\mc{S}^y=\emptyset$, however the value function $v(x;y)$ of the problem \eqref{eq:problem0} has a value and is given by 
\be \label{vinf}
v(x;y) = \frac{(\Phi(r+q)-1)}{\Phi'(r)} \,\et^{y}\,\mc{I}^{(r,q)}(x-y);
\ee
\item[(ii)] if $y=y_\infty$, then $\mc{S}^{y_\infty}=\{\underline{k}\}$, where $\underline{k}$ is given by \eqref{k} and $v(x;y_\infty)$ by \eqref{vinf}; 
\item[(iii)] if $y\in(y_\infty,\infty)$, then $\mc{S}^y=[\underline{k},a_\infty^\star(y)]$, where the optimal threshold $a_\infty^\star(y)$ is the unique solution of equation \eqref{eq:a_infty} and the value function is given by
\be \label{vfainf}
v(x;y)=\underline{v}(x)+\ind_{\{x\in(a_\infty^\star(y),\infty)\}}\Delta(x, a_\infty^\star(y); y),
\ee 
with $\Delta(x,a;y)$ being the positive function defined in \eqref{D}.
\end{enumerate}
\end{enumerate}
\end{thm}

Recall from \cite{mordecki2002}, that the solution to the perpetual American call option ($y=-\infty$) is trivial in both cases covered in Theorem \ref{thm:ult3}. 
Namely, in case (a) we have $\mc{S}^{-\infty}=\emptyset$ and $v(x;-\infty)=\infty$, while in case (b) we have $\mc{S}^{-\infty}=\emptyset$ and $v(x;-\infty)=\et^x$.
In this paper, we demonstrate that the problem in part (a) remains identically trivial. 
Interestingly, the solution to the problem studied in part (b) is on the contrary non-trivial. In fact, the optimal exercise strategy admits a surprising structure, which can be either an up-crossing or a down-crossing one-sided exit strategy, depending on the starting value $x$, for $y\in[y_\infty,\infty)$. Furthermore, the strategy may not be a one-shot scheme if an overshoot occurs, making the result even more fascinating.  

We devote the following sections to proving the aforementioned results and providing an illustrating example.

\section{Comparative statics of the value function $v(\cdot;y)$}\label{sec:comp}

In this section, we present some useful stylized facts for the value function $v(x;y)$ by comparative statics. 
 In view of (\ref{A}) with (\ref{eq:omega}), we can see that $rt\le A_t^y\le (r+q)t$ holds for all $t\ge0$ and $y\in\R$. 
Using these inequalities together with \eqref{eq:problem0}, \eqref{eq:underlinev} and \eqref{eq:overlinev}, we obtain the following results for the monotonicity, continuity of the value function $v(x;y)$, and some information about the structure of the stopping region $\mc{S}^y$. 
\begin{prop}\label{prop:compare1}
The value function $v(x;y)$ satisfies the following properties:
\begin{enumerate}
\item[(i)]
For any $-\infty\le y_1<y_2\le\infty$, we have 
\benn
 \underline{v}(x)\equiv 
v(x;\infty)\le v(x;y_2)\le v(x;y_1)\le v(x;-\infty)
 \equiv\overline{v}(x)
,\quad\forall x\in\R .
\eenn
Moreover, it holds that $\mc{S}^{\infty}\supseteq \mc{S}^{y_2} \supseteq \mc{S}^{y_1} \supseteq \mc{S}^{-\infty}$.
\item[(ii)] The function $v(x;y)$ is strictly increasing  and continuous in $x$ over $\R$, and is non-increasing and continuous in $y$ over $\R$.
\item[(iii)] If there is an $a\in\mc{S}^y$ such that $a\le y$, we must have 
 $y\ge\underline{k}$ and $[\underline{k},a]\subset\mc{S}^y$. 
\end{enumerate}
\end{prop}
\begin{proof}
See Appendix \ref{sec:prof}.
\end{proof}

\begin{rmk}\label{rmk:use1piece}  \normalfont
 Proposition \ref{prop:compare1}$(i)$ concludes that the stopping region $\mc{S}^y$ is sandwiched by two known intervals, namely the optimal stopping regions for the problems \eqref{eq:underlinev}--\eqref{eq:overlinev}, but it is still highly non-trivial to determine the exact shape of $\mc{S}^y$. 
Proposition \ref{prop:compare1}$(ii)$ implies  
that the optimal stopping region $\mc{S}^y$ consists of disjoint unions of closed intervals (including isolated points), and 
 in view of conclusion $(i)$, these intervals {\em continuously ``grow''} with $y$, unless new components of stopping region appear (we call this ``branching"). 
Finally, 
Proposition \ref{prop:compare1}$(iii)$ indicates the special role of $\underline{k}$ and the possibility for a unique stopping region component that lies below $y$, which always takes the form $[\underline{k},a]$ for some $a\in[\underline{k},y]$. 
\end{rmk}

\section{Proofs of Theorem \ref{thm:ult1} and Theorem \ref{thm:ult2}} \label{sec:opt}

In this section we study the case when the discounted asset price process $(\et^{-rt+X_t})_{t\ge0}$ is a super-martingale and we eventually prove Theorems \ref{thm:ult1} and \ref{thm:ult2}. 
Note that, for $r>\psi(1)$, \eqref{stopinf} and  Proposition \ref{prop:compare1}$(i)$ imply that the optimal stopping region $\mc{S}^y$ should always contain the region $\mc{S}^{-\infty}\equiv[\overline{k},\infty)$. Thus, in view of the equivalent expressions of the problem in \eqref{eq:problem0} and \eqref{eq:problem00}, we can further reduce the problem to 
\be
v(x;y)
=\sup_{\tau\in \timset}\ex_x\big[\et^{-A_\tau^y}(\et^{X_{\tau}}-K)^+\ind_{\{\tau<\infty\}}\big]
\equiv \sup_{\tau\in\mc{T}}\ex_x\big[\et^{-A_\tau^y}\underline{v}(X_\tau)\ind_{\{\tau<\infty\}}\big]\,, 
\label{eq:problem} 
\ee
where $\timset$ is the set of all $\mathbb{F}$-stopping times which occur no later than $T_{\overline{k}}^+$. 
We shall focus on the problem \eqref{eq:problem} henceforth in this section. 

\subsection{The exercising thresholds for up-crossing strategies} \label{sec:up}
We begin our analysis by studying the expected discount factor up until a first passage time $T_z^+$,
which can be proved similarly to  \cite[Corollary 2(ii)]{OccupationInterval}.
\begin{prop}\label{prop:lap11}
Recall the positive function $\mc{I}^{(r,q)}(\cdot)$ defined by \eqref{eq:Zdef}, which takes the form $\mc{I}^{(r,q)}(x)=\et^{\prq x}/q$, for all $x\le 0$.
We have
\benn
\ex_x\big[\exp(-A_{T_z^+}^y)\big]=\frac{\mc{I}^{(r,q)}(x-y)}{\mc{I}^{(r,q)}(z-y)},
\quad\forall z>x .
\eenn
\end{prop}

It is possible to reinterpret the result in Proposition \ref{prop:lap11} as the upper tail probability of the running maximum of $X$ at a random time. Indeed, let us consider the left inverse of the additive functional $A^y$ at an independent exponential time with unit mean, $\mathbf{e}_1$:
\be
\zeta\equiv (A^y)^{-1}(\mathbf{e}_1):=\inf\{t>0: A_t^y>\mathbf{e}_1\}.\label{eqzeta}
\ee
Then by the independence, it is seen that, for any $z>x$,
\bq
\pr_x( \overline{X}_{\zeta}>z)=\pr_x(T_z^+<\zeta)=\ex_x\big[\exp(-A_{T_z^+}^y)\big]=\ex_x\big[\exp(-A_{T_z^+}^y)\ind_{\{T_z^+<\infty\}}\big], \label{repr}
\eq
where the last equality is due to $A_t^y\ge rt$, which converges to $\infty$ as $t\uparrow\infty$. 
Let us introduce the ``hazard rate'' function \footnote{It can be easily verified that the right hand side of \eqref{eq:hazard} is indeed independent of $x$.} 
\be
\Lambda(z;y)=\frac{1}{\pr_x(\overline{X}_\zeta>z)}\frac{\pr_x(\overline{X}_\zeta\in\diff z)}{\diff z},\quad\forall z> x.\label{eq:hazard}
\ee
Then we have
\be \label{Lambda2}
\Lambda(z;y)\equiv \overline{\Lambda}(z-y)=\frac{\mc{I}^{(r,q),\prime}(z-y)}{\mc{I}^{(r,q)}(z-y)}=\prq-\frac{W^{(r)}(z-y)}{\mc{I}^{(r,q)}(z-y)},\quad\forall z> x.
\ee
Since $W^{(r)}(x)=0$ for all $x<0$, we see that the function $\overline{\Lambda}(x)\equiv\prq$ for all $x<0$. 
In the following lemma, we present some properties of $\overline{\Lambda}(x)$ for $x \geq 0$, 
which can be proved using Lemma \ref{lem W} and a calculation of the derivative of $\overline{\Lambda}(\cdot)$ given by \eqref{Lambda2}. 
\begin{lem}\label{Lambdaprop}
The function $\overline{\Lambda}(\cdot)$ given by \eqref{Lambda2} is strictly decreasing over $[0,\infty)$, with  
\be
\overline{\Lambda}(0)=\prq-qW^{(r)}(0)\le\prq=\overline{\Lambda}(0-) \quad \text{ and } \quad \overline{\Lambda}(\infty)=\Phi(r).\nn
\ee
In other words, $\mc{I}^{(r,q)}(\cdot)$ is log-concave over $\R$.
\end{lem}
\begin{proof}
See Appendix \ref{sec:prof}.
\end{proof}

In what follows, we denote the value of an up-crossing strategy $T_z^+$ by $U(\cdot;y,z)$, which will be the main topic of study in the remainder of this subsection, and is given by 
\begin{align} \label{U}
U(x;y,z)=&\ex_x\big[\exp(-A_{T_z^+}^y)(\exp(X_{T_z^+})-K)^+\ind_{\{T_z^+<\infty\}}\big]\nn\\
=&\ind_{\{x<z\}}(\et^z-K)^+\ex_x\big[\exp(-A_{T_z^+}^y)\big] + \ind_{\{x\ge z\}}(\et^x-K)^+\nn\\
=&\ind_{\{x<z\}}(\et^z-K)^+\frac{\mc{I}^{(r,q)}(x-y)}{\mc{I}^{(r,q)}(z-y)}+\ind_{\{x\ge z\}}(\et^x-K) ,
\end{align}
where the last equality follows from Proposition \ref{prop:lap11}. 
Fixing an arbitrary $x<\underline{k}$, (which is definitely not inside $\mc{S}^y$ by Proposition \ref{prop:compare1}$(i)$), we look for candidate exercising thresholds greater than $x$. 
By taking the derivative of $U(x;y,z)$ with respect to $z$ for $z>x\vee\log K$ and using (\ref{repr})-(\ref{eq:hazard}), we get
\begin{align}
\p_zU(x;y,z)=&\et^z\ex_x\big[\exp(-A_{T_z^+}^y)\big] + (\et^z-K)\frac{\p}{\p z}\ex_x\big[\exp(-A_{T_z^+}^y)\big]\nn\\
=&\ex_x\big[\exp(-A_{T_z^+}^y)\big]\overline{\Lambda}(z-y)\bigg(K-\et^z\bigg(1-\frac{1}{\overline{\Lambda}(z-y)}\bigg)\bigg)=\frac{\pr_x(\overline{X}_\zeta\in\diff z)}{\diff z}\bigg(K-\et^z\bigg(1-\frac{1}{\overline{\Lambda}(z-y)}\bigg)\bigg).\nn
\end{align}
Hence, a candidate optimal exercising threshold $z^\star$ should satisfy
\be
K=\et^yg(z^\star-y), \quad\text{where } \; 
g(u)=\et^u\bigg(1-\frac{1}{\overline{\Lambda}(u)}\bigg)
=\et^u\bigg(1-\frac{\mc{I}^{(r,q)}(u)}{\mc{I}^{(r,q),\prime}(u)}\bigg).
\label{eq:gK}
\ee
Notice that, since $\overline{\Lambda}(\cdot)$ is monotone by Lemma \ref{Lambdaprop}, the function $g$ defined in (\ref{eq:gK}) satisfies 
\begin{alignat*}{2}
\et^yg(z-y)&<\et^{\underline{k}}\bigg(1-\frac{1}{\prq}\bigg)=\underline{K}\frac{\prq-1}{\prq}=K,\quad&&\forall z<\underline{k},\nn\\
\et^yg(z-y)&>\et^{\overline{k}}\bigg(1-\frac{1}{\Phi(r)}\bigg)=\overline{K}\frac{\Phi(r)-1}{\Phi(r)}=K,\quad&&\forall z\ge\overline{k},\nn
\end{alignat*}
where the strict inequality 
in the second line is due to the fact that $\overline{\Lambda}(x)>\Phi(r)$ for all $x\in\R$. 
Furthermore, $g(\cdot)$ is continuous over $\R$, 
unless the process $X$ has paths of bounded variation, which gives rise to a negative jump at $0$ (see Lemma \ref{lem W}):
\be
g(0)-g(0-)=\frac{-qW^{(r)}(0)}{\prq(\prq-qW^{(r)}(0))}\,\et^y\le 0. 
\label{gjump}
\ee
\begin{rmk}\label{candidate} \normalfont
It thus follows that there exists at least one candidate optimal exercising threshold $z^\star$ in $[\underline{k},\overline{k})$, and there is no optimal exercising threshold in $\R\backslash[\underline{k},\overline{k})$. This is consistent with Proposition \ref{prop:compare1}$(i)$. 
\end{rmk}

Although Remark \ref{candidate} confirms the existence of at least one candidate exercising threshold, there is a possibility of multiple solutions to \eqref{eq:gK}. We investigate this possibility through the analysis of the derivative of $g(u)$:
\be \label{g'}
g'(u)=\begin{dcases}
\et^u\frac{\prq-1}{\prq} >  0,\quad&\forall u<0,\\
\frac{\et^u}{(\ol{\Lambda}(u))^2}\mc{H}(u),\quad&\forall u>0
\end{dcases}
\ee
where 
\be\label{eqHv}
\mc{H}(u):=\frac{1}{\mc{I}^{(r,q)}(u)}\left(\mc{I}^{(r,q),\prime\prime}(u)-\mc{I}^{(r,q),\prime}(u)\right)\equiv(\Phi(r+q)-1)\overline{\Lambda}(u)-\frac{W^{(r)\prime}(u)}{\mc{I}^{(r,q)}(u)},\quad\forall u>0.
\ee
Observe that by the definition of $\bar{\Lambda}(x)$ and its limit as $x\to\infty$ in Lemma \ref{Lambdaprop}, we have 
\be
\mc{H}(\infty)=\Phi(r)(\Phi(r)-1)>0,\label{eq:Hinfnnn}
\ee
which implies that $g(\cdot)$ is ultimately strictly increasing to $\infty$. 
In view of this observation, we define
\be 
\label{u-} 
\bar{u}:=\inf\{u\in\R: g(\cdot) \text{ is non-decreasing over }[u,\infty)\},
\ee
thus $\bar{u}$ is the largest local minimum of $g$ (and is well-defined). 
In all, \eqref{g'} and \eqref{u-} imply that $g(\cdot)$ is strictly increasing over $(-\infty,0)$ and is non-decreasing over $(\bar{u},\infty)$.

The value of $\bar{u}$ will be critically important in distinguishing the different possibilities of solutions to problem \eqref{eq:problem}.
In particular, it is seen from the above analysis that there exist only three possible cases for the value of $\bar{u}$, outlined in the following remark.
\begin{rmk} \label{rmk:u} \normalfont 
We have the following equivalences:
\begin{enumerate}
\item $\bar{u}=-\infty$, so that $g(\cdot)$ is non-decreasing over $\R$, which is equivalent to condition \eqref{cond1}
and $X$ having paths of unbounded variation 
(case treated in Theorem \ref{thm:ult1});
\item $\bar{u}=0$, so that $g(\cdot)$ is non-decreasing over $\R\backslash\{0\}$ with a negative jump at $0$, which is again equivalent to condition \eqref{cond1} 
but $X$ having paths of bounded variation 
(case treated in Theorem \ref{thm:ult2});
\item $\bar{u}>0$, so that $g(\cdot)$ is not monotone over $(0,\infty)$, which is equivalent to condition \eqref{cond1} failing 
(case treated in Theorem \ref{thm:ult2}). 
\end{enumerate}
\end{rmk}

In light of Remark \ref{rmk:u} and the assertions of Theorem \ref{thm:ult2}, the case $\bar{u}\geq 0$ will be treated in a unified way in all parts of that theorem, apart from part $(iii)$, where we need to treat cases $\bar{u}= 0$ and $\bar{u}> 0$ differently. 
In order to further illustrate the dependence of the value/sign of $\bar{u}$ on $q$, $r$ and the Laplace exponent of the L\'evy process $X$, we use the definitions of $\psi(\cdot)$ in \eqref{decomp} and $\prq$, as well as the value of $\mc{H}(0+)$ given by \eqref{eqHv}, to prove the following lemma.
A combination of the latter with Remark \ref{rmk:u} also sheds light on the conditions of Theorems \ref{thm:ult1} and \ref{thm:ult2}. 
\begin{lem}\label{lem:Hsign2q}
If $\psi(1)>0$ and $\sigma>0$, then for all $q>(\psi(r/\psi(1))-r)\vee 0$, we have $\bar{u}>0$. Conversely, if $\sigma>0$ and $\bar{u}=-\infty$, then either $\psi(1)\le 0$, or $\psi(1)>0$ and $q\in(0,\psi(r/\psi(1))-r]$. 
\end{lem}
\begin{proof}
See Appendix \ref{sec:prof}.
\end{proof}
Using the monotonicity of the function $g(\cdot)$ over $[\bar{u},\infty)$,  we define the following $y$-value: 
\[
\bar{y}=\begin{cases} \log(K/g(\bar{u})), &\text{for } \bar{u} \geq 0 \\
\infty, &\text{for } \bar{u}=-\infty \end{cases}
\]
which makes $\bar{u}+\bar{y}$ a solution to the first order condition equation \eqref{eq:gK}.
Then, in view of the facts that $e^y\,g(z-y)$ is strictly increasing in the parameter $y$ (due to the monotonicity of $\bar{\Lambda}(\cdot)$) and that $\bar{u}$ is a local minimum of $g(\cdot)$, we can conclude that there is no solution to \eqref{eq:gK} greater than $\bar{u}+y$ for all $y>\bar{y}$. 
Moreover, for all finite $y\le\bar{y}$, there exists a unique solution to \eqref{eq:gK} greater than or equal to $\bar{u}+y$. We define this candidate optimal threshold by 
\be \label{z*}
z^\star(y):=y+\inf\{u>\bar{u}: g(u)> K\et^{-y}\} \quad \text{for} \quad y\leq \bar{y}.
\ee 
From the observations above, it is also seen that $z^\star(y)$ is the largest root to \eqref{eq:gK}, 
and strictly decreasing and continuously differentiable at $y$ as long as $g'(z^\star(y)-y)\neq0$. 
The limiting behaviour of $z^\star(\cdot)$ follows from the limiting behaviour of $\bar{\Lambda}(\cdot)$ in Lemma \ref{Lambdaprop}, which implies that  
\[
z^\star(y)=\log\Big( \frac{K\bar{\Lambda}(z^\star(y)-y)}{\bar{\Lambda}(z^\star(y)-y)-1} \Big) \,\to\, \log\Big( \frac{K\Phi(r)}{\Phi(r)-1} \Big) \equiv \ol{k} 
\quad \text{as} \quad y\downarrow-\infty, 
\]
and agrees with Remark \ref{candidate}.
One important property of this root is that the function $g(\cdot)$ is non-decreasing over $[z^\star(y),\infty)$, which will be used to show the super-martingale property of the value function $U(\cdot;y,z^\star(y))$ (see proof of Proposition \ref{thm:super}).

Inspired by the analysis in \cite{TimKazuHZ14}, we investigate the connection between the equation \eqref{eq:gK}, that a candidate optimal exercising threshold should satisfy, and the intrinsic value function. 
In particular, using \eqref{eq:gK} and \eqref{eq:hazard} 
along with Proposition  \ref{prop:lap11} and Lemma \ref{lem W}, we can prove the following lemma. 
\begin{lem}\label{func}
Recall the doubly stochastic time $\zeta$ defined in \eqref{eqzeta}. 
We have
\benn
\et^x-K=\ex_x\big[\et^{y}g( \overline{X}_{\zeta}-y)\big]-K,\quad\forall x\in\R\,.
\eenn
\end{lem}
\begin{proof}
See Appendix \ref{sec:prof}.
\end{proof}
We now present an alternative representation of the value of the candidate optimal up-crossing strategy $T_{z^\star(y)}^+$ defined by $U(x;y,z^\star(y))$ from \eqref{U} with \eqref{z*}, for which we can prove some useful properties eventually leading to the optimality of this strategy. 
\begin{prop}\label{thm:super}
For all finite
$y\leq \bar{y}$,
let us define the positive function
\be
V(x;y):=\ex_x\big[(\et^yg( \overline{X}_{\zeta}-y)-K)\ind_{\{\overline{X}_{\zeta}> z^\star(y)\}}\big]. \label{Vxy}
\ee
Then we have
\begin{enumerate}
\item[(i)] The process $(\exp(-A_t^y)V(X_t;y))_{t\ge0}$ is a super-martingale;
\item[(ii)] The process $(\exp(-A_{t\wedge T_{z^\star(y)}^+}^y)V(X_{t\wedge T_{z^\star(y)}^+};y))_{t\ge 0}$ is a martingale;
\item[(iii)] $V(x;y)=\et^x-K$ for all $x\ge z^\star(y)$;
\item[(iv)] $V(x;y)=U(x;y,z^\star(y))$ and $\p_xV(x;y)|_{x=z^\star(y)-}=\et^{z^\star(y)}$.
\end{enumerate}
where $U$ is the value of an up-crossing strategy defined in \eqref{U}.
\end{prop}
\begin{proof}
In order to prove $(i)$, we notice that
\begin{align}
V(x;y)&=\ex_x\big[(\et^y g( \overline{X}_{\zeta}-y)-K)\ind_{\{ \overline{X}_{\zeta}>z^\star(y)\}}\big]\nn\\
&\ge \ex_x\big[\ind_{\{t<\zeta\}}(\et^yg( \overline{X}_{\zeta}-y)-K)\ind_{\{ \overline{X}_{\zeta}>z^\star(y)\}}\big]\nn\\
&=\ex_x\big[\ex_x\big[(\et^yg( \overline{X}_{\zeta}-y)-K)\ind_{\{t<\zeta\}}\ind_{\{ \overline{X}_{\zeta}>z^\star(y)\}}|\mathcal{F}_t\big]\big].\nn\\
&=\ex_x\big[\exp(-A_t^y)\ex_x\big[(\et^yg(\max\{\overline{X}_t,\overline{X}_{t,\zeta_t}\}-y)-K)\ind_{\{ \max\{\overline{X}_t,\overline{X}_{t,\zeta_t}\}>z^\star(y)\}}|\mathcal{F}_t\big]\big],\label{V<=1}
\end{align}
for $t\geq 0$, where for an independent copy of $\mathbf{e}_1$, denoted by $\tilde{\mathbf{e}}_1$, we have defined
\[\overline{X}_{t,\zeta_t}:=\sup_{s\in[t, \zeta_t]}X_s,
\quad \text{and} \quad \zeta_t:=\inf\{s>t: A_s^y-A_t^y>\tilde{\mathbf{e}}_1\}.\] 
Since the function $(\et^yg(x-y)-K)\ind_{\{x>z^\star(y)\}}$ is  non-decreasing, we know from $\max\{\overline{X}_t,\overline{X}_{t,\zeta_t}\}\ge\overline{X}_{t,\zeta_t}$that 
\begin{multline}
\ex_x\big[(\et^yg( \max\{\overline{X}_t,\overline{X}_{t,\zeta_t}\}-y)-K)\ind_{\{ \max\{\overline{X}_t,\overline{X}_{t,\zeta_t}\}>z^\star(y)\}}|\mathcal{F}_t\big]\\
\ge\ex_x\big[(\et^yg(\overline{X}_{t,\zeta_t}-y)-K)\ind_{\{ \overline{X}_{t,\zeta_t}>z^\star(y)\}}|\mathcal{F}_t\big]= V(X_t;y). \label{V<=2}
\end{multline}
As a consequence, $(i)$ follows from the fact that \eqref{V<=1}--\eqref{V<=2} imply 
\begin{align}
V(x;y)
\ge&\ex_x[\exp(-A_t^y)V(X_t;y)].\nn
\end{align}

In order to prove $(ii)$, we notice that 
\begin{align}
V(x;y)=&\ex_x\big[(\et^yg( \overline{X}_{\zeta}-y)-K)\ind_{\{ \overline{X}_{\zeta}>z^\star(y)\}}\big]\nn\\
=&\ex_x\big[\ex_x\big[(\et^yg( \overline{X}_{\zeta}-y)-K)\ind_{\{T_{z^\star(y)}^+<\zeta\}}|\mathcal{F}_{t\wedge T_{z^\star(y)}^+}\big]\big]\nn\\
=&\ex_x\big[\ex_x\big[(\et^yg( \overline{X}_{\zeta}-y)-K)\ind_{\{T_{z^\star(y)}^+<\zeta\}}\ind_{\{t\wedge T_{z^\star(y)}^+<\zeta\}}|\mathcal{F}_{t\wedge T_{z^\star(y)}^+}\big]\big]\nn\\
=&\ex_x\big[\exp(-A_{t\wedge T_{z^\star(y)}^+}^y)\ex_x\big[(\et^yg( \overline{X}_{\zeta}-y)-K)\ind_{\{T_{z^\star(y)}^+<\zeta\}}|\mathcal{F}_{t\wedge T_{z^\star(y)}^+}\big]\big]\nn\\
=&\ex_x\big[\exp(-A_{t\wedge T_{z^\star(y)}^+}^y)\ex_{X_{t\wedge T_{z^\star(y)}^+}}\big[(\et^yg(\tilde{M}_\zeta-y)-K)\ind_{\{\tilde{M}_\zeta>z^\star(y)\}}\big]\big]\nn\\
=&\ex_x\big[\exp(-A_{t\wedge T_{z^\star(y)}^+}^y)V(X_{t\wedge T_{z^\star(y)}^+};y)\big],\nn
\end{align}
where $\tilde{M}_\zeta$ is independent of $\mathcal{F}_t$ and has the same law as $ \overline{X}_{\zeta}$ under $\pr_{X_{t\wedge T_{z^\star(y)}^+}}$. This proves part $(ii)$.

Given the representation in Lemma \ref{func}, it is straightforward to see that $(iii)$ holds. 

Finally, we show that $\big(\exp(-A_{t\wedge T_{z^\star(y)}^+}^y)V(X_{t\wedge T_{z^\star(y)}^+};y)\big)_{t\ge0}$ is a uniformly integrable martingale, which implies $(iv)$ by taking $t\rightarrow \infty$. To this end, notice that, $\overline{X}_\zeta-x\le \overline{X}_{(A^{-\infty})^{-1}(\mathbf{e}_1)}-x$, where the latter has the same law (under $\pr_x$) as an exponential random variable with mean $1/\Phi(r)$. Hence, it follows from \eqref{eq:gK} and Lemma \ref{Lambdaprop} that 
\bq
V(x;y)\le \ex_x[e^y g(\overline{X}_\zeta - y)] 
\le \ex_x\big[\et^{\overline{X}_\zeta}\big]\frac{\prq-1}{\prq} 
=\frac{\Phi(r)}{\Phi(r)-1}\frac{\prq-1}{\prq}\et^x.\nn
\eq
Using the fact that $\exp(-A_{t}^y)\le \et^{-rt}$ for all $t\ge0$, we have 
\begin{align}
\exp\big(-A_{t\wedge T_{z^\star(y)}^+}^y\big) V\big(X_{t\wedge T_{z^\star(y)}^+};y\big)\le& 
\frac{\Phi(r)}{\Phi(r)-1}\frac{\prq-1}{\prq}\et^{-r(t\wedge T_{z^\star(y)}^+)+X_{t\wedge T_{z^\star(y)}^+}}\nn\\
\le&\frac{\Phi(r)}{\Phi(r)-1}\frac{\prq-1}{\prq}\et^{M},\nn
\end{align}
where $M:=\sup_{t\in[0,\infty)}(-rt+X_t)$ is exponentially distributed with mean $1/\tilde{\Phi}(r)$. Here $\tilde{\Phi}(r):=\sup\{\beta\ge0 : \psi(\beta)-r\beta\le 0\}$, which satisfies $\tilde{\Phi}(r)>1$ due to the fact that $r>\psi(1)$. The claimed uniformly integrability follows by the dominated convergence theorem. 

Lastly, the smooth fit condition holds at $x=z^\star(y)$ since 
\begin{align}
\et^{z^\star(y)}-\p_xV(x;y)|_{x=z^\star(y)-}=&\et^{z^\star(y)}-\p_xU(x;y,z^\star(y))|_{x=z^\star(y)-}\nn\\
=&\et^{z^\star(y)}-U(z^\star(y);y,z^\star(y))\overline{\Lambda}(z^\star(y)-y)\nn\\
=&\et^{z^\star(y)}-(\et^{z^\star(y)}-K)\overline{\Lambda}(z^\star(y)-y)\nn\\
=&\overline{\Lambda}(z^\star(y)-y)(K-\et^{y}g(z^\star(y)-y))=0,\nn
\end{align}
where the last equality is due to the definition of $z^\star(y)$.
\end{proof}

\begin{rmk}\label{rmk:j} \normalfont
Since Proposition \ref{thm:super} identifies $U(x;y,z^\star(y))\equiv V(x;y)$, from \eqref{U} and \eqref{Vxy}, 
we can use the properties outlined in parts $(i)$-$(iii)$ and the classical verification method (see, e.g. proof of theorems in \cite[Section 6]{Alili2005}), in order to establish the optimality of the threshold strategy $T_{z^\star(y)}^+$, or equivalently that the value function is given by $v(x;y)=U(x;y,z^\star(y))$. 
For this, it only remains to verify the additional property that $V(x;y)> \et^x-K$ for all $x<z^\star(y)$. 
This is equivalent to proving that 
\be
R(x;y) < 1 , \quad\forall \; x < z^\star(y) , 
\quad \text{where } \; R(x;y):=\frac{\et^x-K}{V(x;y)}.\label{eq:ratio}
\ee
\end{rmk}

Below, we examine the optimality of the threshold strategy $T^+_{z^\star(y)}$ in two separate subsections, based on the possible values of $\bar{u}$, and we provide the alternative optimal strategies when $T^+_{z^\star(y)}$ is shown not to be optimal.


\subsection{The case $\bar{u}=-\infty$: Proof of Theorem \ref{thm:ult1}}\label{sec:H>0}

As mentioned in Remark \ref{rmk:j}, the only non-trivial part remaining in order to prove the optimality of $T_{z^\star(y)}^+$ is to show that 
the inequality \eqref{eq:ratio} holds true.
As soon as we prove this, we will know by Proposition \ref{thm:super} and Remark \ref{rmk:j}, that $v(x;y)=V(x;y)=U(x;y,z^\star(y))$ and smooth fit holds at the exercise threshold $x=z^\star(y)$. 

To this end, we notice that $R(z^\star(y);y)=1$ and that for all $x<z^\star(y)$, we have
\be
\p_xR(x;y)=\frac{\bar{\Lambda}(x-y)}{V(x;y)}(K-\et^yg(x-y)). \label{dxR}
\ee
Using the fact that $g(\cdot)$ is increasing in this case, we conclude that 
$K-\et^yg(x-y)>K-\et^yg(z^\star(y)-y)=0$ for $x\in(-\infty,z^\star(y))$, so that \eqref{dxR} implies that $\p_xR(x;y)>0$ and thus $R(x;y)<R(z^\star(y);y)=1$ for all $x\in(-\infty,z^\star(y))$.

\subsection{The case $\bar{u}\geq 0$: Proof of Theorem \ref{thm:ult2}}\label{sec:BV}

Recall that in both cases $\bar{u}=0$ and $\bar{u}> 0$, $g(\cdot)$ is not monotone, hence there is no guarantee for 
the validity of inequality \eqref{eq:ratio}.
In fact, we can show that \eqref{eq:ratio} fails to hold for $y=\bar{y}$. 
To see this, observe that $z^\star(\bar{y}) \equiv \bar{y}+\bar{u}$, where $\bar{u}$ is a local minimum of $g(\cdot)$. We thus know that 
$\et^{\bar{y}}g(\bar{u}-\epsilon)>K$ for any sufficiently small $\epsilon>0$.
As a result, for all $x\in(z^\star(\bar{y})-\epsilon, z^\star(\bar{y}))$, we get from \eqref{dxR} that $\p_x R(x;\bar{y})<0$, 
which yields that $R(x;\bar{y})>R(z^\star(\bar{y});\bar{y})=1$ in a sufficiently small left neighborhood of $z^\star(\bar{y})$ 
and \eqref{eq:ratio} fails. 
It is therefore crucial to find the {\em critical} $y$-interval, such that \eqref{eq:ratio} remains valid, thus the up-crossing threshold strategy $T^+_{z^\star(y)}$ is optimal and establish Theorem \ref{thm:ult2}$(i)$.

\vspace{2pt}
\noindent{\bf Proof of part $(i)$ of Theorem \ref{thm:ult2}.}
We begin by using the fact that $g(\cdot)$ is ultimately increasing, in order to define  
\be
y_0 := \inf\Big\{ y\leq \bar{y} \;:\, \sup_{u\in(-\infty,\bar{u}]}g(u) = \et^{-y} K \Big\}. \label{y0}
\ee
Since $\bar{u}$ is a local minimum of $g(\cdot)$, we have $y_0<\bar{y}$. 
Then, for any fixed $y\le y_0$, we have 
\[g(u)\le \et^{-y}K, \quad\forall u\in(-\infty, \bar{u}].\]
In addition, for any fixed $y\le y_0$, by $z^*(y)-y\ge z^*(y_0)-y_0>z^*(\bar{y})-\bar{y}=\bar{u}$ and   
the fact that $g(\cdot)$ is non-decreasing over $[\bar{u},\infty)$, we know (by constructions of $y_0$ and $z^*(y)$) that
\[g(u)<g(z^*(y)-y)=\et^{-y}K, \quad\forall u\in[\bar{u},z^*(y)-y).\]
All together, we have for any fixed $y\leq y_0$, that
\[
\et^{y}g(x-y)\le K,\quad\forall x\in(-\infty, z^\star(y)),
\]
where the inequality is strict at least when $x\in[y+\bar{u}, z^\star(y))$.
Combining this with \eqref{dxR}, we see that $R(\cdot;y)$ is non-decreasing over $(-\infty, y+\bar{u}]$, and is strictly increasing over $[y+\bar{u}, z^\star(y))$, which yields that
\be
R(x;y)<R(z^\star(y);y)\equiv 1 , \quad\forall x\in(-\infty, z^\star(y)) 
\quad \text{for } \; y\in(-\infty, y_0] . \label{y<y0}
\ee
Hence, for any fixed $y\le y_0$, the threshold type strategy $T_{z^*(y)}^+$ is optimal.

For $y\in (y_0,\bar{y})$, although  $R(\cdot;y)$ is still increasing over $[y+\bar{u}, z^\star(y))$ (by the constructions of $\bar{u}$ and $z^*(y)$, and \eqref{dxR}), $R(\cdot;y)$ is not monotone over  $(-\infty, y+\bar{u}]$ and has at least one local maximum in this interval. 
Also, note that for every fixed $y\in [y_0,\bar{y})$ and fixed $x\in(-\infty,z^\star(y))$,  $R(x;\cdot)$ is strictly increasing over $(y_0,y)$. To see this, use the facts that $\p_z U(x;y,z)$ vanishes at $z=z^\star(y)$ and that $z^\star(y)$ is non-increasing to calculate
\begin{align}
\diff_y\log R(x;y)
=& -\frac{\p_y U(x;y,z^\star(y))\diff y + \p_z U(x;y,z))|_{z=z^\star(y)}\cdot\diff_y z^\star(y)}{U(x;y,z^\star(y))}\nn\\
=& \begin{cases} 
\overline{\Lambda}(z^\star(y)-y)\,\diff y \,, & \text{for } x\in(-\infty, y) \\
\left(\overline{\Lambda}(x-y)-\overline{\Lambda}(z^\star(y)-y)\right)\diff y \,, & \text{for } x\in[y,z^\star(y))
\end{cases}
, \nn
\end{align}
which 
implies that $\p_y\log R(x;y)$ is well-defined and is strictly positive, due to $\overline{\Lambda}(\cdot)$ being strictly decreasing over $[0,\infty)$ (see Lemma \ref{Lambdaprop}).
The monotonicity of $R(x;\cdot)$ implies that the inequality \eqref{eq:ratio} will fail for some values of $y$ larger than $y_0$. 
In view of these observations, we can define 
\be
\tilde{y} := \inf \Big\{ y\leq \bar{y} \,:\, \sup_{x<\bar{u}+y} R(x;y)=1\Big\} \label{ytilde}
\ee
By the definition \eqref{y0} of $y_0$ and the discussion on $\bar{y}$ at the beginning of this subsection, we know that $\tilde{y}$ is well-defined and satisfies $\tilde{y}\in(y_0,\bar{y})$. Moreover, recall that $g(\cdot)$ is increasing over $(\bar{u},\infty)$,  so for each fixed $y<\bar{y}$, $z^\star(y)$ is a local maximum of $R(\cdot;y)$. On the other hand, 
recall that $R(\cdot;y)$ is strictly increasing over $[\bar{u}+y,z^\star(y))$ and that $R(z^\star(y);y)=1$, so
we know that, for any fixed $y \in (y_0,\tilde{y})$, 
\be
R(x;y)<1 , \quad\forall x\in(-\infty, z^\star(y)). \label{y<tildey}
\ee
Hence, for any fixed $y\in(y_0,\tilde{y})$, the threshold type strategy $T_{z^*(y)}^+$ is optimal.
The proof of Theorem \ref{thm:ult2}$(i)$ is complete by combining inequalities \eqref{y<y0} and \eqref{y<tildey}.
\qed
\vspace{3pt}

We have proved that when $y$ is sufficiently small, i.e. $y<\tilde{y}$, the optimal stopping region $\mc{S}^y=[z^\star(y),\infty)$ has only one connected component. On the other hand, a similar situation occurs if $y$ is sufficiently large, in which case
the optimal stopping reigon $\mc{S}^y=\mc{S}^\infty\equiv[\underline{k},\infty)$. 
To show this, we demonstrate below the proof of Theorem \ref{thm:ult2}$(ii)$, namely that, $\underline{v}(\cdot)$ is the value function of the problem \eqref{eq:problem} for some $y$. 
We will demonstrate the analysis of the remaining, most challenging, proofs for parts $(iii)$ and $(iv)$ of Theorem \ref{thm:ult2} afterwards. 

Let us introduce the following functions $h(\cdot)$ and $f(\cdot)$, which play a crucial role in the rest of the paper (as $\underline{v}(\cdot)$ is also part of \eqref{eq:problem}). We first define  
\be
h(x):=\underline{v}(x)-(\et^x-K),\quad\forall x\in\R ,\label{h}
\ee
which is a non-negative function that is vanishing for all $x\ge\underline{k}$ and is uniformly bounded from above by $K$ (using the obvious fact that $\underline{v}(x)<\et^x$ for all $x\in\R$).
Then, we also define 
\be
f(x):=\int_{(-\infty,\underline{k})}h(z)\Pi(-x+\diff z)\equiv\int_{(-\infty,\underline{k}-x)}h(x+w)\Pi(\diff w),\quad x\ge\underline{k} , \label{eq:ffun}
\ee
which has the following properties, 
that can be proved by using the definition \eqref{eq:ffun} of $f(\cdot)$, the aforementioned properties of $h(\cdot)$ from \eqref{h} as well as that $h(\cdot)$ is  strictly decreasing over $(-\infty,\underline{k})$, 
and the definition \eqref{decomp} of $\psi(\cdot)$. 
\begin{lem}\label{lem:ffun00}
The function $f(\cdot)$ is non-negative, decreasing, continuous over $[\underline{k},\infty)$, and is in $C^1[\underline{k}+\delta,\infty)$ for any $\delta>0$. Moreover, 
\be
f(\underline{k})=(r+q-\psi(1))\underline{K}-(r+q)K-\frac{1}{2}\sigma^2\prq K\quad \text{and} \quad f(\infty)=0.\label{eq4f}
\ee
\end{lem}
\begin{proof}
See Appendix \ref{sec:prof}.
\end{proof}

\vspace{2pt}
\noindent{\bf Proof of part $(ii)$ of Theorem \ref{thm:ult2}.} 
In order to show that the value function is given by $v(\cdot;y)=\underline{v}(\cdot)$ for some $y\ge\bar{y}(>\tilde{y})$,
we adopt the method of proof through the variational inequalities, 
summarized below.

\begin{lem} \label{verif} 
Let $y\in\R$ and a 
function  
$w:\R\rightarrow(0,\infty)$ in $C^1(\R)\cap C^2(\R\backslash\{\theta_1,\ldots,\theta_k\})$ for some $\theta_1,\ldots,\theta_k \in\R$, 
such that $w(x) \geq (\et^x-K)^+$ and is super-harmonic, i.e. it satisfies the variational inequalities
\be \label{verifeq}
\max\{(\mc{L}-r-q\ind_{\{x<y\}})w(x), (\et^x-K)^+ - w(x)\}=0,\quad\forall x\in\R\backslash\{\theta_1,\ldots,\theta_k\} \,.
\ee
Then, by using It\^o-L\'evy lemma and the compensation formula, 
we know that $v(x;y)\equiv w(x)$ is the value function of the problem \eqref{eq:problem} for this $y\in\R$. 
\end{lem}
\begin{proof}
See e.g. \cite{OksendalSulemBook}, or  \cite[Section 3.3]{KazuEgamiAAP14} for a precise example. 
\end{proof}
We already know that $\underline{v}(x)$ 
satisfies all conditions of Lemma \ref{verif} for $\{\theta_1\}=\{\underline{k}\}$, apart from \eqref{verifeq}. We thus need to 
identify all the $y$-values for which \eqref{verifeq} holds true.
On one hand, using the explicit formula of $\underline{v}(x)$ for $x<\underline{k}$ as given in \eqref{eq:underlinev}, we know that $(\mc{L}-r-q)\underline{v}(x)=0$ on $(-\infty,\underline{k})$. 
In conjunction with Lemma \ref{lem:ffun00}, which implies for all $x>\underline{k}$, that 
\be
(\mc{L}-r-q)\underline{v}(x)=(r+q)K-(r+q-\psi(1))\et^x+f(x)\le (r+q)K-(r+q-\psi(1))\underline{K}+f(\underline{k})\le 0,\label{eq:supervl}
\ee
we have
that \eqref{verifeq} holds true and 
$\underline{v}(\cdot)$ is indeed super-harmonic with respect to the discount rate $r+q$ for $x<y$.
In order to examine if $\underline{v}(\cdot)$ is still super-harmonic with respect to discount rate $r$ for $x\ge y$, we use the functions $h(\cdot)$ and $f(\cdot)$ from \eqref{h}--\eqref{eq:ffun} and the above analysis, to define and calculate 
\be
\chi(x):=(\mc{L}-r)\underline{v}(x)=
\begin{dcases}
rK-(r-\psi(1))\et^x+f(x),\quad&\forall\;x\ge\underline{k},\\
q\underline{v}(x), &\forall\;x<\underline{k}.
\end{dcases}\label{eq:57}
\ee
We clearly have $\chi(x)>0$ for all $x<\underline{k}$, with $\chi(-\infty)=0$. Moreover, recalling the standing assumption that $r>\psi(1)$, we know from 
\eqref{eq:57} and Lemma \ref{lem:ffun00} that the function $\chi(\cdot)$ is strictly decreasing over $[\underline{k},\infty)$, with $\chi(\infty)=-\infty$. 
Hence, there exists a unique critical $y$-value, denoted by $y_m$, such that 
\be 
(\mc{L}-r)\underline{v}(x) \equiv\chi(x)=
\begin{cases}
>0 , \quad \text{if } \; x<y_m ,\\ 
\leq 0 , \quad \text{if } \; x\geq y_m 
\end{cases}\label{ym}
\ee 
In words, $y_m$ is the smallest $y$-value such that $\underline{v}(\cdot)$ is super-harmonic with respect to discount rate $r$. So we obviously have $y_m<\ol{k}$. 
%
%
Overall, in light of the above observations, we know that $\underline{v}(\cdot)$ 
satisfies the variational inequalities \eqref{verifeq} 
for all $y\ge y_m$.

Note that, the above analysis also implies that $y_m\ge\bar{y}$. To see this we argue by contradiction, supposing that $y_m<\bar{y}$. Then, for any $x$ in a sufficiently small left neighborhood of $z^\star(y_m)$ (defined in \eqref{z*}), waiting until $T_{z^\star(y_m)}^+$ will yield a strictly better value than stopping immediately, so this $x$ must be in the continuation region $(-\infty,\underline{k})$. 
But then, by the arbitrariness of $x$, we must have $z^\star(y_m)=\underline{k}$, and consequently, by the monotonicity of $z^\star(\cdot)$, we have
\[
\underline{k}-y_m=z^\star(y_m)-y_m>z^\star(\bar{y})-\bar{y}=\bar{u}\ge0\, \quad \Rightarrow \quad \, y_m<\underline{k}.\]
However, the last inequality 
implies from \eqref{eq:57} that $\chi(y_m)=q\,\underline{v}(y_m)>0$, 
which contradicts with the definition of $y_m$ in \eqref{ym}. 
\qed
\vspace{3pt}

We now prove the case when $y=\tilde{y}$.\\
\noindent{\bf Proof of part $(iii)$ of Theorem \ref{thm:ult2}.}
It follows from the proof of Theorem \ref{thm:ult2}$(i)$ and the definition \eqref{ytilde} of $\tilde{y}$, that the inequality \eqref{eq:ratio} fails at some point $x_0 <\bar{u}+\tilde{y}$, which satisfies   
\be
R(x_0;\tilde{y}) = \sup_{x<\bar{u}+\tilde{y}} R(x;\tilde{y}) = 1 . \label{x0}
\ee
Moreover, $x_0$ is a stationary point of $R(\cdot;\tilde{y})$ and solves $K=\et^{\tilde{y}}g(x_0-\tilde{y})$.
Hence, for $y=\tilde{y}$, there is branching of the optimal stopping region due to the facts that $\{x_0\}\cup[z^\star(\tilde{y}),\infty)\subset\mc{S}^{\tilde{y}}$ and  $(\bar{u}+\tilde{y}, z^\star(\tilde{y}))\subset(\mc{S}^{\tilde{y}})^c$. 
Below we examine the two distinct scenarios corresponding to the cases $\bar{u}=0$ and $\bar{u}>0$ (see also Remark \ref{rmk:u}). 

Firstly, if $X$ has paths of bounded variation and \eqref{cond1} holds 
(i.e. $\bar{u}=0$), then by \eqref{g'}--\eqref{eqHv} 
and \eqref{dxR}, we know that $\p_xR(x;\tilde{y})>0$ over $(\tilde{y},z^\star(\tilde{y}))$, implying that $x_0<\tilde{y}$. 
Therefore, by observing that 
\be \label{x0=}
K=\et^{\tilde{y}}g(x_0-\tilde{y})=\frac{\prq-1}{\prq}\et^{x_0} 
\quad \Rightarrow \quad 
x_0=\underline{k}\,,
\ee 
we conclude, in view of Proposition \ref{prop:compare1}$(iii)$, that  $\mc{S}^{\tilde{y}}=\{\underline{k}\}\cup[z^\star(\tilde{y}),\infty)$.
Secondly, if \eqref{cond1} does not hold 
(i.e. $\bar{u}>0$), we 
have in view of Lemma \ref{lem:equivalent} the following equivalence:
\begin{equation} \label{x0y}
\text{Hypothesis \ref{1stop} holds true for }y=\tilde{y} \quad \Leftrightarrow \quad 
\text{The point $x_0$ given by \eqref{x0} satisfies } x_0\in(-\infty,\tilde{y}) \,.
\end{equation}

\begin{rmk} \label{morestop} \normalfont
As it will be shown later on, an additional unique connected component of $\mc{S}^{y}$ inevitably appears in $(-\infty,y)$ for some $y\geq \tilde{y}$, independently of how many disjoint connected components of $\mc{S}^{y}$ exist in $[y,\infty)$. 
In order to present the main ideas in a concise manner, for the purpose of this paper, we do not expand in the direction where Hypothesis \ref{1stop} fails. 
\end{rmk}
\noindent 
Now, when $x_0\in(-\infty,\tilde{y})\cap\mc{S}^{\tilde{y}}$ holds, we have similarly to \eqref{x0=} that $x_0=\underline{k}$ 
and so $\mc{S}^{\tilde{y}}=\{\underline{k}\}\cup[z^\star(\tilde{y}),\infty)$.

Finally, in order to complete the proof, it remains to show that  
the smooth fit condition holds at $x=\underline{k}$ when $y=\tilde{y}$. 
However, this follows directly from the smoothness of $v(x;y)\equiv U(x;\tilde{y};z^\star(\tilde{y}))$.
\qed
\vspace{3pt}

\vspace{3pt}
The above analysis, together with 
\eqref{ytilde} 
and the monotonicity of $R(\underline{k};y)$ in $y$, outlined in the proof of Theorem \ref{thm:ult2}$(i)$, implies that the critical value $\tilde{y}$ is the smallest $y$-value such that $\underline{k}\in\mc{S}^y$, provided that Hypothesis \ref{1stop} holds whenever \eqref{cond1} fails.
Consequently, we know from Proposition \ref{prop:compare1}$(i)$ that 
\be
\{\underline{k}\}\cup[z^\star(\tilde{y}),\infty)\subset\mc{S}^y
\quad \text{and} \quad 
(-\infty,\underline{k})\subset(\mc{S}^y)^c,\quad\forall\; y\in[\tilde{y}, \infty).
\label{eq:someeq}
\ee
Moreover, it follows from this observation, Remark \ref{rmk:use1piece} and Theorem \ref{thm:ult2}$(ii)$, that the disjoint components of stopping region will merge into one when $y\ge y_m$. 
Taking into account all the above, we are ready to study the only remaining case, when $y\in(\tilde{y},y_m)$.

\vspace{2pt}
\noindent{\bf Proof of part $(iv)$ of Theorem \ref{thm:ult2}.}
Provided that Hypothesis \ref{1stop} holds for $\tilde{y}$ (so that $\tilde{y}>\underline{k}$ and $\underline{k}\in\mc{S}^{\tilde{y}}$), 
 a combination of 
\eqref{eq:someeq}, Proposition \ref{prop:compare1}$(iii)$ and the observation from \eqref{ym} that $(y,y_m)\subset(\mc{S}^y)^c$, dictates the consideration of the following pasting points, for all $y\in(\tilde{y},y_m)$:
\begin{align}
a^\star(y):=&\sup\{x\in[\underline{k}, y]: v(x;y)=\et^x-K\},\label{eq:def_as}\\
b^\star(y):=&\inf\{x\in[y_m, z^\star(\tilde{y})]: v(x;y)=\et^x-K\}.\label{eq:def_bs}
\end{align}
In fact, Proposition \ref{prop:compare1}$(iii)$ implies that $[\underline{k},a^\star(y)]\in\mc{S}^y$. 
Hence, for any $x\in(a^\star(y),b^\star(y))$, it is optimal to wait until $T_{a^\star(y)}^-\wedge T_{b^\star(y)}^+$. 
To be more precise, stopping immediately is optimal when either the event $\{T_{b^\star(y)}^+<T_{a^\star(y)}^-\}$ or $\{T_{a^\star(y)}^-<T_{b^\star(y)}^+, X_{T_{a^\star(y)}^-}\in[\underline{k},a^\star(y)]\}$ occurs. 
However, an immediate stop is not optimal when 
$\{T_{a^\star(y)}^-<T_{b^\star(y)}^+, X_{T_{a^\star(y)}^-}<\underline{k}\}$ occurs due to an overshoot; waiting until $X$ increases to $\underline{k}$ would then be optimal. 
Taking these into account, 
the value function $v(\cdot;y)$ in \eqref{eq:problem}, for all $x\in(a^\star(y),b^\star(y))$, takes the form
\begin{align}
v(x;y)&=\ex_x\Big[ \exp\big(-A_{T_{a^\star(y)}^-\wedge T_{b^\star(y)}^+}^y\big) \underline{v}\big(X_{T_{a^\star(y)}^-\wedge T_{b^\star(y)}^+}\big) \Big] 
=:V(x;y,a^\star(y),b^\star(y)).\label{eqeqeqeqeq}
\end{align}

In view of deriving the above value, 
we use a result from \cite[Theorem 2]{OccupationInterval} for the occupation time at the first up-crossing exit. 
For any $r\ge 0, q>0$, and $x\le b$ with $a\le y\le b$, we have
\begin{align}
\ex_x[\exp(-A_{T_b^+}^y)\ind_{\{T_b^+<T_a^-\}}]&=\frac{W^{(r,q)}(x,a)}{W^{(r,q)}(b,a)},
\label{uphit}
\end{align}
where we define the non-negative function (see  \cite[(6)-(7)]{OccupationInterval})
\begin{align}
W^{(r,q)}(x,a)
&:=W^{(r+q)}(x-a)-q\int_{(y,x\vee y)}W^{(r)}(x-z)W^{(r+q)}(z-a)\diff z \label{eq:bigW1} \\
&=W^{(r)}(x-a)+q\int_{(a,y)}W^{(r)}(x-z)W^{(r+q)}(z-a)\,\diff z\,. \label{newW}
\end{align}
In addition, we prove and use the 
following proposition, which also provides a generalization of 
the case with deterministic discounting $r$ in \cite[Theorem 2]{loeffen_outshoot} to the case with state-dependent discount rate $r+q\ind_{\{X_t<y\}}$. 

\begin{prop}\label{prop:V2value}
Let $F(\cdot)$ be a positive, non-decreasing, continuously differentiable function on $\R$, and further suppose that $F(\cdot)$ has an absolutely continuous derivative with a bounded density over $(-\infty,b]$ for any fixed $b$ if $X$ has paths of unbounded variation.
We have for all $a\le y<b$ and $x\in(a,b)$ that 
\begin{align} 
\label{V2V22Fab}
\ex_x\big[\exp(-A_{T_a^-\wedge T_b^+}^y) F(X_{T_a^-\wedge T_b^+})&\big] 
= F(x) + \int_{(a,b)}u^{(r,q)}(x,w;y,a,b)\cdot(\mc{L}-r-q\ind_{\{w<y\}}) F(w) \diff w \,,\\
\label{V2V22Fa}
\ex_x\big[\exp(-A_{T_a^-}^y) F(X_{T_a^-})\ind_{\{T_a^-<T_b^+\}}&\big] \\
= F(x) - &\frac{W^{(r,q)}(x,a)}{W^{(r,q)}(b,a)}F(b) + \int_{(a,b)}u^{(r,q)}(x,w;y,a,b)\cdot(\mc{L}-r-q\ind_{\{w<y\}}) F(w) \diff w\,,\nn
\end{align}
where 
\begin{align}
u^{(r,q)}(x,w;y,a,b):=&\ind_{\{w\in(a,y)\}}\bigg(\frac{W^{(r,q)}(x,a)}{W^{(r,q)}(b,a)}W^{(r,q)}(b,w)-W^{(r,q)}(x,w)\bigg)\nn\\
&+\ind_{\{w\in[y,b)\}}\bigg(\frac{W^{(r,q)}(x,a)}{W^{(r,q)}(b,a)}W^{(r)}(b-w)-W^{(r)}(x-w)\bigg)\,.\nn
\end{align}
\end{prop}
\begin{proof}
See Appendix \ref{sec:prof}.
\end{proof}

Therefore, if $\bar{u}\ge0$ and $y\in(\tilde{y},y_m)$, we have for all $\underline{k}\le a\le y<b$ that 
the function $V(x;y,a,b)$  defined in \eqref{eqeqeqeqeq} for all $x\in(a,b)$ with $a\ge \underline{k}$, is given by
\begin{align} \label{V2V22}
V(x;y,a,b)
&=\underline{v}(x)+\int_{(a,b)}u^{(r,q)}(x,w;y,a,b)\cdot \big[\chi(w)-q\ind_{\{w<y\}}\underline{v}(w)\big] \diff w \,, 
\end{align}
where $\chi(\cdot)$ is defined in \eqref{eq:57}.
In view of 
this explicit formula, 
one can easily see that the mapping $x\mapsto v(x;y)\equiv V(x;y,a^\star(y),b^\star(y))$ is in $C^1(a^\star(y),b^\star(y))$.
If $X$ has bounded variation, then we already know from Proposition \ref{prop:compare1}$(ii)$ that continuous fit should hold at $a^\star(y)$ and $b^\star(y)$. 
However, if $X$ has unbounded variation, by exploiting the optimality of thresholds $a^\star(y)$ and $b^\star(y)$, we show in what follows that the smooth fit conditions must hold as well. 
Remarkably, using a similar argument, we show that smooth fit holds at $b^\star(y)$ even when $X$ has bounded variation. 
\begin{prop}\label{prop:sfit}
The following smooth fit properties holds:
\begin{enumerate}
\item[(i)]
If $\bar{u}\ge 0$, then for any $y\in(\tilde{y};y_m)$, smooth fit holds at $b^\star(y)$, i.e.
\[ \lim_{x\uparrow b^\star(y)}(\p_xv(x;y))=\et^{b^\star(y)}.\]
\item[(ii)] If $X$ has unbounded variation, $\bar{u} > 0$ and $y\in(\tilde{y},y_m)$, then $y\notin\mc{S}^y$ and smooth fit holds at $a^\star(y)$, i.e.
\[\lim_{x\downarrow a^\star(y)}(\p_x v(x;y))=\et^{a^\star(y)}.\]
\end{enumerate}
\end{prop}
\begin{proof}
See Appendix \ref{sec:prof}.
\end{proof}

Therefore, using the fact that we have smooth/continuous fit at the optimal exercising thresholds $a^\star(y)$ and $b^\star(y)$ by Propositions \ref{prop:compare1} and \ref{prop:sfit}, 
as well as using \eqref{eqeqeqeqeq}--\eqref{V2V22}, we derive 
a necessary condition for the pair $(a,b)$ to be identified with the optimal $(a^\star(y),b^\star(y))$.
\begin{cor}\label{cor:fit}
For the case $\bar{u}\ge0$ and $y\in(\tilde{y},y_m)$, the optimal thresholds $a^\star(y)$ and $b^\star(y)$ solve the following system:
\benn
\Delta(b,a;y)=\p_b\Delta(b,a;y)=0,\quad \text{for } \; (a,b)\in[\underline{k},y)\times[y_m,z^\star(y)], 
\eenn
where, with $\chi(\cdot)$ defined in \eqref{eq:57}, we have 
\begin{equation} \label{D}
\Delta(x,a;y):=\int_{(a,y)}W^{(r,q)}(x,w)\cdot[q\,\underline{v}(w)-\chi(w)]\diff w-\int_{[y,x\vee y)}W^{(r)}(x-w)\cdot\chi(w)\diff w.
\end{equation}
\end{cor}
\begin{proof}
{\ul{\it Bounded variation case:}} 
Letting $x\in(a^\star(y),y)$ and using 
\eqref{eq:bigW1}, we can conclude from 
\eqref{V2V22} and some straightforward calculations 
that 
\[V(x;y,a^\star(y),b^\star(y))=\et^x - K + \Delta(x,a^\star(y);y) - \frac{W^{(r+q)}(x-a^\star(y))}{W^{(r,q)}(b^\star(y),a^\star(y))}\Delta(b^\star(y),a^\star(y);y).
\]
On one hand, it is easily seen from \eqref{eqeqeqeqeq} that 
$V(x;y,a^\star(y),b^\star(y))|_{x=a^\star(y)+}=\et^{a^\star(y)}-K$ 
holds by construction. On the other hand,
suppose $\Delta(b^\star(y),a^\star(y);y)\neq0$ in the above equation, then since $W^{(r+q)}(0)>0$ due to Lemma \ref{lem W}, the continuous fit condition at $a^\star(y)$ will not hold, which is a contradiction. 
The other equality can be straightforwardly obtained from the smooth fit condition satisfied by $V$ at $b^\star(y)$.

{\ul{\it Unbounded variation case:}}
In this case, $W^{(r+q)}(0)=0$ by Lemma \ref{lem W}, thus the above equation in Step 1 is satisfied immediately. However, we can derive from it, for $x\in[a^\star(y),y]$, that
\begin{align*}
\p_xV(x;y,a^\star(y),b^\star(y))=&\et^x + \p_x\Delta(x,a^\star(y);y)-\frac{W^{(r+q)\prime}(x-a^\star(y))}{W^{(r,q)}(b^\star(y),a^\star(y))}\Delta(b^\star(y),a^\star(y);y).
\end{align*} 
We know from \eqref{eqeqeqeqeq} and Proposition \ref{prop:sfit}$(ii)$ that 
$\p_xV(x;y,a^\star(y),b^\star(y))|_{x=a^\star(y)+}=\et^{a^\star(y)}$ holds. 
By supposing that $\Delta(b^\star(y),a^\star(y);y)\neq 0$ in the above equation and using the facts that $W^{(r,q)}(b^\star(y),a^\star(y))>0$ and $W^{(r+q)\prime}(0+)>0$ due to  Lemma \ref{lem W}, we conclude that the smooth fit condition at $a^\star(y)$ does not hold, which is a contradiction. Hence, $a^\star(y)$ and $b^\star(y)$ solve $\Delta(b,a;y)=0$.
The other equality can be straightforwardly obtained from the smooth fit condition satisfied by $V$ at $b^\star(y)$.
\end{proof}

This corollary states that $x=b^\star(y)$ is  both a zero and a stationary point for the function $\Delta(x,a^\star(y);y)$. Based on these facts, we notice from \eqref{eqeqeqeqeq}--\eqref{V2V22} 
that 
\be \label{Vxa}
v(x;y)
= V(x;y,a^\star(y),b^\star(y))
= \Delta(x,a^\star(y);y)+\et^x-K,\quad\forall x\in[a^\star(y),b^\star(y)].
\ee
By the definitions of $a^\star(y)$ and $b^\star(y)$ (see \eqref{eq:def_as} and \eqref{eq:def_bs}), we have $\Delta(x,a^\star(y);y)>0$ for all $x$ in a sufficiently small left neighborhood of $b^\star(y)$, from which we conclude that $b^\star(y)$ is either a local minimum or an inflection point of the function $V(\cdot;y,a^\star(y),b^\star(y))$. 
By exploiting this observation, we are able to give a complete characterization of
(i.e. sufficient conditions for determining) the pair $(a^\star(y),b^\star(y))$. 

\begin{prop}\label{thm:pair}
Assuming that $\bar{u}\ge0$ and $y\in(\tilde{y},y_m)$, we have: 
\begin{enumerate}
\item[(i)] $\Delta(a,a;y)=0$, for any fixed $a\in[\underline{k},y)$. 
If $X$ has unbounded variation, then we also have $\p_x\Delta(x,a;y)|_{x=a+}=0$. Moreover, $\Delta(x,a;y)>0$ for all $x\in(a,y]$;
\item[(ii)] Let $\mc{N}^a:=\{x\in(y,z^\star(\tilde{y})]: \Delta(x,a;y)\le 0\}$, which is a closed set for each $a\in[\underline{k},y)$. Then $a^\star(y)\in(\underline{k},y)$, such that 
\be
\label{a*b*}
a^\star(y)=\inf\{a\in[\underline{k},y): \mc{N}^a\neq\emptyset\} 
\quad \text{and} \quad 
b^\star(y)=\inf\mc{N}^{a^\star(y)} \,. 
\ee
That is, $b^\star(y)$ is a global minimum point of the function $\Delta(\cdot,a^\star(y);y)$ over $(a^\star(y),z^\star(\tilde{y}))$, and $a^\star(y)$ is the unique $a\in(\underline{k},y)$ such that $\inf_{x\in(y,z^\star(\tilde{y})]}\Delta(x,a;y)=0$.
\end{enumerate}
\end{prop}
\begin{proof}
Claim $(i)$ follows from the construction of $\Delta(x,a;y)$. To see this, let $x\in(a,y)$ and use 
\eqref{D} and 
\eqref{eq:bigW1}
to see that 
\[
\Delta(x,a;y)=\int_{(0,x-a)}\big[q\,\underline{v}(x-z)-\chi(x-z)\big]W^{(r+q)}(z)\diff z.
\]
The smoothness follows by combining the above with \eqref{Vxa} and Proposition \ref{prop:sfit}. 
Using also the monotonicity of $f(\cdot)$ (and hence $\chi(\cdot)$) from Lemma \ref{lem:ffun00}, \eqref{eq4f} and the fact that $r>\psi(1)$, we have for all $z\in(0,x-a)$ that
\begin{align}
q\,\underline{v}(x-z)-\chi(x-z) 
&\ge q\,\underline{v}(a)-\chi(a)\ge q\,\underline{v}(\underline{k})-\chi(\underline{k})= -(\mc{L}-r-q)\underline{v}(x)|_{x=\underline{k}} \ge 0. \label{eq:fneg}
\end{align}  
Hence, the first claim is proved. 

In order to prove claim $(ii)$, we let $\underline{a}(y):=\inf\{a\in[\underline{k},y): \mc{N}^a\neq\emptyset\}$. First recall from Theorem \ref{thm:ult2}$(ii)$, that $a^\star(\tilde{y})=\underline{k}$ and $b^\star(\tilde{y})=z^\star(\tilde{y})$, hence $\underline{a}(\tilde{y})=\underline{k}$ and $\Delta(x,\underline{k};\tilde{y})\ge0$ for all $x\in[\underline{k},z^\star(\tilde{y})]$. 
Then, for any fixed $x>y>a\ge\underline{k}$, using \eqref{eq:bigW1}, 
\eqref{newW} and \eqref{Vxa} we have
\begin{align}
\p_y\Delta(x,a;y)
=&qW^{(r)}(x-y)\left(\underline{v}(y)+\int_{(a,y)}W^{(r+q)}(y-w)[q\,\underline{v}(w)-\chi(w)]\diff w\right)\nn\\
=&qW^{(r)}(x-y)\big(\underline{v}(y) + \Delta(y,a;y)\big) >0,\label{eq:Ddev}
\end{align}
where we used the conclusion from $(i)$ in the last inequality.
It thus follows that, for all $y>\tilde{y}$, we have
\[
\inf_{x\in(y,z^\star(\tilde{y})]}\Delta(x,\underline{k};y)>0,
\]
which implies that $\mc{N}^{\underline{k}}=\emptyset$ hence  $\underline{a}(y)>\underline{k}$, for $y>\tilde{y}$. On the other hand, 
 since $b^\star(y)\in \mc{N}^{a^\star(y)}$, we know that $\underline{k}<\underline{a}(y)\le a^\star(y)$. Suppose now that $\underline{a}(y)<a^\star(y)$ and since $\mc{N}^{\underline{a}(y)}\neq\emptyset$, let $\underline{b}(y)=\inf\mc{N}^{\underline{a}(y)}$. By the definition of $\underline{b}(y)$ and claim $(i)$, we conclude that $\Delta(x,\underline{a}(y);y)>0$ for all $x\in(\underline{a}(y),\underline{b}(y))$. However, for any fixed $x>a^\star(y)$, taking the derivative of $\Delta(x,a;y)$ with respect to $a\in(\underline{k},y)$ we get 
\be
\p_a\Delta(x,a;y)
=-W^{(r,q)}(x,a)\left(q\underline{v}(a)-\chi(a)\right) <0 , 
\label{eq:inverse}
\ee
where the inequality follows from \eqref{eq:fneg}. Thus, we have that $\Delta(x,\underline{a}(y);y)>\Delta(x,a^\star(y);y)\ge0$ for all $x\in(a^\star(y),b^\star(y)]$ so $\underline{b}(y)>b^\star(y)$. However, from \eqref{eqeqeqeqeq}--\eqref{V2V22}, we know that for these $x$,
\begin{align}
v(x;y)\ge&\ex_{x}\Big[\exp\big(-A_{T_{\underline{a}(y)}^-\wedge T_{\underline{b}(y)}^+}^y\big) \underline{v}\big(X_{T_{\underline{a}(y)}^-\wedge T_{\underline{b}(y)}^+}\big) \Big]
=\Delta(x,\underline{a}(y);y)+\underline{v}(x)
>\Delta(x,a^\star(y);y)+\underline{v}(x)
=v(x;y),\nn
\end{align}
which is a contradiction. Hence, the only possibility is to have  $a^\star(y)=\underline{a}(y)\in(\underline{k},y)$ and thus  $b^\star(y)=\inf\mc{N}^{a^\star(y)}$ and $\inf_{x\in(y,z^\star(\tilde{y})]}\Delta(x,a^\star(y);y)=0$. Finally, since $\Delta(x,\cdot;y)$ is strictly decreasing for every fixed $x>y$, we know that there is no other $a\in(a^\star(y),y)$ such that $\inf_{x\in(y,z^\star(\tilde{y})]}\Delta(x,a;y)=0$. This completes the proof.
\end{proof}

Notice that Proposition \ref{thm:pair} proves that $a^\star(y)\in(\underline{k},y)$, if $\bar{u}\ge0$ and $y\in(\tilde{y},y_m)$, where $a^\star(y)$ is the smallest $a$-value such that the curve $\Delta(\cdot,a;y)$, emanating from the $x$-axis at $x=a$, will revisit the x-axis, 
so that $\Delta(\cdot,a^\star(y);y)>0$ for all $x\in(a^\star(y),b^\star(y))$.
Recalling from Proposition \ref{prop:compare1}$(iii)$ and Remark \ref{rmk:use1piece} that $[\underline{k},a^\star(y)]$ is the only component of stopping region in the interval $(-\infty,y)$, we know from Proposition \ref{prop:compare1}$(ii)$ that $a^\star(\cdot)$ is necessarily continuous over $[\tilde{y},y_m)$. 
Furthermore, 
recall from \eqref{eq:someeq} that $\mc{S}^{\tilde{y}}=\{\underline{k}\}\cup[z^\star(\tilde{y}),\infty)\equiv[\underline{k},a^\star(\tilde{y})]\cap[b^\star(\tilde{y}),\infty) \subseteq \mc{S}^{y}$, 
i.e. there is no continuation region in $[b^\star(\tilde{y}),\infty)$
for all $y\in[\tilde{y},y_m)$.
Hence, recalling Lemma \ref{lem:equivalent}, we have the following equivalence:
\begin{equation} \label{bcont}
\text{Hypothesis \ref{1stop} holds true for }y\in(\tilde{y},y_m) \quad \Leftrightarrow \quad 
\text{$b^\star(\cdot)$ given by \eqref{a*b*} is continuous over } (\tilde{y},y_m) \,.
\end{equation}
\begin{rmk} \label{morestop2} \normalfont
Equivalence \ref{bcont} implies that, if Hypothesis \ref{1stop} fails for some $y_0\in(\tilde{y},y_m)$, namely the function $b^\star(\cdot)$ experiences a jump at $y_0$, then in view of \eqref{a*b*} and Proposition \ref{prop:compare1}$(i)$, we will have further branching of the stopping region inside $(y_0,\infty)$. 
In other words, 
we will have $\mc{S}^{y_0}=[\underline{k}, a^\star(y_0)] \cup \{b^\star(y_0)\} \cup [b^\star(y_0-), \infty)$. 
\end{rmk}

Indeed, 
the right hand side of \eqref{bcont} is equivalent to the fact that 
$[b^\star(y),\infty)$ is a component of $\mc{S}^y$ for all $y\in[\tilde{y},y_m)$. 
Overall, if $y\in(\tilde{y},y_m)$, then $\mc{S}^y=[\underline{k},a^\star(y)]\cup[b^\star(y),\infty)$, and the value function of the problem \eqref{eq:problem} is given by \eqref{vflast}, 
which completes the proof.
\qed
\vspace{3pt}

\begin{rmk}\label{rmk:mon} \normalfont
In view of Proposition \ref{prop:compare1}$(i)$ we know that, under Hypothesis \ref{1stop}, the mapping $y\mapsto a^\star(y)$ ($y\mapsto b^\star(y)$, resp.) is continuous and increasing (decreasing, resp.) in $y$. Moreover, we know from the proof of Proposition \ref{thm:pair} (more specifically, equation \eqref{eq:Ddev}), that the monotonicity is strict. Furthermore, because $\mc{S}^{y_m}=[\underline{k},\infty)$ (see Theorem \ref{thm:ult2}$(ii)$), we know that 
\[\lim_{y\uparrow y_m}a^\star(y)=\lim_{y\uparrow y_m}b^\star(y)=y_m.\]
\end{rmk}


\section{Proof of Theorem \ref{thm:ult3}} 
\label{sec:mart}

\noindent{\bf Proof of part (a) of Theorem \ref{thm:ult3}.}
Let us first consider the case $r<\psi(1)$, so that $(\et^{-rt+X_t})_{t\ge0}$ is a $\pr_x$-sub-martingale. 
In this case, 
we have 
\begin{align}
\ex_x\big[\et^{-A_t^y}(\et^{X_t}-K)^+\big] 
&\ge \ex_x\Big[\exp\Big(X_t-rt-q\int_{(0,t]}\ind_{\{X_s<y\}}\diff s\Big)\Big] 
- K\,\ex_x\Big[\exp\Big(-rt-q\int_{(0,t]}\ind_{\{X_s<y\}}\diff s\Big)\Big] \nn\\
&\ge \et^x\,\ex_x\Big[\exp\big((X_t-x)-\psi(1)t\big) \cdot \exp\Big((\psi(1)-r)t-q\int_{(0,t]}\ind_{\{X_s<y\}}\diff s\Big)\Big] - K \nn\\
&= \et^{x+(\psi(1)-r)t}\,\ex_x^1\Big[\exp\Big(-q\int_{(0,t]}\ind_{\{X_s<y\}}\diff s\Big)\Big] - K \,, \label{eq:psi>1}
\end{align}
Since $X$ drifts to $\infty$ under $\pr_x^1$ (see \cite[(8.3)]{Kyprianou2006} for $c=1$), we know from   \cite[Corollary 3]{OccupationInterval} that the expectation in the last line of \eqref{eq:psi>1} remains bounded as $t\to\infty$. However, the prefactor $\et^{x+(\psi(1)-r)t}\to\infty$ as $t\to\infty$, so we know that the value function $v(x;y)=\infty$ for all $x,y\in\R$.

\vspace{2pt}
\noindent{\bf Proof of part (b) of Theorem \ref{thm:ult3}.}
We now assume $r=\psi(1)$, so that $(\et^{-rt+X_t})_{t\ge0}$ is a $\pr_x$-martingale. 
In this case, recall from \cite{mordecki2002} that $\ol{v}(x)\equiv
v(x;-\infty)$ defined in \eqref{eq:overlinev} takes the form $\ol{v}(x)=\et^x$, so we know that $v(x;y)$ is finite, due to Proposition \ref{prop:compare1}. 
On the other hand, we also have 
that $v(x;y)\geq \underline{v}(x)$ holds, where $\underline{v}(x)\equiv v(x;\infty)$ is given by \eqref{eq:underlinev} as in the previous section since $r+q>\psi(1)$. 
Therefore, recalling from \cite{mordecki2002} that the problem \eqref{eq:overlinev} for $\ol{v}(x)$ has no optimal stopping region, we get using Proposition \ref{prop:compare1}$(i)$ that $\emptyset=\mc{S}^{-\infty}\subseteq \mc{S}^{y}\subseteq \mc{S}^{\infty}=[\underline{k},\infty)$, so $(-\infty,\underline{k})$ should always belong to the continuation region. 
In view of this observation, we can treat $\underline{v}(\cdot)$ as the reward function and consider the representation \eqref{eq:problem00} of the value function $v(\cdot;y)$. 
We thus have for all $x\ge \underline{k}\vee y$ (see also 
\eqref{eq:57} for the definition of $\chi(\cdot)$), that
\be
(\mc{L}-(r+q\ind_{\{x<y\}}))\underline{v}(x)=\chi(x)=rK+f(x)\ge rK>0.\nn
\ee
Hence, it will always be beneficial to continue as long as $X$ stays over $\underline{k}\vee y$. In view of this and the fact that $\mc{S}^y\subset[\underline{k},\infty)$, we either have $\mc{S}^y=\emptyset$ if $y<\underline{k}$, or $\mc{S}^y\subset[\underline{k},y]$ if $y\geq \underline{k}$. 
By Proposition \ref{prop:compare1}$(iii)$, if the stopping region exists, then it is at most one connected interval of the form $[\underline{k},a(y)]$ for some $a(y)\in[\underline{k},y]$. Thus, it will be crucial to identify the smallest $y$-value such that $\underline{k}\in\mc{S}^y$. 
To this end, we calculate the value of waiting forever:
\begin{align}
V_\infty(x;y)
=&\lim_{t\to\infty}\ex_x\big[\et^{-A_t^y}(\et^{X_t}-K)^+\big] 
\ge\lim_{t\to\infty}\ex_x\big[\et^{-A_t^y}(\et^{X_t}-K)\big]\nn\\
=&\et^x\,\ex_x^1\Big[\exp\Big(-q\int_{(0,\infty)}\ind_{\{X_s<y\}}\diff s\Big)\Big]
=\psi'(1)\,\big(\Phi(r+q)-1\big)\,\et^{y}\,\mc{I}^{(r,q)}(x-y) , \nn
\end{align}
where we used  \cite[Corollary 3]{OccupationInterval}. One can similarly establish the reverse inequality, so the above is in fact an equality. 
It follows that 
\begin{align}
V_\infty(x;y)
=\frac{(\Phi(r+q)-1)}{\Phi'(r)q}\et^{y+\prq(x-y)}, \quad \text{for all } \;x<y .
\label{eq:52}
\end{align}
For any fixed $x$, the function $V_\infty(x;\cdot)$ is obviously strictly decreasing over $(x,\infty)$, since $\Phi(r+q)>\Phi(r)=1$ in this case. 
In particular, we notice that there exists a unique value $y_\infty\in(\underline{k},\infty)$ that solves 
\be
V_\infty(\underline{k};y_\infty)=\et^{\,\underline{k}}-K \quad \Leftrightarrow \quad 
y_\infty=\underline{k}+\frac{1}{\prq-1}\log\bigg(\frac{\prq(\Phi(r+q)-1)}{\Phi'(r)q}\bigg). \label{y_inf}
\ee
It can be verified using the convexity of $\psi$ (see e.g.  \cite[Exercise 3.5]{Kyprianou2006}), that $y_\infty>\underline{k}$ indeed holds. 
Moreover, for any $x\in[\underline{k},y_\infty)$, we have
\[\p_x\log\bigg(\frac{\et^x-K}{V_\infty(x;y_\infty)}\bigg)=\frac{\ol{\Lambda}(x-y_\infty)}{\et^x-K}\bigg(K-\et^x\bigg(1-\frac{1}{\ol{\Lambda}(x-y_\infty)}\bigg)\bigg)=\frac{\prq}{\et^x-K}\bigg(K-\frac{\prq-1}{\prq}\et^x\bigg).\]
Hence, we know that the mapping $x\mapsto (\et^x-K)/V_\infty(x;y_\infty)$ is strictly decreasing over $[\underline{k},y_\infty)$. 
As a result, we know that $\mc{S}^{y}=\emptyset$ for $y<y_\infty$, $\mc{S}^{y_\infty}=\{\underline{k}\}$ and the value function is given by \eqref{vinf} in both cases.

For $y>y_\infty$, in view of Proposition \ref{prop:compare1}$(i)$ and Proposition \ref{prop:compare1}$(iii)$, we know that the stopping region $\mc{S}^y=[\ul{k},a_\infty^\star(y)]$ for some $a_\infty^\star(y)\in(\ul{k},y)$. To determine $a_\infty^\star(y)$, we consider the value of the threshold type strategy $T_a^-$. This can be done by taking the limit of \eqref{V2V22Fab} in Proposition \ref{prop:V2value} as $b\to\infty$. 
To this end, we begin by using
\eqref{newW} and 
a standard application of the dominated convergence theorem, 
to prove that (recall that $\Phi(r)=1$)
\be
\label{limW}
\lim_{b\to\infty}\et^{-b}W^{(r,q)}(b,a)=\Phi'(r)\,T^{(r,q)}(a),
\text{ where } 
T^{(r,q)}(a):=\et^{-a}+q\int_{(\ul{k},y)}\et^{-z}W^{(r+q)}(z-a)\diff z\,.
\ee
Obviously, $T^{(r,q)}(\cdot)$ is a continuously differentiable, strictly decreasing, positive function over $(\ul{k},y)$, whose domain can be extended over $[y,\infty)$ by setting $T^{(r,q)}(a)=\et^{-a}$ for $a\in[y,\infty)$. 
Therefore, for any fixed $a\in(\ul{k},y)$ and $x,w\in(a,\infty)$, we have 
\begin{align}
u^{(r,q)}(x,w;y,a,\infty):=\lim_{b\to\infty}u^{(r,q)}(x,w;y,a,b)=&\ind_{\{w\in(a,y)\}}\bigg(\frac{T^{(r,q)}(w)}{T^{(r,q)}(a)}W^{(r,q)}(x,a)-W^{(r,q)}(x,w)\bigg)\nn\\
&+\ind_{\{w\in[y,\infty)\}}\bigg(\frac{\et^{-w}}{T^{(r,q)}(a)}W^{(r,q)}(x,a)-W^{(r)}(x-w)\bigg).\nn
\end{align}
Taking into account the above expression along with \eqref{newW} and \eqref{limW}, we get the limit of \eqref{V2V22Fab} in Proposition \ref{prop:V2value} as $b\to\infty$, for any fixed $x\in(a,y)$, given by
\begin{align}
v(x;y) =\, &\ex_x[\exp(-A_{T_a^-}^y)\ul{v}(X_{T_a^-})]=\ul{v}(x) + \int_{(a,\infty)}(\chi(w)-q\ind_{\{w<y\}}\ul{v}(w))\cdot u^{(r,q)}(x,w;y,a,\infty)\,\diff w\nn\\
=\, &\ul{v}(x) + \frac{W^{(r,q)}(x,a)}{T^{(r,q)}(a)} \left( 
\int_{(a,y)} T^{(r,q)}(w) \, \big( \chi(w)-q\ul{v}(w) \big) \diff w 
+ \int_{[y,\infty)} \et^{-w} \, \chi(w) \diff w 
\right) \nn\\
&- \int_{(a,y)}W^{(r,q)}(x,w)\cdot(\chi(w)-q\ul{v}(w))\diff w-\int_{[y,\infty)}W^{(r)}(x-w)\cdot\chi(w) \diff w.\label{eq:542}
\end{align}
Similar as in Corollary \ref{cor:fit}, we know that the optimal threshold $a_\infty^\star(y)\in(\ul{k},y)$ must be such that the coefficient of $W^{(r,q)}(x,a)$ in \eqref{eq:542} vanishes, i.e. it must solve 
\be
\Delta_\infty(a_\infty^\star(y);y)=0, \text{ where }
\Delta_\infty(a;y):= 
\int_{(a,y)} T^{(r,q)}(w) \, \big( \chi(w)-q\ul{v}(w) \big) \diff w 
+ \int_{[y,\infty)} \et^{-w} \, \chi(w) \diff w 
\label{eq:a_infty}
\ee
To finish the proof, we demonstrate that there is at most one solution to \eqref{eq:a_infty}, which follows from the fact that for all $a\in(\ul{k},y)$, we have
\begin{align}
\frac{1}{T^{(r,q)}(a)}\p_a\Delta_\infty(a;y)
&= q\ul{v}(a)-\chi(a) > q\ul{v}(\ul{k})-\chi(\ul{k}) 
= \frac{1}{2}\sigma^2\prq K\ge0, 
\end{align}
where the last equality is due to \eqref{eq4f}.
In conclusion, taking into account the expression of the value function $v$ in \eqref{eq:542} for $a=a_\infty^\star(y)$, and combining it with \eqref{eq:a_infty} and \eqref{D}, we conclude that $v(x;y)$ is given by 
\eqref{vfainf} 
with $\psi(1)=r$.

\section{Example}\label{sec:gbm}

In this section we study the optimal stopping problem \eqref{eq:problem0} using the compound Poisson model, also used in \cite{OccupationInterval}. 
In particular, we assume that the L\'evy process $X$ is given by
\be X_t=X_0+\gamma t+\sigma B_t-\sum_{i=1}^{N_t}Y_t,\nn\ee
where 
$\gamma=0.3$,
$\sigma=0.2$, 
$B\equiv(B_t)_{t\ge0}$ is a standard Brownian motion, 
$(N_t)_{t\ge0}$ is a Poisson process with intensity $\lambda=0.6$ independent of $B$ and $Y_1, Y_2,\ldots$ are i.i.d. positive hyper-exponential random variables with density given by $p(z)=\ind_{\{z>0\}}\eta\et^{-\eta z}$,
where $\eta=1$ is the intensity of the exponential distribution. 
We also let $r=0.05$, $q=1$ and $K=10$.
Using this data we can compute that  
$\bar{u}=0.1665$, $\psi(1)=0.02$, $\tilde{y}=2.7870$ and $y_m=3.7383$, which correspond to Theorem \ref{thm:ult2}. 
\begin{figure}
\centering
\subfloat[$y<\tilde{y}$]{{\includegraphics[width=140pt]{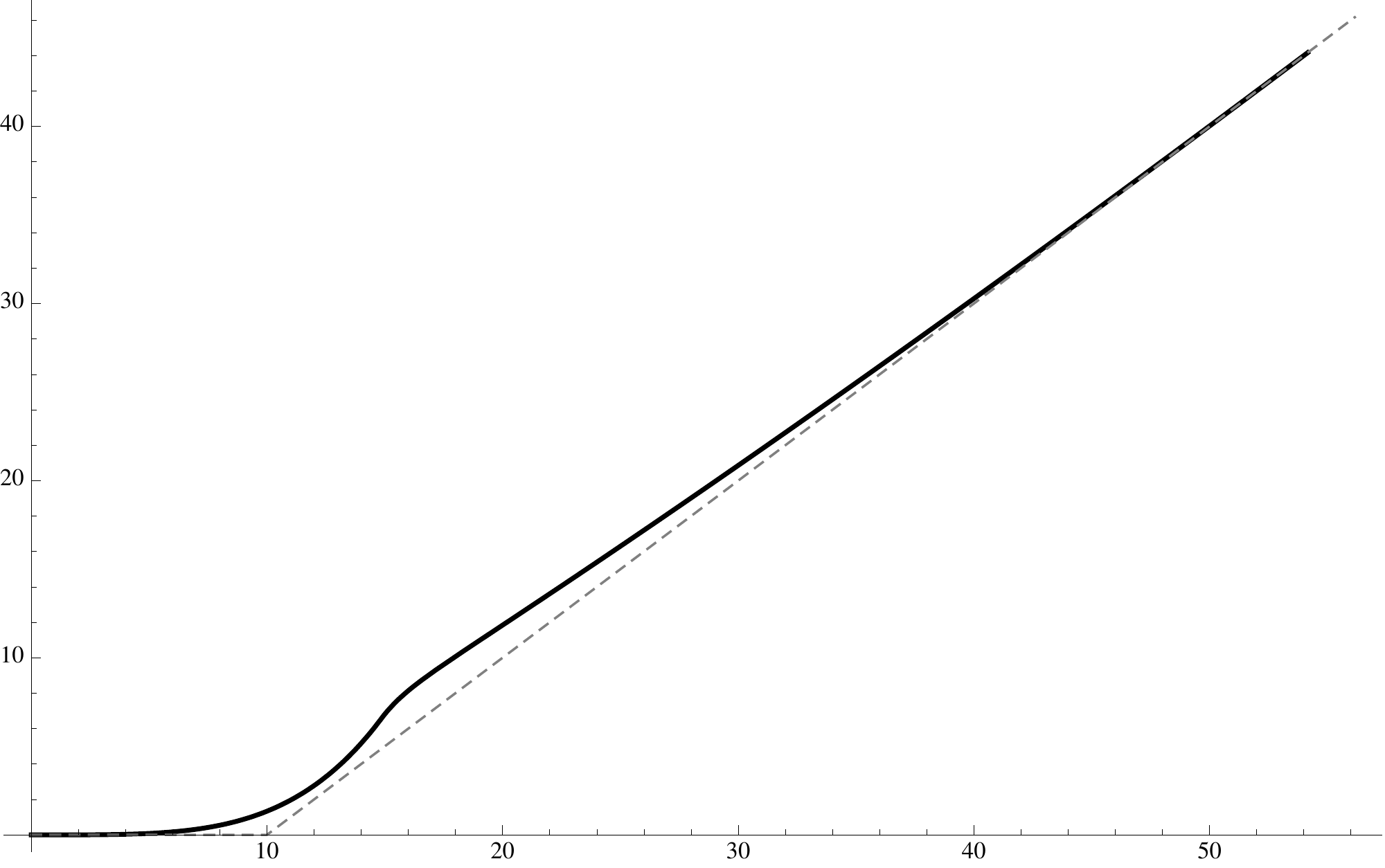}}}\quad\subfloat[$y=\tilde{y}$]{{\includegraphics[width=140pt]{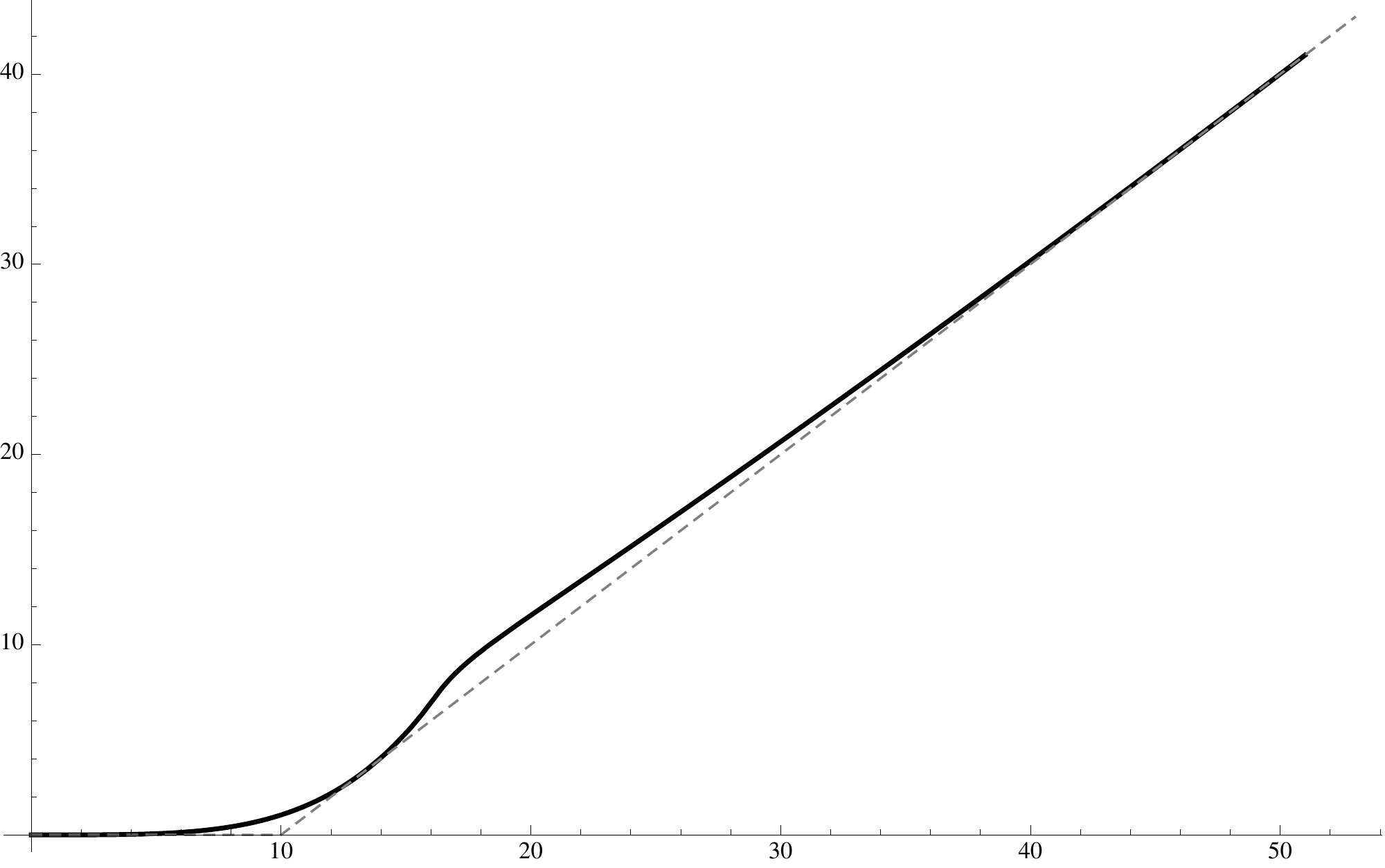}}}\quad
\subfloat[$y\in(\tilde{y},y_m)$]{{\includegraphics[width=140pt]{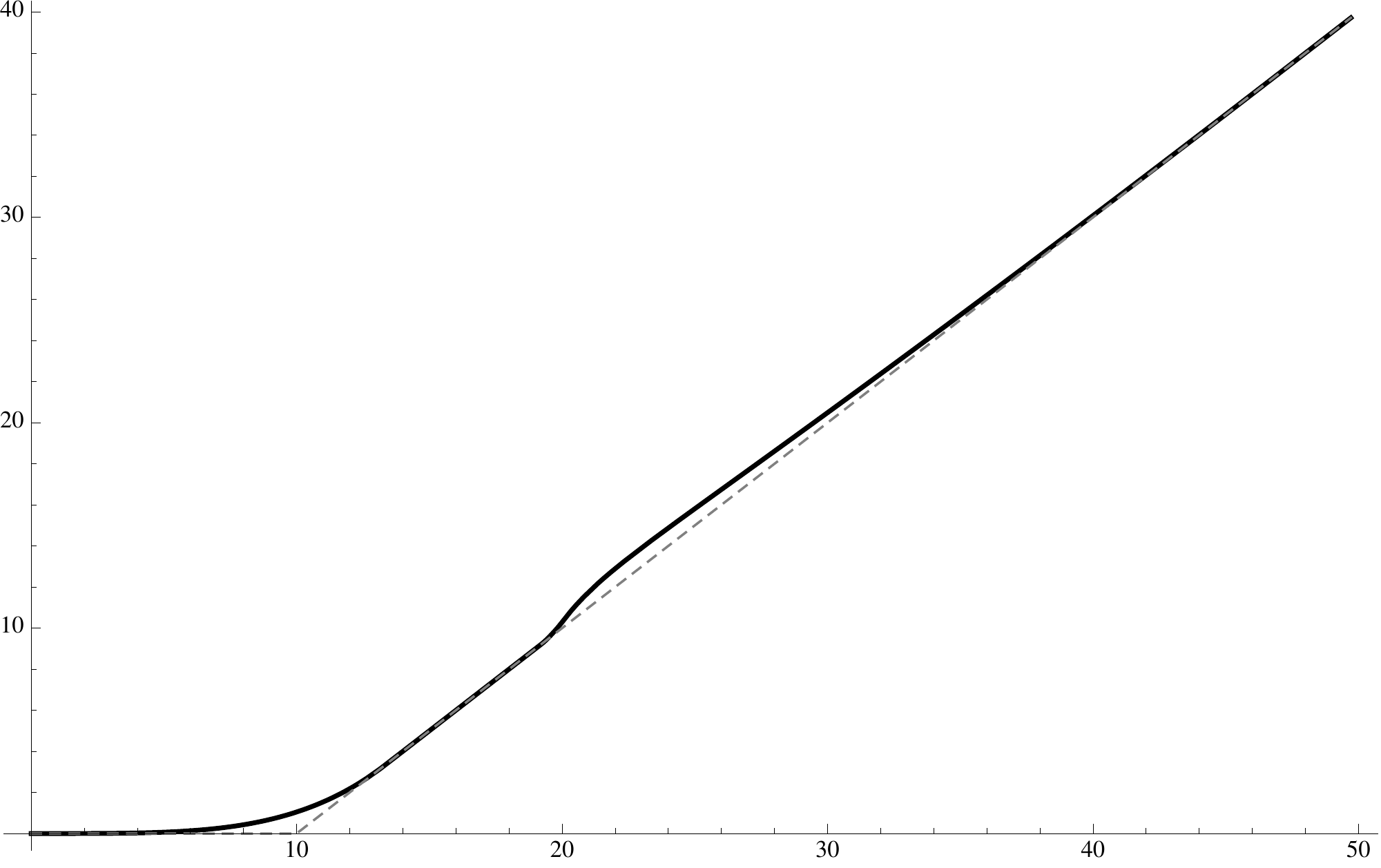}}}
\caption{{\it \underline{Illustrations of $v(x;y)$.}} 
In all panels, the horizontal axis indicates the underlying price level $\et^x$ (in the absolute price scale) and the vertical axis indicates the value; 
Black solid line: the value function $v(x;y)$; 
gray dashed line: the intrinsic value function $(\et^x-K)^+$.
Panel (a): 
for $y=2.7<\tilde{y}=2.7870$, 
the optimal stopping region is $\exp(\mc{S}^y)=[54.1636,\infty)$. 
Panel (b): 
for $y=\tilde{y}=2.7870$, 
the optimal stopping region is  
$\exp(\mc{S}^y)=\{13.3081\}\cup[50.9883,\infty)$. 
Panel (c): 
$y=3\in(\tilde{y},y_m)\equiv (2.7870, 3.7383)$, the optimal stopping region is 
$\exp(\mc{S}^y)=[13.3081,19.1801]\cup[47.7146,\infty)$.}
\label{fig1}
\end{figure}

We thus consider three values for $y$ representing the cases: $(i)$ $y<\tilde{y}$; $(ii)$ $y=\tilde{y}$; $(iii)$ $y\in(\tilde{y},y_m)$, which are illustrated in three panels in Figure \ref{fig1}. 
In panels (a) and (b) of Figure \ref{fig1},  we obtain the optimal threshold $z^\star(y)$ through its definition \eqref{z*}, while 
in panel (c) of Figure \ref{fig1}, we numerically solve for the optimal thresholds $a^\star(y)$ and $b^\star(y)$ through their definition \eqref{a*b*} in Proposition \ref{thm:pair}. 
From these plots, it is seen that the value functions are no longer convex in the absolute price scale. In addition, the optimal exercising region $\exp(\mc{S}^y)$ indeed ``grows'' with increasing $y$-values (as indicated in Proposition \ref{prop:compare1}$(i)$), there is branching (as indicated in Theorem \ref{thm:ult2}) and we can compare it in each of the above cases with $\exp(\mc{S}^{\pm \infty})$, given by $[\underline{K},\infty)=[13.3081,\infty)$ and $[\overline{K},\infty)=[77.7536,\infty)$.

\section{Conclusions}\label{sec:con}
We have studied an optimal stopping problem with a random time-horizon that depends on the occupation time of the underlying process. The problem is equivalent to the evaluation of a perpetual American call option with a random discount rate. To the best of our knowledge, our work is the first that addresses optimal stopping with a random discount rate under such general L\'evy models. Moreover, the results reveal a rich class of optimal stopping strategies. 
As seen in Theorem \ref{thm:ult1} and Theorem \ref{thm:ult2}$(i)-(ii)$, up-crossing strategies may still be optimal under certain circumstances. However, in some cases as in Theorem \ref{thm:ult2}$(iii)-(iv)$, both the optimal stopping region and the continuation region can have two disconnected components, and the optimal exercising strategy can be two-sided and may not be a one-shot scheme if overshoot occurs; which are both interesting new features. 
Lastly, because of the random discount factor, there are non-trivial optimal exercising strategies in the martingale case, as opposed to the standard perpetual American call option \cite{mordecki2002}. Precisely, as seen in Theorem \ref{thm:ult3}, the optimal stopping region is always connected, but the one-sided optimal exercising strategy can be either an up-crossing or a down-crossing exit, with the latter even possibly not being a one-shot scheme, which are all surprising results.

In order to characterize the optimal exercising thresholds in Theorem \ref{thm:ult2}$(iv)$ and Theorem \ref{thm:ult3}, we obtain the joint distribution of the discounting factor and the underlying process at the first exit time, when the discount rate is random (see Proposition 
\ref{prop:V2value}). This novel result can be applied to solve optimal stopping problems with alternative payoffs under that random discount rate. These could be natural directions for future research.

\appendix

\section{Proofs}\label{sec:prof}

\begin{proof}[Proof of Lemma \ref{lem:equivalent}]
By Proposition \ref{prop:compare1}$(i)$, the continuation region must include $(-\infty,\underline{k})$. Thus it remains to study the possibility of other parts of the continuation region in $[\underline{k},\infty)$. To that end, 
we first notice from Proposition \ref{prop:compare1}$(iii)$ that there cannot exist any component of continuation region that lies strictly in the interior of the set $[\underline{k},y]$, where the property  $(\mc{L}-(r+q\ind_{\{x<y\}}))\underline{v}(x)<0$ holds due to the calculation \eqref{eq:supervl}.  
Therefore, the only two possibilities are either for the whole set $[\underline{k},y]$ to be part of the continuation region, or a subset $(k_0,y]$ of it for some $k_0\in(\underline{k},y)$. 

Hence, in what follows we let $(a,b)$ be a maximal component of the continuation region, such that $-\infty\leq a < y < b < \infty$ and $b>\underline{k}$ (if there is no such $b$, then in light of Proposition \ref{prop:compare1}$(i)$, we must have $[\underline{k},\infty)\subset\mc{S}^y$). 
Then it holds that 
\be
(\mc{L}-(r+q\ind_{\{x<y\}}))\underline{v}(x)|_{x=b}\le 0, \label{eq:monot}
\ee
for otherwise, it will not be beneficial (by It\^o-L\'evy lemma) to wait until $X$ reaches a level where the above inequality 
does not hold. 
At this point, we notice that (using 
\eqref{eq:57})
\be
(\mc{L}-(r+q\ind_{\{x<y\}}))\underline{v}(x)=\chi(x)-q\ind_{\{x<y\}}\underline{v}(x), \quad\forall x\in(\underline{k},\infty).
\label{eq:monot2}
\ee
Using the monotonicity of $f(\cdot)$ (and hence $\chi(\cdot)$) from Lemma \ref{lem:ffun00},  
we can conclude that $(\mc{L}-(r+q\ind_{\{x<y\}}))\underline{v}(x)$ is strictly decreasing for all $x\in(\underline{k},\infty)\backslash\{y\}$.
Combining this with \eqref{eq:monot}, we know that 
\[
(\mc{L}-(r+q\ind_{\{x<y\}}))\underline{v}(x)<0,\quad\forall x>b.
\]

If Hypothesis \ref{1stop} holds, then there cannot be any continuation region in $(b,\infty)$, which means that $[b,\infty)$ is the only component of stopping region that lies on the right of $y$. 
If Hypothesis \ref{1stop} does not hold, then from the above discussion we know that, there exists a component of continuation region that lies on the right of $b$. Proposition \ref{prop:compare1}$(i)$ implies in this case that there are two components of stopping region on the right of $y$.
\end{proof}

\vspace{3pt}
\begin{prop}\label{prop:1stop}
Suppose the value function $v(\cdot;y)$ is sufficiently smooth, then Hypothesis \ref{1stop} holds if the jump tail measure,  $\overline{\Pi}(x)\equiv\Pi(-\infty,-x)$ for $x>0$, has a monotone density, i.e. for all $x>0$, $\bar{\Pi}(x)=\int_{(-\infty,-x)}\pi(y)\diff y$ and $\pi(-x)$ is non-increasing for $x\in(0,\infty)$.
\end{prop}
\begin{proof}
By Proposition \ref{prop:compare1}$(i)$, the continuation region must include $(-\infty,\underline{k})$. Thus it remains to study the possibility for other continuation regions in $(\underline{k},\infty)$. To that end, we notice from \eqref{eq:monot2} and Lemma \ref{lem:ffun00} 
that $(\mc{L}-(r+q\ind_{\{x<y\}}))\underline{v}(x)$ is strictly decreasing for all $x\in(\underline{k},\infty)\backslash\{y\}$.
By Proposition \ref{prop:compare1}$(iii)$, we know that there cannot be any component of continuation region that completely falls inside $\{x\in(\underline{k},y]: (\mc{L}-(r+q\ind_{\{x<y\}}))\underline{v}(x)<0\}$.
Therefore, we only need to verify the assertion of Hypothesis \ref{1stop} for $x\in(y\wedge\underline{k},\infty)$.

Let $(a,b)$ be the maximum component of a continuation region such that $b>y\wedge\underline{k}$ and $a\in[-\infty,b)$. 
Let $\Delta(x;y):=v(x;y)-\underline{v}(x)$, then $\Delta(x;y)>0$ for all $x\in(a,b)$. We also notice that $\Delta(x;y)=0$ for all $x\in(-\infty,a]$. This is because, if $a>-\infty$, then by Proposition \ref{prop:compare1}$(i)$ and Proposition \ref{prop:compare1}$(iii)$, we know that $a\in(\underline{k},y)$, and $[\underline{k},a]$ is the only component of the stopping region $\mc{S}^y$ that lies below $y$. Overall, $(a,b)$ is the only interval in $(-\infty,b)$ where the time value $\Delta(x;y)$ is positive. To finish the proof, we demonstrate that there is no continuation region in $(b,\infty)$.  To this end, consider the function
\[\tilde{v}(x;y):=v(x;y)\ind_{\{x\le b\}}+\underline{v}(x)\ind_{\{x>b\}}\equiv \underline{v}(x)+\Delta(x;y)\ind_{\{x\in(a,b)\}}\equiv \underline{v}(x)+\Delta(x;y)\ind_{\{x\in(-\infty,b)\}}.\] 
By construction, it holds that $\tilde{v}(x;y)\ge\underline{v}(x)\ge(\et^x-K)^+$ for all $x\in\R$, 
thus in view of Lemma \ref{verif}, we only need to show that $\tilde{v}(x;y)$ satisfies the variational inequalities \eqref{verifeq} for $\{\theta_1,\theta_2,\theta_3\}=\{\underline{k},a,b\}$, to conclude that it is the value function 
and consequently $[b,\infty)$ is a component of stopping region.  

Suppose $v(\cdot;y)$ is sufficiently smooth, then $\tilde{v}(x;y)$ is also continuously differentiable in $x$ over $\R$. Moreover, 
\be
\label{HJB0}
(\mc{L}-(r+q\ind_{\{x<y\}}))\tilde{v}(x;y)=(\mc{L}-(r+q\ind_{\{x<y\}}))v(x;y)=0,\quad\forall x\in(a,b).
\ee
Letting $x\uparrow b$ in \eqref{HJB0}, and using smooth fit of $\tilde{v}(x;y)$ at $x=b$, we have 
\be
0=\chi(b)+\half\sigma^2\p_x^2\Delta(x;y)|_{x=b}+\int_{(a-b,0)}\Delta(b+z;y)\Pi(\diff z).\label{HJB1}
\ee
If $\sigma>0$ and $\p_x^2\Delta(x;y)$ is continuous in $x$ with a finite limit as $x\uparrow b$, then because $\p_x\Delta(x;y)|_{x=b-}=0$ and $\Delta(x;y)>0$ in the small left neighborhood of $b$, we know that $\Delta(x;y)$ is convex at $x=b$ if $\sigma>0$, so \eqref{HJB1} leads to
\begin{align}
0 \ge \chi(b)+\int_{(a-b,0)}\Delta(b+z;y)\pi(z)\diff z=\chi(b)+\int_{(a,b)}\Delta(w;y)\pi(-b+w)\diff w.\label{HJB2}
\end{align}
Thanks to the monotonicity of $\chi(x)$ and $\pi(x)$, \eqref{HJB2} implies for any $x>b$, that 
\begin{align*}
0 &\ge \chi(x)+\int_{(a,b)}\Delta(w;y)\pi(-x+w)\diff w=\chi(x)+\int_{(a-x,b-x)}\Delta(x+z;y)\pi(z)\diff z\\
&= \chi(x)+\int_{(-\infty,b-x)}\Delta(x+z;y)\pi(z)\diff z=(\mc{L}-r)\tilde{v}(x;y).
\end{align*}
In summary, we know that 
\eqref{verifeq} holds true, 
thus $v(x;y)\equiv \tilde{v}(x;y)$ and the proof is completed. 
\end{proof}

\vspace{3pt}
\begin{proof}[Proof of Lemma \ref{lem:explicit}]
For all $x>0$, using integration by parts and Lemma \ref{lem W}, we have
\begin{align}\mc{I}^{(r,q),\prime\prime}(x)=&-W^{(r),\prime}(x)+\prq\mc{I}^{(r,q),\prime}(x),\label{eqII1}\\
\mc{I}^{(r,q),\prime}(x)=&-W^{(r)}(x)+\prq\mc{I}^{(r,q)}(x).\label{eqII2}\end{align}
Using again Lemma \ref{lem W} and $\mc{I}^{(r,q)}(0)=\frac{1}{q}$, we know that,
\begin{align*}
\mc{I}^{(r,q),\prime\prime}(0+)-\mc{I}^{(r,q),\prime}(0+)=&\frac{1}{q}\left((\prq-1)(\prq-qW^{(r)}(0))-qW^{(r),\prime}(0+)\right).
\end{align*}
Therefore, if the right hand side of the above is negative, \eqref{cond1} will not hold. In the sequel we demonstrate that, if $\mc{I}^{(r,q),\prime\prime}(0+)-\mc{I}^{(r,q),\prime}(0+)\ge0$,  then \eqref{cond1} must hold. 
To this end, we define  function $\mc{J}^{(r,q)}(x):=\et^{-\prq x}(\mc{I}^{(r,q),\prime\prime}(x)-\mc{I}^{(r,q),\prime}(x))$, and use \eqref{eq:Hinfnnn} to obtain that 
\[
\lim_{x\to\infty}\frac{\mc{J}^{(r,q)}(x)}{\et^{-\prq x}\mc{I}^{(r,q)}(x)}=\mc{H}(\infty)=\Phi(r)(\Phi(r)-1)>0,\quad\text{so}\quad \mc{J}^{(r,q)}(\infty)=\mc{H}(\infty)\lim_{x\to\infty}\et^{-\prq x}\mc{I}^{(r,q)}(x)=0,
\]
which suggests that $\mc{J}^{(r,q)}(x)>0$ for all sufficiently large $x>0$, and there must be an increasing sequence $(x_n)_{n\ge1}$ that goes to $\infty$, such that $\mc{J}^{(r,q)\prime}(x_n)<0$ (in order for $0<\mc{J}^{(r,q)}(x_n)\searrow \mc{J}^{(r,q)}(\infty)=0$).
Let us suppose that \eqref{cond1} is violated, 
then there must be some  global minimum $x_0\in(0,\infty)$, such that $\mc{J}^{(r,q),\prime}(x_0)=0$ and $\mc{J}^{(r,q)}(x_0)<0$. However, as $x$ increases, $\mc{J}^{(r,q)}(x)$ must be decreasing below 0 and then increasing over 0, and then ultimately decreasing at least over the points in the sequence $(x_n)_{n\ge1}$.  
However, from \eqref{eqII1} and \eqref{eqII2} we know that 
\begin{align}
\mc{J}^{(r,q),\prime}(x)=\et^{-\prq x}\left(W^{(r),\prime}(x)-W^{(r),\prime\prime}(x)\right).
\end{align}
And by Lemma \ref{lem:CM} below, the only possible sign change for $\mc{J}^{(r,q),\prime}(x)$ is from being positive to being negative.  Therefore, the above assumed $x_0$ does not exist, and \eqref{cond1} must hold.

Finally, under the assumptions about the tail jump distribution, it is clear that we have a monotone jump density. So by Proposition \ref{prop:1stop}, we know that Hypothesis \ref{1stop} holds, and results in Theorem \ref{thm:ult1} and Theorem \ref{thm:ult2} hold.
This completes the proof. 
\end{proof}

\vspace{3pt}
\begin{lem}\label{lem:CM}
Suppose $W^{(r)}(\cdot)\in C^2(0,\infty)$ and the tail jump measure of $X$, denoted by $\overline{\Pi}(x):=\Pi(-\infty,-x)$ for $x>0$, either has a completely monotone density or is  log-convex. Then, as $x$ increases from 0 to $\infty$, the only possible sign change of function $W^{(r),\prime}(x)-W^{(r),\prime\prime}(x)$ is from being positive to negative, which can happen at most once. 
\end{lem}
\begin{proof}
If $\overline{\Pi}(x)$ has a complete monotone density, then we know from the proof of  \cite[Theorem 3.4]{Kuznetsov_2011} that, there exist constants $a,b\ge0$ and a measure $\xi$ concentrated on $(0,\infty)$, satisfying $\int_{(0,\infty)}(1\wedge t)\,\xi(\diff t)<\infty$, such that
\[W^{(r)}(x)=\et^{\Phi(r)x}\bigg(a+bx+\int_{(0,\infty)}(1-\et^{-xt})\,\xi(\diff t)\bigg).\]
It is then straightforward to show that
\begin{align}
W^{(r)\prime\prime}(x)-W^{(r)\prime\prime\prime}(x)&=u''(x)-u'''(x)+\int_{(0,\infty)}\big(H(\Phi(r),x)-H(\Phi(r)-t,x)\big)\,\xi(\diff t)\nn
\end{align}
where $u(x)=(a+bx)\et^{\Phi(r)x}$ and 
\[H(t,x)=t^2\,(1-t)\,\et^{tx}, \quad \forall x>0 \quad \forall  t\in\mathbb{R}
\quad \Rightarrow \quad 
\partial_t H(t,x)
=t\,\et^{tx}\,G(t,x)\,,\]
where for each fixed $x>0$, $G(\cdot,x)$ is a quadratic function given by 
\[G(t,x)=-xt^2+(x-3)t+2.\]
Since $G(0,x)=2$ and $G(1,x)=-1<0$, we know that $G(t,x)=0$ has two roots $t_1$ and $t_2$, satisfying $t_1<0<t_2<1$. Overall, we see that the function $H(\cdot,x)$ is first strictly increasing over $(-\infty,t_1)$, strictly decreasing over $(t_1,0)$, then again strictly increasing over $(0,t_2)$, and finally strictly decreasing over $(t_2,\infty)$. 
Moreover, it is easily seen that  
\[H(-\infty,x)=H(0,x)=H(1,x)=0\,,\]
which therefore implies (since $\Phi(r)>1$) that 
\[H(t,x)>H(\Phi(r),x), \quad \forall t<\Phi(r) 
\quad \Rightarrow \quad 
H(\Phi(r),x)-H(\Phi(r)-t,x)<0\,, \quad \forall t>0. 
\]
Moreover, 
\begin{align*}
u''(x)-u'''(x)=
\et^{\Phi(r)x}\Phi(r)[b(2-3\Phi(r))+(a+bx)\Phi(r)(1-\Phi(r))]<0.\nn
\end{align*}
Hence, we see that the function $x\mapsto W^{(r)\prime}(x)-W^{(r)\prime\prime}(x)$ is strictly decreasing over $(0,\infty)$. So this function is either  does not change sign, or changes sign from being positive to being negative exactly once.

If instead $\overline{\Pi}(x)$ is log-convex, we know from  \cite[Theorem 3.5]{Kuznetsov_2011} that the function $l(x):=\log W^{(r)\prime}(x)$ is  convex over $(0,\infty)$. It follows that $l'(x)=W^{(r)\prime\prime}(x)/W^{(r)\prime}(x)$ is increasing. Hence the sign of $W^{(r)\prime}(x)[1-l'(x)]$ will either have the same sign throughout $(0,\infty)$ or change from being positive to negative, as $x$ increases over $(0,\infty)$. 
\end{proof}

\vspace{3pt}
\begin{proof}[Proof of Proposition \ref{prop:compare1}]
$\, $

\ul{\it Part $(i)$:} 
For any fixed $x,y\in\R$, using the definition (\ref{eq:underlinev}) of $\underline{v}(x) \equiv v(x;\infty)$ and assuming that $\underline{\tau}$ is an optimal stopping time (e.g. $\underline{\tau}=T_{\underline{k}}^+$, if $r>\psi(1)$) for the problem (\ref{eq:underlinev}), we have that
\begin{align}
\underline{v}(x) &= 
\sup_{\tau\in\timset_0} \ex_x[\exp(-(r+q) \tau)(\exp(X_{\tau})-K)^+\ind_{\{\tau<\infty\}}]   
= \ex_x[\exp(-(r+q)\underline{\tau})(\exp(X_{\underline{\tau}})-K)^+\ind_{\{\underline{\tau}<\infty\}}] \nn\\
&\le \ex_x[\exp(-A_{\underline{\tau}}^y)(\exp(X_{\underline{\tau}})-K)^+\ind_{\{\underline{\tau}<\infty\}}]
\le \sup_{\tau\in\timset_0}\ex_x[\exp(-A_\tau^y)(\exp(X_\tau)-K)^+\ind_{\{\tau<\infty\}}]=v(x;y). \nn
\end{align}
The other inequalities can be proved similarly. 

\ul{\it Part $(ii)$:}

{\bf Step 1:}
We let $v(x;y,K)\equiv v(x;y)=\sup_{\tau\in\timset}\ex_x[\et^{-A_\tau^y}(\et^{X_\tau}-K)^+\ind_{\{\tau<\infty\}}]$ be the value function. Then for any $0<K_1<K_2<\infty$ and any $\tau\in\timset$ and $x\in\R$, we have 
\[\ex_x[\et^{-A_\tau^y}(\et^{X_\tau}-K_1)^+\ind_{\{\tau<\infty\}}]\ge \ex_x[\et^{-A_\tau^y}(\et^{X_\tau}-K_2)^+\ind_{\{\tau<\infty\}}].\]
Using the above and Proposition \ref{prop:compare1}  we know that the mapping $(y,K)\mapsto v(x;y,K)$ is jointly decreasing on $\R\times\R_+$. 
Moreover, for any $\tau\in\timset$ and any $x\in\R$, we have
\be\ex_x[\et^{-A_\tau^y}(\et^{X_\tau}-K)^+\ind_{\{\tau<\infty\}}]=\et^x \,\ex_0[\exp(-A_\tau^{y-x})(\exp(X_{\tau})-K\et^{-x})^+\ind_{\{\tau<\infty\}}],\nn\ee
which implies that
\[v(x;y,K)=\et^x v(0;y-x,K\et^{-x}).\]
Using this expression along with the non-negativity and monotonicity of $v$ in parameters $y$ and $K$, we see that the mapping $x\mapsto v(x;y)$ is strictly increasing (due to the exponential factor).  

{\bf Step 2:}
Taking into account the above expressions and Proposition \ref{prop:compare1}$(i)$, we can rewrite problem \eqref{eq:problem0} in the form:
\begin{align}
v(x;y)=&\sup_{\tau\in\timset_0}\ex_x[\exp(-A_{\tau\wedge T_{\overline{k}}^+}^y)(\exp(X_{\tau\wedge T_{\overline{k}}^+})-K)^+\ind_{\{\tau<\infty\}}]\nn\\
=&\sup_{\tau\in\timset_0}\ex_0[\exp(-A_{\tau\wedge T_{\overline{k}-x}^+}^{y-x})(\exp(x+X_{\tau\wedge T_{\overline{k}-x}^+})-K)^+\ind_{\{\tau<\infty\}}]\label{eq:A1}.
\end{align}
We now consider $x<x'$ and express similarly $v(x';y)$ as
\begin{align}
v(x';y)=&\sup_{\tau\in\timset_0}\ex_{x'}[\exp(-A_{\tau\wedge T_{\overline{k}}^+}^y)(\exp(X_{\tau\wedge T_{\overline{k}}^+})-K)^+\ind_{\{\tau<\infty\}}]\nn\\
=&\sup_{\tau\in\timset_0}\ex_0[\exp(-A_{\tau\wedge T_{\overline{k}-x'}^+}^{y-x'})(\exp(x'+X_{\tau\wedge T_{\overline{k}-x'}^+})-K)^+\ind_{\{\tau<\infty\}}]\nn\\
=&\sup_{\tau\in\timset_0}\ex_0[\exp(-A_{\tau\wedge T_{\overline{k}-x}^+}^{y-x'})(\exp(x'+X_{\tau\wedge T_{\overline{k}-x}^+})-K)^+\ind_{\{\tau<\infty\}}],\label{eq:A2}
\end{align}
where the last equality follows from the facts that $\pr_0(T_{\overline{k}-x'}^+\le T_{\overline{k}-x}^+)=1$ and that $T_{\overline{k}-x'}^+$ never occurs before the optimal stopping time $\tau^\star:=\inf\{t>0: X_t\in\mc{S}^y\}$, see Proposition \ref{prop:compare1}$(i)$.
By fixing a stopping time $\tau\in\timset_0$, we have in the event that $\{\tau<\infty\}$, that ($\pr_0$-a.s.):
\begin{align}
&\exp\big(-A_{\tau \wedge T_{\overline{k}-x}^+}^{y-x'}\big) \Big(\et^{x'+X_{\tau\wedge T_{\overline{k}-x}^+}}-K \Big)^+ 
- \exp\big(-A_{\tau\wedge T_{\overline{k}-x}^+}^{y-x}\big) \Big(\et^{x+X_{\tau\wedge T_{\overline{k}-x}^+}}-K \Big)^+ \nn\\
&= \left(\exp\big(-A_{\tau \wedge T_{\overline{k}-x}^+}^{y-x'}\big)
-\exp\big(-A_{\tau\wedge T_{\overline{k}-x}^+}^{y-x}\big)\right) 
\Big(\et^{x'+X_{\tau\wedge T_{\overline{k}-x}^+}}-K\Big)^+ \nn\\
&\quad+\exp\big(-A_{\tau \wedge T_{\overline{k}-x}^+}^{y-x}\big) 
\left(\Big(\et^{x'+X_{\tau\wedge T_{\overline{k}-x}^+}}-K\Big)^+ 
-\Big(\et^{x+X_{\tau\wedge T_{\overline{k}-x}^+}}-K\Big)^+\right)\nn\\
&\le \Big|\exp\big(-\omega_{\tau\wedge T_{\overline{k}-x}^+}^{y-x}\big)
-\exp\big(-\omega_{\tau\wedge T_{\overline{k}-x}^+}^{y-x'}\big)\Big| \, \exp\Big(x'-r(\tau\wedge T_{\overline{k}-x}^+)+X_{\tau\wedge T_{\overline{k}-x}^+}\Big) \nn\\
&\quad+ (\et^{x'}-\et^x) \exp\Big(-A_{\tau\wedge T_{\overline{k}-x}^+}^{y-x}+X_{\tau\wedge T_{\overline{k}-x}^+}\Big) \nn\\
&\le \et^{(\overline{k}-x)^+} \left(\et^{x'} \, \Big|\exp\big(-\omega_{\tau\wedge T_{\overline{k}-x}^+}^{y-x}\big)-\exp\big(-\omega_{\tau\wedge T_{\overline{k}-x}^+}^{y-x'}\big)\Big| + \et^{x'} - \et^x \right),\label{eq:vxv0}
\end{align}
where we used the inequalities $(S'-K)^+<S'$ and $(S'-K)^+-(S-K)^+\le S'-S$ for $0<S<S'$ in the first inequality, and the $\pr_0$-a.s. inequality $-r(\tau\wedge T_{\overline{k}-x}^+)+X_{\tau\wedge T_{\overline{k}-x}^+}\le(\overline{k}-x)^+$ in the last inequality.
Then, by the inequality $\sup_n(a_n+b_n)\le\sup_n a_n+\sup_n b_n$, \eqref{eq:A1}, \eqref{eq:A2} and \eqref{eq:vxv0} we have
\begin{align}
0\le& v(x';y)-v(x;y)\nn\\
\le&\et^{(\overline{k}-x)^+}\left(\et^{x'}\sup_{\tau\in\timset_0}\ex_0\Big[\big|\exp(-\omega_{\tau\wedge T_{\overline{k}-x}^+}^{y-x})-\exp(-\omega_{\tau\wedge T_{\overline{k}-x}^+}^{y-x'})\big| \, \ind_{\{\tau<\infty\}}\Big] + \et^{x'}-\et^x \right).
\label{eqcont}
\end{align}
But for any $\tau\in\timset_0$, we have $\pr_0$-a.s. that
\begin{align}
\big|\exp(-\omega_{\tau\wedge T_{\overline{k}-x}^+}^{y-x})-\exp(-\omega_{\tau\wedge T_{\overline{k}-x}^+}^{y-x'})\big| \, \ind_{\{\tau<\infty\}}
&=\exp(-\omega_{\tau\wedge T_{\overline{k}-x}^+}^{y-x'})\left(1-\exp\big(-(\omega_{\tau\wedge T_{\overline{k}-x}^+}^{y-x}-\omega_{\tau\wedge T_{\overline{k}-x}^+}^{y-x'})\big)\right) \ind_{\{\tau<\infty\}}\nn\\
&\le \left(1-\exp\big(-(\omega_{\tau\wedge T_{\overline{k}-x}^+}^{y-x}-\omega_{\tau\wedge T_{\overline{k}-x}^+}^{y-x'})\big)\right) \ind_{\{\tau <\infty\}}\nn\\
&\le 1-\exp\big(-(\omega_{T_{\overline{k}-x}^+}^{y-x}-\omega_{T_{\overline{k}-x}^+}^{y-x'})\big) \, \ind_{\{T_{\overline{k}-x}^+<\infty\}} \le 1, \label{eq18}
\end{align}
where we used, in the last inequality, the fact that $$t\mapsto(\omega_t^{y-x}-\omega_t^{y-x'})=q\int_{(0,t]}\ind_{\{X_s\in[y-x',y-x)\}}\diff s$$ is increasing. 
Therefore, we have from \eqref{eq18} that
\begin{align}
0\le\Delta(x,x')&:=\ex_0\Big[\big|\exp(-\omega_{\tau\wedge T_{\overline{k}-x}^+}^{y-x})-\exp(-\omega_{\tau\wedge T_{\overline{k}-x}^+}^{y-x'})\big| \, \ind_{\{\tau<\infty\}}\Big] \nn\\
&\le 1-\ex_0\Big[\exp\big(-(\omega_{T_{\overline{k}-x}^+}^{y-x}-\omega_{T_{\overline{k}-x}^+}^{y-x'})\big) \, \ind_{\{T_{\overline{k}-x}^+<\infty\}}\Big].\nn
\end{align}
Moreover, using  \cite[Corollary 2(i) and (11)]{OccupationInterval}, it is clear by dominated convergence theorem that 
\[\lim_{x'\downarrow x}\exp(-(\omega_{T_{\overline{k}-x}^+}^{y-x}-\omega_{T_{\overline{k}-x}^+}^{y-x'}))\ind_{\{T_{\overline{k}-x}^+<\infty\}}=1\;,\;\; \pr_0-\text{a.s.} 
\]
and thus $\lim_{x'\downarrow x}\Delta(x,x')=0$.
Using this result in \eqref{eqcont}, we finally see that the mapping $x\mapsto v(x;y)$ is right-continuous in $x$. The left continuity of $x\mapsto v(x;y)$ can be proved similarly.

{\bf Step 3:} To prove the continuity in $y$, we fix
 $y>y'$. For any stopping time $\tau\in\mc{T}$, we have $\pr_x$-a.s. that
\begin{align*}
&\left(\exp(-A_{\tau\wedge T_{\overline{k}}^+}^y)-\exp(-A_{\tau\wedge T_{\overline{k}}^+}^{y'})\right)\left(\et^{X_{\tau\wedge T_{\overline{k}}^+}}-K\right)^+\ind_{\{\tau<\infty\}}\nn\\
\le&\exp(-A_{\tau\wedge T_{\overline{k}}^+}^{y'})\left|\exp(-q\int_{]0,\tau\wedge T_{\overline{k}}^+]}\ind_{\{X_s\in[y',y)\}}\diff s)-1\right|\et^{X_{\tau\wedge T_{\overline{k}}^+}}\ind_{\{\tau<\infty\}}\nn\\
\le&\et^{x\vee\overline{k}}\bigg(1-\exp(-q\int_{]0,T_{\overline{k}}^+]}\ind_{\{X_s\in[y',y)\}}\diff s)\bigg)\ind_{\{\tau<\infty\}},
\end{align*}
which is a non-negative random variable that converges to 0 almost surely as $|y-y'|$ goes to 0. On the other hand,
\begin{align*}
0\le v(x;y')-v(x;y)\le &\sup_{\tau\in\mc{T}}\ex_x\bigg[\left(\exp(-A_{\tau\wedge T_{\overline{k}}^+}^y)-\exp(-A_{\tau\wedge T_{\overline{k}}^+}^{y'})\right)\bigg(\et^{X_{\tau\wedge T_{\overline{k}}^+}}-K\bigg)^+\ind_{\{\tau<\infty\}}\bigg]\\
\le & \et^{x\vee\underline{k}}\,\ex_x\bigg[\bigg(1-\exp(-q\int_{]0,T_{\overline{k}}^+]}\ind_{\{X_s\in[y',y)\}}\diff s)\bigg)\ind_{\{\tau<\infty\}}\bigg].
\end{align*}
From the dominated convergence theorem we know that $v(x;y')-v(x;y)$ goes to 0 as $|y-y'|$ tends to 0.

\ul{\it Part $(iii)$:} 
Because $\mc{S}^y\subset[\underline{k},\infty)$, we know that $y\ge a\ge\underline{k}$. If $a>\underline{k}$, then for any $x\in[\underline{k},a]$, we have
\begin{align}
v(x;y)=&\sup_{\tau\in\mc{T}}\ex_x[\et^{-A_\tau^y}(\et^{X_\tau}-K)^+\ind_{\{\tau<\infty\}}]\nn\\
=&\sup_{\tau\in\mc{T}}\ex_x[\et^{-A_{(\tau\wedge T_a^+)}^y}(\et^{X_{\tau\wedge T_a^+}}-K)^+\ind_{\{\tau\wedge T_a^+<\infty\}}]\nn\\
=&\sup_{\tau\in\mc{T}}\ex_x[\et^{-(r+q)(\tau\wedge T_a^+)}(\et^{X_{\tau\wedge T_a^+}}-K)^+\ind_{\{\tau\wedge T_a^+<\infty\}}]\nn\\
=&\underline{v}(x),\nn
\end{align}
where we used the facts that $X_{T_a^+}\in\mc{S}^y$ on the event $\{T_a^+<\infty\}$ in the second line, and $A_t^y=(r+q)t$ for $t\le T_a^+\le T_y^+$ in the third line. It follows that $[\underline{k},a]\subset\mc{S}^y$. 
\end{proof}

\vspace{3pt}
\begin{proof}[Proof of Lemma \ref{Lambdaprop}]
The limits as $x\downarrow0$ and $x\uparrow\infty$ follow straightforwardly from the asymptotic behaviors of scale functions, see Lemma \ref{lem W}. Thus, we only present the proof of the monotonicity. To this end, fix $x<z$, then it is clear that the mapping 
$y\mapsto \ex_x[\exp(-A_{T_{z}^+}^y)]\equiv \mc{I}^{(r,q)}(x-y)/\mc{I}^{(r,q)}(z-y)$ is 
strictly decreasing in $y$. 
Hence, for $z>x>y$ we have
\begin{align*}\p_y\bigg(\frac{\mc{I}^{(r,q)}(x-y)}{\mc{I}^{(r,q)}(z-y)}\bigg)=&\frac{\mc{I}^{(r,q)}(x-y)}{\mc{I}^{(r,q)}(z-y)}\bigg(\frac{\mc{I}^{(r,q),\prime}(x-y)}{\mc{I}^{(r,q)}(x-y)}-\frac{\mc{I}^{(r,q),\prime}(z-y)}{\mc{I}^{(r,q)}(z-y)}\bigg)\nn\\
=&\frac{\mc{I}^{(r,q)}(x-y)}{\mc{I}^{(r,q)}(z-y)}\left(\ol{\Lambda}(x-y)-\ol{\Lambda}(z-y)\right)\le 0.\end{align*}
So $\ol{\Lambda}(\cdot)$ is non-increasing over $[0,\infty)$. Suppose there is a non-empty interval $(l,u)\subset[0,\infty)$ such that $\ol{\Lambda}(l)=\ol{\Lambda}(u)$. If not, there is a $y_1<x$ such that $\ol{\Lambda}(x-y_1)=\ol{\Lambda}(z-y_1)$. Setting $m=(u+l)/2$, then we have 
\[\p_y\bigg(\frac{\mc{I}^{(r,q)}(m-y)}{\mc{I}^{(r,q)}(u-y)}\bigg)=0,\quad\forall y\in(0,\frac{u-l}{2})\]
But this means that $y\mapsto \ex_m[\exp(-A_{T_{u}^+}^y)]$ is not strictly decreasing. This is a contradiction.  Therefore $\ol{\Lambda}(\cdot)$ must be strictly decreasing over $[0,\infty)$.

\end{proof}

\vspace{3pt}
\begin{proof}[Proof of Lemma \ref{lem:Hsign2q}]
Using the definition of $\psi(\cdot)$ in \eqref{decomp} and $\prq$, we have
\[\psi(\prq)=r+q=\frac{\sigma^2}{2}(\prq)^2+\mu\prq+\int_{(-\infty,0)}(\et^{\prq z}-1-\prq z1_{\{z>-1\}})\Pi(\diff z).\]
Combining this with $\mc{H}(0+)$ from \eqref{eqHv} when $\sigma>0$ (i.e. $W^{(r)}(0)=0$ and $W^{(r),\prime}(0+) = 2/\sigma^2$ by Lemma \ref{lem W}), which is given by   
\[
\mc{H}(0+)
=\prq\left(\prq-1\right) - \frac{2q}{\sigma^2}\,,
\]
we get   
\begin{align}
\frac{\sigma^2}{2}\mc{H}(0+)
= \frac{\sigma^2}{2}\prq(\prq-1)-q 
= r-\bigg(\mu+\frac{\sigma^2}{2}\bigg)\prq-F(\prq)\,,\label{eq:H0phii}
\end{align}
where 
\begin{align}
F(x)=\int_{(-\infty,0)}(\et^{x z}-1-x z \ind_{\{z>-1\}})\Pi(\diff z)
.\nn
\end{align}
Notice that $F(\cdot)$ is (strictly if $\Pi\not\equiv0$) convex over $(0,\infty)$ and $F(0)=0$, hence $F(x)\ge F(1)x$ for all $x\ge1$. In particular, for $\prq>1$ we have  
\[\frac{\sigma^2}{2}\mc{H}(0+)\le r-\bigg(\frac{\sigma^2}{2}+\mu\bigg)\prq-F(1)\prq=r-\psi(1)\prq.\]
Hence, if $\bar{u}=-\infty$ holds, then necessarily $g'(0+)=\frac{\mc{H}(0+)}{\prq^2}\ge0$ and therefore $r\ge \psi(1)\prq$. On the other hand, if $\psi(1)>0$ and $\prq>r/\psi(1)$ hold, then $g'(0+)=\frac{\mc{H}(0+)}{\prq^2}<0$ and therefore $\bar{u}>0$. 
\end{proof}

\vspace{3pt}
\begin{proof}[Proof of Lemma \ref{func}]  
Straightforward calculation using \eqref{eq:hazard} yields that
\begin{align}
&\ex_x[\et^yg( \overline{X}_{\zeta}-y)]\nn\\=&
\int_{(x,\infty)} \et^z\bigg(1-\frac{1}{\overline{\Lambda}(z-y)}\bigg) \,\pr_x( \overline{X}_{\zeta}\in\diff z)\nn\\
=&\int_{(x,\infty)} \et^z\,\pr_x( \overline{X}_{\zeta}\in \diff z)-\int_{(x,\infty)} \et^z\,\pr_x( \overline{X}_{\zeta}>z)\diff z\nn\\
=&\int_{(x,\infty)} \et^z\pr_x( \overline{X}_{\zeta}\in \diff z)+\et^x\pr_x( \overline{X}_{\zeta}>x)-\lim_{z\uparrow\infty}\left(\et^z\ex_x[\exp(-A_{T_z^+}^y)]\right)-\int_{(x,\infty)} \et^z\pr_x( \overline{X}_{\zeta}\in\diff z)=\et^x,\nn
\end{align}
where we used integration by parts in the third equality and the regularity of $(x,\infty)$ for $x$, as well as Proposition \ref{prop:lap11} and Lemma \ref{lem W} in the last equality.
This completes the proof.
\end{proof}

\vspace{3pt}
\begin{proof}[Proof of Lemma \ref{lem:ffun00}]
Let $w'>w\ge\underline{k}$, then we have by the definition of $f$ in \eqref{eq:ffun} that
\begin{align}
f(w)-f(w')=&\int_{(-\infty,\underline{k}-w)}h(z+w)\Pi(\diff z)-\int_{(-\infty,\underline{k}-w')}h(z+w')\Pi(\diff z)\nn\\
=&\int_{[\underline{k}-w',\underline{k}-w)}h(z+w)\Pi(\diff z)+\int_{(-\infty,\underline{k}-w')}(h(z+w)-h(z+w'))\Pi(\diff z).\label{feqdiff}
\end{align}
It can be easily verified that $h(\cdot)$ defined in \eqref{h} is a strictly decreasing, positive function over $(-\infty,\underline{k})$, so $f(w)-f(w')>0$. That is, $f(\cdot)$ is strictly decreasing. Moreover, since $0\le h(\underline{k}-\e)=C\e^2$ for some fixed constant $C>0$ and all sufficiently small $\e>0$, we can use the dominated convergence theorem to show the continuity of $f(\cdot)$. 

In order to prove the continuous differentiability of $f(\cdot)$, we will first show the right-differentiability by considering the expression $(f(w')-f(w))/(w'-w)$ using \eqref{feqdiff}. To this end, we fix a $\delta>0$ and let $\underline{k}+\delta\le w<w'$. Then notice that ${\tilde h}(S):=h(\log S)$ is strictly decreasing and convex over $(0,\underline{K}\et^{-\delta}]$, with ${\tilde h}(\underline{K})={\tilde h}'(\underline{K})=0$ and $0>{\tilde h}'(S)\ge {\tilde h}'(0+)=-1$ for all $S\in(0,\underline{K}\et^{-\delta}]$. It thus follows that 
\begin{align*}
0&\ge-\int_{[\underline{k}-w',\underline{k}-w)}h(z+w)\,\Pi(\diff z)
\ge -{\tilde h}(\underline{K}\et^{w-w'})\,\Pi[\underline{k}-w',\underline{k}-w)\\
&\ge \underline{K}\,(1-\et^{w-w'})\,{\tilde h}'(\underline{K}\et^{w-w'})\,\Pi[\underline{k}-w',\underline{k}-w)
\ge \underline{K}\,(w'-w)\,{\tilde h}'(\underline{K}\et^{w-w'})\,\Pi[\underline{k}-w',\underline{k}-w).
\end{align*}
Similarly, for any $z<\underline{k}-w'$, using ${\tilde h}'(S)=h'(\log S)/S$ and the convexity of ${\tilde h}$ we have
\begin{align*}
h(z+w')-h(z+w) &\ge {\tilde h}'(\et^{z+w})\et^z(\et^{w'}-\et^{w}) = h'(z+w)(\et^{w'-w}-1), \nn\\
h(z+w')-h(z+w) &\le {\tilde h}'(\et^{z+w'})\et^z(\et^{w'}-\et^{w}) = h'(z+w')(1-\et^{w-w'}).
\end{align*}
Taking into account the above inequalities, the fact that $h(\cdot)$ is decreasing, we can conclude from \eqref{feqdiff} that
\begin{align*}
\frac{f(w')-f(w)}{w'-w}
&\geq \underline{K}\,{\tilde h}'(\underline{K}\et^{w-w'})\,\Pi(\underline{k}-w',\underline{k}-w)+\frac{1-\et^{w'-w}}{w-w'}\int_{(-\infty,\underline{k}-w')}h'(z+w)\Pi(\diff z)\nn\\
&\geq \underline{K}\,{\tilde h}'(\underline{K}\et^{w-w'})\,\Pi(-\infty,-\delta)+\frac{\et^{w'-w}-1}{w'-w}\int_{(-\infty,\underline{k}-w')}h'(z+w)\Pi(\diff z),\nn\\
\frac{f(w')-f(w)}{w'-w}
&\le \frac{1-\et^{w-w'}}{w'-w}\int_{(-\infty,\underline{k}-w')}h'(z+w')\Pi(\diff z).
\end{align*}
By the dominated convergence theorem and the fact that ${\tilde h}'(\underline{K})=0$, we know that 
\[\lim_{w'\downarrow w}\frac{f(w')-f(w)}{w'-w}=\int_{(-\infty,\underline{k}-w)}h'(z+w)\Pi(\diff z).\]
Using similar arguments one can prove that $f(w)$ is left-differentiable at $w$. The continuity of $f'(\cdot)$ follows from dominated convergence theorem and the fact that $h'(\cdot)$ is uniformly bounded. 

Finally, the expressions in \eqref{eq4f} are a consequence of staightforward computations, using the definition of the Laplace exponent $\psi$ in \eqref{decomp} and $f$ in \eqref{eq:ffun}, and the expressions coming from these definitions:
\begin{align}
\int_{(-\infty,0)}(\et^{\prq z}-1-\prq z&\ind_{\{z>-1\}})\,\Pi(\diff z)=r+q-\mu\prq-\frac{1}{2}\sigma^2(\prq)^2,\nn\end{align}
\be
\int_{(-\infty,0)}(\et^z-1-z\ind_{\{z>-1\}})\,\Pi(\diff z)=\psi(1)-\mu-\frac{1}{2}\sigma^2,\nn\ee
\be
f(\underline{k}) = \int_{(-\infty,0)}\bigg(\frac{K}{\prq-1}\et^{\prq z}-\frac{\prq K}{\prq-1}\et^{z}+K\bigg)\Pi(\diff z).\nn
\ee
\end{proof}

\vspace{3pt}
\begin{proof}[Proof of Proposition \ref{prop:V2value}]
We focus on proving \eqref{V2V22Fab}, since the equality \eqref{V2V22Fa} is derived directly from the former and \eqref{uphit}.  
Notice that for $a=y$, using the definition given in \eqref{newW}, we know that $W^{(r,q)}(x,y)=W^{(r)}(x-y)$, therefore the result follows directly from  an application of \eqref{eq:loeffen}. Hence, we only need to prove the result for $a\in[\underline{k},y)$.
Let us first assume that $a>\underline{k}$. The result for $a=\underline{k}$ can then be obtained in the limit. 

We will use similar techniques as in \cite{OccupationInterval}. In particular, in the same probability space, for a given $n\ge 1$, we let $X^n$ be a spectrally negative L\'evy process with L\'evy triplet $(\mu,0,\Pi_n)$ where
\[\Pi_n(\diff z):=\ind_{\{z\le-\frac{1}{n}\}}\Pi(\diff z)+\sigma^2n^2\delta_{-\frac{1}{n}}(\diff z),\]
with $\delta_{-1/n}(\diff z)$ being the Dirac measure at $-1/n$. The process $X^n$ has paths of bounded variation, with the drift given by $\gamma_n:=\mu+\int_{(-1,-1/n]}|z|\Pi(\diff z)+\sigma^2n^2$. 
If $X$ has paths of unbounded variation,
we know that, for all sufficiently large $n\ge 1$, we will have $\gamma_n>0$. Without loss of generality, we will assume that $n\ge 1$ is sufficiently large and $a-\underline{k}>1/n$. 

Let us denote by $\mc{L}_n$ and $\psi_n(\cdot)$ the infinitesimal generator and the Laplace exponent of $X^n$, respectively. Moreover, introduce 
\[\mc{G}^n(x):=(\mc{L}_n-r-q\ind_{\{x<y\}})\underline{v}(x).\]
By the construction of $\Pi_n$, we know that, for any $x\ge a$, 
\[\mc{L}_n\underline{v}(x)=\psi_n(1)\et^x+\mc{L}_n h(x),\]
where 
\begin{align}
\psi_n(1)=&\mu+\sigma^2n^2\bigg(\et^{-\frac{1}{n}}-1+\frac{1}{n}\bigg)+\int_{(-\infty,-\frac{1}{n})}(\et^z-1-z\ind_{\{z>-1\}})\Pi(\diff z),\label{eq:psin}\\
\mc{L}_nh(x)=&\int_{(-\infty,-\frac{1}{n})}h(x+z)\Pi(\diff z)=\int_{(-\infty,\underline{k}-x)}h(x+z)\Pi(\diff z)=f(x)=\mc{L}h(x).\label{eq:Lnh}
\end{align}
where the equalities in \eqref{eq:Lnh} are due to the support of $h(\cdot)$. We notice that, $\mc{G}^n(x)\to(\mc{L}-r-q\ind_{\{x<y\}})\underline{v}(x)$ uniformly over $[a,b]$.

Let us denote by $W^{(r),n}(\cdot)$ the $r$-scale function of $X^n$, and let $T_{a}^{-,n}$, $T_b^{+,n}$ and $A_t^{y,n}$ be the first passage times of $a$ (from above), $b$ (from below), and the occupation time below $y$ for $X^n$, respectively. Define 
\[
V^n(x;y):=\ex_x\Big[\exp\big(-A_{T_{a}^{-,n}}^{y,n}\big) \underline{v}\big(X_{T_a^{-,n}}^n\big)\ind_{\{T_a^{-,n}<T_b^{+,n}\}}\Big] + \ex_x\Big[\exp\big(-A_{T_{b}^{+,n}}^{y,n}\big) \underline{v}\big(X_{T_b^{+,n}}^n\big)\ind_{\{T_b^{+,n}<T_a^{-,n}\}}\Big],\quad\forall x\in\R.
\]
We will study this function in four distinct intervals of $x$. We first observe that, for $x\in(-\infty,a)\cap(b,\infty)$, we have $V^n(x;y)=\underline{v}(x)$.

Now, for $x\in[a,y)$, we use the strong Markov property of $X^n$ and \eqref{uphit}, followed by an application of \eqref{eq:loeffen} for the test function $\underline{v}(\cdot)$, to obtain that 
\begin{multline}
V^n(x;y)=\underline{v}(x)+\frac{W^{(r+q),n}(x-a)}{W^{(r+q),n}(y-a)}\bigg(V^n(y;y)-\underline{v}(y)+\int_{(a,y)}W^{(r+q),n}(y-w)\mc{G}_n(w)\diff w\bigg)\\
-\int_{(a,x)}W^{(r+q),n}(x-w)\mc{G}_n(w)\diff w,\label{eq:V2nxy}
\end{multline}

Finally, for $x\in[y,b)$, by the strong Markov property of $X^n$, and using \eqref{eq:V2nxy}, we have
\begin{align}
&V^n(x;y)
=\ex_x\Big[\et^{-rT_y^{-,n}}V^n\big(X_{T_y^{-,n}}^n;y\big)\ind_{\{T_y^{-,n}<T_{b}^{+,n}\}}\Big] + \ex_x\Big[\et^{-rT_b^{+,n}}\underline{v}\big(X_{T_b^{+,n}}^n\big)\ind_{\{T_b^{+,n}<T_{y}^{-,n}\}}\Big]\nn\\
&=\ex_x\Big[\et^{-rT_y^{-,n}}\underline{v}\big(X_{T_y^{-,n}}\big)\ind_{\{T_y^{-,n}<T_b^{+,n},X_{T_y^{-,n}}<a\}}\Big] + \ex_x\Big[\et^{-rT_y^{-,n}}\underline{v}\big(X_{T_y^{-,n}}^n\big)\ind_{\{T_y^{-,n}<T_b^{+,n},X_{T_y^{-,n}}\ge a\}}\Big]\nn\\
&\;\;+\frac{\ex_x\Big[\et^{-rT_y^{-,n}}W^{(r+q),n}\big(X_{T_y^{-,n}}^n-a\big)\ind_{\{T_y^{-,n}<T_{b}^{+,n}, X_{T_y^{-,n}}\ge a\}}\Big]}{W^{(r+q),n}(y-a)}\bigg(V^n(y;y)-\underline{v}(y)+\int_{(a,y)}W^{(r+q),n}(y-w)\mc{G}_n(w)\diff w\bigg)\nn\\
&\;\;-\int_{(a,y)}\ex_x\Big[\et^{-rT_y^{-,n}}W^{(r+q),n}\big(X_{T_y^{-,n}}^n-w\big)\ind_{\{T_y^{-,n}<T_{b}^{+,n}, X_{T_y^{-,n}}\ge w\}}\Big]\mc{G}_n(w)\diff w+\frac{W^{(r),n}(x-y)}{W^{(r),n}(b-y)}\underline{v}(b),\label{eq:V2nxy1}
\end{align}
where we used Fubini theorem in the last line. Observe that, by another application of \eqref{eq:loeffen}, the first line of the right hand side of \eqref{eq:V2nxy1} is given by 
\begin{multline*}
\ex_x\Big[\et^{-rT_y^{-,n}}\underline{v}\big(X_{T_y^{-,n}}^n\big)\ind_{\{T_y^{-,n}<T_{b}^{+,n}\}}\Big]\\
=\underline{v}(x)-\frac{W^{(r),n}(x-y)}{W^{(r),n}(b-y)}\underline{v}(b)
+\int_{[y,b)}\bigg(\frac{W^{(r),n}(x-y)W^{(r),n}(b-w)}{W^{(r),n}(b-y)}-W^{(r),n}(x-w)\bigg)\mc{G}_n(w)\diff w.\end{multline*}
Moreover, using  \cite[Lemma \ref{lem:equivalent} and (19)]{OccupationInterval}, and the fact that $W^{(r+q),n}$ vanishes on $(-\infty,0)$, we have for any fixed $u<y$ that
\begin{align}
\ex_x\Big[\et^{-rT_y^{-,n}}W^{(r+q),n}\big(X_{T_y^{-,n}}^n-u\big)\ind_{\{T_y^{-,n}<T_{b}^{+,n}, X_{T_y^{-,n}}\ge u\}}\Big]=&\ex_x\Big[\et^{-rT_y^{-,n}}W^{(r+q),n}\big(X_{T_y^{-,n}}^n-u\big)\ind_{\{T_y^{-,n}<T_{b}^{+,n}\}}\Big]\nn\\
=&W^{(r,q),n}(x,u)-\frac{W^{(r),n}(x-y)}{W^{(r),n}(b-y)}W^{(r,q),n}(b,u),\nn\end{align}
where $W^{(r,q),n}(x,u)$ is defined in view of \eqref{eq:bigW1} as
\be
W^{(r,q),n}(x,u):=W^{(r+q),n}(x-u)-q\int_{(y,x\vee y)}W^{(r),n}(x-z)W^{(r+q),n}(z-u)\diff z.\nn
\ee
Therefore, we conclude from the above analysis and the fact that $W^{(r),n}(x,u)$ vanishes on $(-\infty,0)$ that \eqref{eq:V2nxy1} becomes
\begin{align}
V^n(x;y)=&\underline{v}(x)
+\int_{[y,b)}\bigg(\frac{W^{(r),n}(x-y)W^{(r),n}(b-w)}{W^{(r),n}(b-y)}-W^{(r),n}(x-w)\bigg)\mc{G}_n(w)\diff w\nn\\
&+\frac{W^{(r,q),n}(x,a)-\frac{W^{(r),n}(x-y)}{W^{(r),n}(b-y)}W^{(r,q),n}(b,a)}{W^{(r+q),n}(y-a)}\bigg(V^n(y;y)-\underline{v}(y)+\int_{(a,y)}W^{(r+q),n}(y-w)\mc{G}_n(w)\diff w\bigg)\nn\\
&-\int_{(a,y)}\bigg(W^{(r,q),n}(x,w)-\frac{W^{(r),n}(x-y)}{W^{(r),n}(b-y)}W^{(r,q),n}(b,w)\bigg)\mc{G}_n(w)\diff w,\label{eq:key}
\end{align}
for all $x\in(a,b)$. Now, letting $x=y$ and using the facts that $W^{(r,q),n}(y,w)=W^{(r+q),n}(y-w)$ for all $w\in[a,y]$ and  $W^{(r),n}(0)\neq 0$ since $X^n$ has bounded variation, we obtain that 
\begin{align*}
V^n(y;y)=&\underline{v}(y)+\int_{(a,y)}\bigg(\frac{W^{(r+q),n}(y-a)}{W^{(r,q),n}(b,a)}W^{(r,q),n}(b,w)-W^{(r+q),n}(y-w)\bigg)\mc{G}_n(w)\diff w\\
&+\frac{W^{(r+q),n}(y-a)}{W^{(r,q),n}(b,a)}\int_{[y,b)}W^{(r),n}(b-w)\mc{G}_n(w)\diff w.
\end{align*}
Plugging the above expression of $V^n(y;y)$ into \eqref{eq:key}, we obtain that 
\begin{align*}
V^n(x;y)=&\underline{v}(x)+\int_{(a,y)}\bigg(\frac{W^{(r,q),n}(x,a)}{W^{(r,q),n}(b,a)}W^{(r,q),n}(b,w)-W^{(r,q),n}(x,w)\bigg)\mc{G}_n(w)\diff w\\
&+\int_{[y,b)}\bigg(\frac{W^{(r,q),n}(x,a)}{W^{(r,q),n}(b,a)}W^{(r),n}(b-w)-W^{(r),n}(x-w)\bigg)\mc{G}_n(w)\diff w.
\end{align*}
As $X^n$ converges to $X$ uniformly on compact time intervals $\pr_x$-a.s. (see  \cite[p. 210]{Bertoin1996}), and $\underline{v}(\cdot)$ is bounded over $(-\infty,b]$, we may use the dominated convergence theorem and L\'evy's extended continuity theorem (see e.g.  \cite[Theorem 5.22]{Kall02Book}) as in the proof of  \cite[Theorem 1]{OccupationInterval} to complete the proof. If $X$ has paths of bounded variation, then we obtain the final result without taking the limit. 
\end{proof}

\vspace{3pt}
\begin{proof}[Proof of Proposition \ref{prop:sfit}]
For ease of notation, we let $a^*\equiv a^\star(y)$ and $b^\star\equiv b^\star(y)$. 

\ul{\it Part $(i)$:} We prove the desired claim by exploiting the optimality of $b$ among all up-crossing thresholds larger than $y$ (recall that the optimal threshold $b^*\ge y_m>y$). More specifically, for each pair $(x,b)$ such that $x\in(y,b^\star)$ and $b>x$, we consider the mapping $b\mapsto V(x;y,a^*,b)$, where $V(\cdot;\cdot,\cdot,\cdot)$ is defined in \eqref{V2V22}. Then we know that $V(x;y,a^*,b^*)=\sup_{b\ge x}V(x;y,a^*,b)$ for every fixed $x\in(y,b^\star)$. 
On one hand, by Proposition \ref{prop:V2value}, we have 
\begin{multline}
V(x;y,a^*,b)=\underline{v}(x)+\int_{(a^*,y)}\bigg(\frac{W^{(r,q)}(x,a^*)}{W^{(r,q)}(b,a^*)}W^{(r,q)}(b,w)-W^{(r,q)}(x,w)\bigg)\left[\chi(w)-q\underline{v}(w)\right]\diff w\nn\\
+\int_{[y,b)}\bigg(\frac{W^{(r,q)}(x,a^*)}{W^{(r,q)}(b,a^*)}W^{(r)}(b-w)-W^{(r)}(x-w)\bigg)\chi(w)\diff w
,\nn
\end{multline}  
for all $x\in(y,b^\star)$ and $b\ge x$. As discussed above, for all fixed $x\in(y,b^\star)$, the mapping $b\mapsto V(x;y,a^*,b)$ is maximized at $b^\star$, a fact that will be exploited in the subsequent analysis.
This function is continuously differentiable in $b$ ($x$, resp.).\footnote{The only nontrivial part is the differentiability of $b\mapsto\int_{[y,b)}W^{(r)}(b-w)\chi(w)\diff w=\int_{(0,b-y]}W^{(r)}(w)\,\chi(b-w)\diff w$, which follows from that of $\chi(\cdot)$.}
Hence, $b^\star$ satisfies the first order condition:
\[\p_b V(x;y,a^\star,b)|_{b=b^\star}=0,\quad\forall x\in(y,b^\star).\]
Using the above condition and some straightforward calculations, for all $x\in(y,b^\star)$, we get 
\begin{multline}
0=-\frac{W_1^{(r,q)}(b^\star,a^\star)\,W^{(r,q)}(x,a^\star)}{\left(W^{(r,q)}(b^\star,a^\star)\right)^2}\bigg(\int_{(a^\star,y)}W^{(r,q)}(b,w)\,[\chi(w)-q\underline{v}(w)]\diff w+\int_{[y,b^\star)}W^{(r)}(b^\star-w)\chi(w)\diff w\bigg)\nn\\
+\frac{W^{(r,q)}(x,a^\star)}{W^{(r,q)}(b^\star,a^\star)}\bigg(\int_{(a^\star,y)}W_1^{(r,q)}(b,w)[\chi(w)-q\underline{v}(w)]\diff w+\int_{[y,b^\star)}W^{(r)}(b^\star-w)\chi'(w)\diff w+W^{(r)}(b^\star-y)\chi(y)\bigg),\nn
\end{multline}
where $W_1^{(r,q)}(b,a)=\p_bW^{(r,q)}(b,a)$. 
Since the factor $W^{(r,q)}(x,a^\star)/W^{(r,q)}(b^\star,a^\star)>0$, we know that 
\begin{align}
0=&-\frac{W_1^{(r,q)}(b^\star,a^\star)}{W^{(r,q)}(b^\star,a^\star)}\bigg(\int_{(a^\star,y)}W^{(r,q)}(b,w)\,[\chi(w)-q\underline{v}(w)]\diff w+\int_{[y,b^\star)}W^{(r)}(b^\star-w)\chi(w)\diff w\bigg)\nn\\
&+\int_{(a^\star,y)}W_1^{(r,q)}(b,w)[\chi(w)-q\underline{v}(w)]\diff w+\int_{[y,b^\star)}W^{(r)}(b^\star-w)\chi'(w)\diff w+W^{(r)}(b^\star-y)\chi(y).\label{eq:spast1}
\end{align}
On the other hand, for all $x\in(y,b^\star)$, we have
\begin{align}
\p_xv(x;y)&=\p_x V(x;y,a^\star,b^\star)\nn\\
&=\et^x+\frac{W_1^{(r,q)}(x,a^\star)}{W^{(r,q)}(b^\star,a^\star)}\bigg(\int_{(a^\star,y)}W^{(r,q)}(b^\star,w)[\chi(w)-q\underline{v}(w)]\diff w+\int_{[y,b^\star)}W^{(r)}(b^\star-w)\chi(w)\diff w\bigg) \nn\\
&\quad-\int_{(a^\star,y)}W_1^{(r,q)}(x,w)[\chi(w)-q\underline{v}(w)]\diff w-\int_{[y,x)}W^{(r)}(x-w)\chi'(w)\diff w-W^{(r)}(x-y)\chi(y). \label{eq:spast2}
\end{align}
Notice that the domain for the last integral in  \eqref{eq:spast2} can be replaced by $[y,b^\star)$ without affecting its value. Combining \eqref{eq:spast1} and \eqref{eq:spast2}, we know that $\p_xV(x;y)\to \et^{b^\star}$ as $x\to b^\star$.

\ul{\it Part $(ii)$:}
Let us consider $V(y;y,a,y_m)$, the value of the two-sided strategy $T_a^-\wedge T_{y_m}^+$ for $a<y$ when starting from $x=y$. 
Since $X$ is assumed to have unbounded variation, we know that $V(y;y,y-,y_m)=\et^y-K$. 
Moreover, straightforward calculations using \eqref{V2V22} show that the mapping $a\mapsto V(y;y,a,y_m)$ is continuously differentiable over $(\underline{k},y)$. 
Indeed, we have
\begin{align}
\p_aV(y;y,a,y_m)=&f_1(a;y,y_m)W^{(r+q)\prime}(y-a)+f_2(a;y,y_m)W^{(r+q)}(y-a).\nn
\end{align}
where 
\begin{align}
f_1(a;y,y_m)&:=-\int_{(a,y_m)}\frac{W^{(r,q)}(y_m,w)}{W^{(r,q)}(y_m,a)}[\chi(w)-q\ind_{\{w<y\}}\underline{v}(w)]\diff w,\nn
\\
f_2(a;y,y_m)&:=-\frac{\p_aW^{(r,q)}(y_m,a)}{W^{(r,q)}(y_m,a)^2}\int_{(a,y_m)}W^{(r,q)}(y_m,w)[\chi(w)-q\ind_{\{w<y\}}\underline{v}(w)]\diff w\nn.
\end{align}
It is easily seen that $f_2(y-;y,y_m)$ is finite, hence by $W^{(r+q)}(0)=0$ (since $X$ has unbounded variation), we have
\[\lim_{a\uparrow y}f_2(a;y,y_m)W^{(r+q)}(y-a)=0.\]
On the other hand, using 
\eqref{newW} we have
\begin{align}
\lim_{a\uparrow y}f_1(a;y,y_m)=&-\frac{1}{W^{(r)}(y_m-y)}\bigg(\int_{(0,y_m-y)}\chi(y_m-z)W^{(r)}(z)\diff z\bigg)\nn.
\end{align}
Recall from 
\eqref{ym} that $\chi(x)>0$ if and only if $x<y_m$, therefore 
$f_1(y-;y,y_m)\in(-\infty,0)$. 
Moreover, recall from Lemma \ref{lem W} that 
$W^{(r+q)\prime}(y-a)>0$ for all $a<y$ and it converges to either $2/\sigma^2$ or $\infty$ as $a\uparrow y$, depending on whether $\sigma>0$ or not. Consequently, we have $\p_a V(y;y,a,y_m)|_{a=y-}<0$, which implies that $v(y;y)\ge V(y;y,y-\e,y_m)>V(y;y,y,y_m)=\et^y-K$ for all sufficiently small $\e>0$, where $v$ is given by \eqref{eqeqeqeqeq} and the first inequality follows from the fact that $V(y;y,a,y_m)$ is clearly suboptimal for all $a<y$, thus $V(y;y,a,y_m)\leq v(y;y)$. This proves that $y\notin\mc{S}^y$.

We now prove that smooth fit holds at $a^\star(y)$:

 {\bf Step 1: }We first treat the case when $a^\star(y)=\underline{k}$. To this end, notice that for any $\e\in(0,y-a^\star(y))$, we have $\et^{\,\underline{k}+\e}-K\le v(\underline{k}+\e;y)\le v(\underline{k}+\e;\tilde{y})=(\et^{\,\underline{k}+\e}-K)/R(\underline{k}+\e;\tilde{y})$ (due to Proposition \ref{prop:compare1} and the definition of $R(\cdot;\cdot)$ in \eqref{eq:ratio}). Hence, if $X$ has unbounded variation,  
\begin{align}
\et^{\underline{k}}=\lim_{\e\downarrow 0}\frac{(\et^{\underline{k}+\e}-K)-(\et^{\underline{k}}-K)}{\e}\le&\lim_{\e\downarrow0}\frac{v(\underline{k}+\e;y)-v(\underline{k};y)}{\e}\le\lim_{\e\downarrow0}\frac{(\et^{\underline{k}+\e}-K)/R(\underline{k}+\e;\tilde{y})-(\et^{\underline{k}}-K)}{\e}=\et^{\underline{k}},\nn\end{align}
where we used the facts that $\p_xR(x,y)=\overline{\Lambda}(x-y)(K-\et^yg(x;y))/v(x;y)$ and that $x=\underline{k}$ is a root of $\et^yg(x;y)=K$ for every $y\ge\underline{k}$. 

{\bf Step 2:} We treat the case when $a^\star(y)\in(\underline{k},y)$. This is the only case left, thanks to the proven fact that $y\notin\mc{S}^y$. Hence, without loss of generality, we assume that, there is a $\delta>0$ sufficiently small, such that $a^\star\equiv a^\star(y)\in(\underline{k}+\delta,y-\delta)$. In the proof below we use similar arguments to those of Proposition \ref{prop:sfit}. To be more precise, we consider the mapping
\begin{align}
(x,a)\mapsto {\bar U}(x,a;y):=&\ex_x\big[\et^{-(r+q)T_y^+}\ind_{\{T_y^+<T_a^-\}}\big] v(y;y) + \ex_x\big[\et^{-(r+q)T_a^-}\underline{v}(X_{T_a^-})\ind_{\{T_a^-<T_y^+\}}\big]\nn
\end{align}
for all $x\in(a^\star,y)$ and $a\in(\underline{k}-\delta, x]$. 
By using \eqref{up} and \eqref{eq:loeffen} (with strictly increasing, continuously differentiable testing function $\underline{v}(\cdot)$), we obtain that (see also 
\eqref{eq:57})
\begin{multline}
{\bar U}(x,a;y)=
\underline{v}(x)+\frac{W^{(r+q)}(x-a)}{W^{(r+q)}(y-a)}\left(v(y;y)-\underline{v}(y)\right)\\
+\int_{(a,y)}\bigg(\frac{W^{(r+q)}(x-a)W^{(r+q)}(y-w)}{W^{(r+q)}(y-a)}-W^{(r+q)}(x-w)\bigg)\,\left(\chi(w)-q\,\underline{v}(w)\right)\,\diff w.\nn
\end{multline}
The function ${\bar U}(x,a;y)$ is continuously differentiable in $a$ ($x$, resp.) for all $a\in(\underline{k}+\delta, x)$ ($x\in(a,y)$, resp.). 
Indeed, the first line of ${\bar U}(x,a;y)$ is obviously continuously differentiable in $a$ ($x$, resp.). 
Moreover, we know 
from \eqref{eq:57}  and \eqref{eq:underlinev} that $\chi(\cdot)-q\,\underline{v}(\cdot)$ is decreasing and hence negative continuous function that is uniformly bounded over $[\underline{k},y)$,
hence 
the integral part is also continuously differentiable in $a$ ($x$, resp.). 
Thus, for all $x>a^\star$, we have by the optimality of $a^\star$, that the following first order condition for ${\bar U}(x,a;y)$ holds at $a=a^\star$:  
\begin{multline}
0=\p_a{\bar U}(x,a;y)|_{a=a^\star}=\frac{W^{(r+q)\prime}(y-a^\star)W^{(r+q)\prime}(x-a^\star)}{\left(W^{(r+q)}(y-a^\star)\right)^2}\bigg(\frac{W^{(r+q)}(x-a^\star)}{W^{(r+q),\prime}(x-a^\star)}-\frac{W^{(r+q)}(y-a^\star)}{W^{(r+q)\prime}(y-a^\star)}\bigg)\\
\times\bigg(v(y;y)-\underline{v}(y)+\int_{(a^\star,y)}W^{(r+q)}(y-w)\left(\chi(w)-q\,\underline{v}(w)\right)\diff w\bigg).\label{eq:kmn1}
\end{multline}
Because $W^{(r+q),\prime}(x)>0$ for all $x>0$, the pre-factor of the above expression is positive. Moreover, due to the fact that the mapping $x\mapsto W^{(r+q),\prime}(x)/W^{(r+q)}(x)$ is strictly decreasing over $(0,\infty)$ (see Lemma \ref{lem W}), we know that 
\[\frac{W^{(r+q)}(x-a^\star)}{W^{(r+q),\prime}(x-a^\star)}-\frac{W^{(r+q)}(y-a^\star)}{W^{(r+q)\prime}(y-a^\star)}<0,\]
since $0<x-a^\star<y-a^\star$. Hence, necessarily, we have 
\be
v(y;y)-\underline{v}(y)+\int_{(a^\star,y)}W^{(r+q)}(y-w)\left(\chi(w)-q\,\underline{v}(w)\right)\diff w=0.\label{fiteq1}
\ee
On the other hand,  we have $v(x;y)\equiv {\bar U}(x,a^\star;y)$ for $x\in(a^\star,y)\subset(\underline{k},\infty)$, therefore $\underline{v}'(x)=\et^x$ and using \eqref{fiteq1} we obtain that
\begin{align}
\p_x v(x;y)=\p_x{\bar U}(x,a^\star;y)=\et^x-\int_{(a^*,x)}W^{(r+q),\prime}(x-w)\left(\chi(w)-q\underline{v}(w)\right)\diff w.
\label{eq:kmn2}
\end{align}
From \eqref{fiteq1} and \eqref{eq:kmn2} we have for all $x\in(a^\star,y)$ that
\begin{align}
\big|\p_xv(x,y)-\et^{a^\star}\big|
\le&\;\big|\et^x-\et^{a^\star}\big|+ W^{(r+q)}(x-a^\star)\cdot\sup_{w\in[\underline{k}+\delta,y]}\big|\chi(w)-q\underline{v}(w)\big|,\nn
\end{align}
where we have
used the absolute bound for $\chi(w)-q\underline{v}(w)$ and the fact that $W^{(r+q)}(0)=0$ in the case of unbounded variation. 
As a consequence, we see that $|\p_xv(x,y)-\et^{a^\star}|$ converges to 0 as $x\downarrow a^\star$ and this completes the proof.
%
%
\end{proof}

\section{Preliminaries on scale functions}\label{sec:pre}

In this appendix, we briefly review a collection of useful results on spectrally negative L\'evy processes and their scale functions. 

The $r$-scale function $W^{(r)}(\cdot)$ is closely related to exit problems of the spectrally negative L\'{e}vy process $X$ with respect to first passage times of the form \eqref{firstpass}.
A well-known fluctuation identity of spectrally negative L\'{e}vy
processes (see e.g. \cite[Theorem 8.1]{Kyprianou2006}) is given, for
$r\geq0$ and $x\in[a, b]$, by 
\begin{align}
\mathbb{E}_{x}[  \et^{-rT_{b}^{+}} \ind_{\{  T_{b}^{+}<T_{a}^{-}\}
}]  =\frac{W^{(r)}(x-a)}{W^{(r)}(b-a)} \,. \label{up} 
\end{align}
Moreover, letting  
$F(\cdot)$ be a positive, non-decreasing, continuously differentiable function on $\R$, 
and further supposing that $F(\cdot)$ has an absolutely continuous derivative with a bounded density over $(-\infty,b]$ 
if $X$ has paths of unbounded variation,
it is known from \cite[Theorem 2]{loeffen_outshoot} that, for any fixed $a,b$ such that $-\infty<a<b<\infty$, we have
 \begin{multline}
\ex_x[\et^{-r T_a^-}F(X_{T_a^-})\ind_{\{T_a^-<T_b^+\}}]\\
=F(x)-\frac{W^{(r)}(x-a)}{W^{(r)}(b-a)}F(b)+\int_{(a,b)}\bigg(\frac{W^{(r)}(x-a)}{W^{(r)}(b-a)}W^{(r)}(b-w)-W^{(r)}(x-w)\bigg)\cdot(\mc{L}-r)F(w)\diff w.\label{eq:loeffen}
\end{multline}

The following lemma gives the behavior of scale functions at $0+$ and $\infty
$; see, e.g.,   
\cite[Lemmas 3.1, 3.2, 3.3]{Kuznetsov_2011}, and 
\cite[(3.13)]{Leung_Yamazaki_2011}.

\begin{lem} \label{lem W}
For any $r>0$,
\begin{align*}
W^{(r)}(0)    =\left\{
\begin{array}
[c]{ll}%
0, & \text{unbounded variation},\\
\frac{1}{\gamma}, & \text{bounded variation},%
\end{array}
\right. 
W^{(r)^{\prime}}(0+)    =\left\{
\begin{array}
[c]{ll}%
\frac{2}{\sigma^{2}}, & \text{if }\sigma>0,\\
\infty, & \text{if }\sigma=0\text{ and }\Pi(-\infty,0)=\infty,\\
\frac{r+\Pi(-\infty,0)}{\gamma^{2}}, & \text{if }\sigma=0\text{ and }\Pi
(-\infty,0)<\infty,
\end{array}
\right.
\end{align*}
Moreover, $\et^{-\Phi(r) x}W^{(r)}(x)$ converges to $\frac{1}{\psi'(\Phi(r))}\in(0,\infty)$ as $x\to\infty$, 
The function $x\mapsto\frac{W^{(r)\prime}(x)}{W^{(r)}(x)}$ is strictly decreasing, and converges to
$\Phi(r)$ as $x\to\infty$.
\end{lem}

\def\cprime{$'$}

\end{document}